\documentclass[times,onecolumn,12p]{elsarticle}

\journal{Information and Computation}

\usepackage{url}
\usepackage{float}
\usepackage[T1]{fontenc}
\usepackage{amssymb}
\usepackage[cmex10]{amsmath}
\usepackage{amssymb,amsthm,alltt}
\usepackage{amscd,amsfonts,latexsym,xspace}
\usepackage{ifthen}
\usepackage{flowchart}                                                                                                                           
\usepackage{hyperref}
\usepackage{tikz}
\usetikzlibrary{arrows, arrows.meta, decorations.markings, calc, petri}
\usetikzlibrary{positioning, decorations.pathmorphing, decorations.pathreplacing}
\usetikzlibrary{positioning}
\usepackage{tikzpagenodes}
\usetikzlibrary{decorations.pathmorphing}
\usetikzlibrary{shadows.blur}
\usetikzlibrary{decorations.pathreplacing}
\usepackage{wrapfig}
\usepackage{todonotes}

\newtheorem{theorem}{Theorem}
\newtheorem{conjecture}{Conjecture}
\newtheorem{corollary}{Corollary}
\newtheorem{example}{Example}
\newtheorem{remark}{Remark}
\newtheorem{claim}{Claim}
\newtheorem{lemma}{Lemma}

\newcommand{\scaleOfPictures}{0.8}

\newcommand{\submultiseteq}{\sqsubseteq}
\newcommand{\multisets}[1]{{\cal M}(#1)}

\newcommand{\trans}{\longrightarrow}

\newcommand{\atoms}{{\mathbb A}}
\newcommand{\para}[1]{\vspace{-2mm}\paragraph{\bf #1}}
\newcommand{\paraemph}[1]{\paragraph{#1}}

\newcommand{\Q}{{\mathbb Q}}
\newcommand{\N}{{\mathbb N}}

\newcommand{\aut}[1]{\text{Aut}(#1)}

\newcommand{\nat}{{\mathbb N}}
\newcommand{\setof}[2]{\left\{\, #1 \, | \, #2\,\right\}}

\newcommand{\set}[1]{\{#1\}}

\newcommand{\places}{{\cal P}}
\newcommand{\transitions}{{\cal T}}
\newcommand{\masz}{{\cal M}}

\newcommand{\age}[1]{\text{\sc Age}(#1)}

\newcommand{\embedsin}{\unlhd}
\newcommand{\embedsinpar}[1]{\embedsin_{#1}}
\newcommand{\str}[1]{{\cal #1}}

\newcommand{\wqo}{{\sc{wqo}}\xspace}  
\newcommand{\Wqo}{{\sc{Wqo}}\xspace}
\newcommand{\A}{\atoms}

\newcommand{\G}{{\mathbb G}}

\newcommand{\Aeq}{\A_{=}}
\newcommand{\Aneq}{\A_{1}}
\newcommand{\Ato}{\A_{\leq}}

\pgfdeclarelayer{prebackground}
\pgfdeclarelayer{background}
\pgfdeclarelayer{foreground}
\pgfsetlayers{prebackground,background,main,foreground}

\newcommand{\lEdge}[4][cGray]{
  \ifthenelse{\equal{#3}{}}{
    \draw[Edge, #1] 
      (#2.center) to
      (#4.center);
  }{
    \draw[Edge, #1] 
      (#2.center) to 
        node[EdgeLabel] {{\strut\small#3}} 
      (#4.center);
  }
}
\newcommand{\lsEdge}[5][cGray]{
  \ifthenelse{\equal{#3}{}}{
    \draw[Edge, #1] 
      (#2.center) to
      (#4.center);
  }{
    \path[Edge, #1] 
      (#2.center) to 
        node[EdgeLabel, #5] {{\strut\small#3}} 
      (#4.center);
  }
}

\definecolor{cBlack}{RGB}{0,0,0}

\definecolor{cA}{RGB}{86,122,217}
\definecolor{cB}{RGB}{79,171,58}
\definecolor{cC}{RGB}{223,64,102}

\definecolor{cX}{RGB}{227,172,42}
\definecolor{cY}{RGB}{74,198,224}
\definecolor{cZ}{RGB}{224,74,157}
\definecolor{cZZ}{RGB}{23,167,158}

\definecolor{cGray}{RGB}{90,90,90}
\definecolor{cLightGray}{RGB}{180,180,180}
\definecolor{cVeryLightGray}{RGB}{220,220,220}

\tikzstyle{PNPlace} = [
  anchor=center,
  fill=white,
  circle,
  inner sep=0,
  draw=black,
  minimum size=0.6cm
]

\tikzstyle{PNBigPlace} = [
  PNPlace,
  minimum size=1cm,
  outer sep=0.05cm,
]

\tikzstyle{PNTransition} = [
  anchor=center,
  fill=black,
  rectangle,
  inner xsep=0.06cm,
  inner ysep=0.3cm,
  outer sep=0.05cm,
]

\tikzstyle{PNToken} = [
  anchor=center,
  fill=gray,
  circle,
  inner sep=0,
  minimum size=0.3cm,
  text=white,
]

\tikzstyle{PNDataToken} = [
  anchor=center,
  fill=black,
  text=white,
  circle,
  inner sep=0,
  minimum size=0.35cm,
  font={\fontsize{10pt}{12}\selectfont}
]

\tikzstyle{PNArrow} = [
  arrows={-Triangle[angle=50 :3pt  5]},
]

\tikzstyle{Basic} = [
    anchor=center,
    inner sep=0,
    minimum size=0
]
\tikzstyle{Node} = [
  anchor=center,
  fill=cGray,
  circle,
  text=white,
  inner sep=0,
  minimum size=0.25cm,
  draw=white,
  thick
]
\tikzstyle{SNode} = [
  anchor=north west,
  fill=white,
  rectangle,
  text=black,
  rounded corners=.125cm,
  inner sep=0.1cm,
  minimum size=0.25cm,
  draw=black
]

\tikzstyle{ANode} = [
    Node,
    fill=white,
    draw=cGray
]

\tikzstyle{LemmaNode} = [
  anchor=center,
  fill=cLemmaBox, 
  inner sep=0.3cm,
  blur shadow={
    shadow blur steps=5,
    shadow blur extra rounding=1.3pt,
    shadow xshift=0, 
    shadow yshift=-1
  }
]

\tikzstyle{Edge} = [
  draw=black, 
  line width=0.1cm,
  text=black
]

\tikzstyle{ThinEdge} = [
  Edge,
  line width=0.05cm
]

\tikzstyle{AEdge} = [
    Edge,
    draw=cLightGray,
    dashed
]

\tikzstyle{EdgeLabel} = [
  circle,
  line width=0.12cm,
  midway,
  auto,
  text height=1.1ex,
  text depth=0.2ex,
  text width=1em,
  align=center,
  fill=white,
  inner ysep=0.0cm,
  inner xsep=0.0cm,
  minimum size=0,
  opacity=1,
  text opacity=1
]
\tikzstyle{Selection} = [
  semithick,
  double distance = 0.5cm,
  line cap=round,
  line join=round,
  draw=cGray
]
\tikzstyle{doubleEdge} = [
  thin,
  double distance = 0.1cm,
  line cap=round,
  line join=round,
  draw=white
]
\tikzstyle{BgFrame} = [
  inner sep=0
]
\tikzstyle{FlowNode} = [
  draw=cA
]
\tikzstyle{FlowArrow} = [
arrows={-Triangle[angle=50 :0pt  3]},
line width=0.3cm,
draw=cVeryLightGray
]
\tikzstyle{FlowArrow2} = [
arrows={-Triangle[angle=50 :0pt  3]},
line width=0.1cm,
draw=cVeryLightGray
]
\tikzstyle{CaseArrow} = [
  arrows={-Triangle[angle=50 :0pt  3]},
  line width=0.1cm,
  draw=cGray
]
\tikzstyle{FArrow} = [
  arrows={-Triangle[angle=50 :0pt  5]},
  line width=0.05cm,
  draw=cBlack
]
\tikzstyle{FLine2} = [
  line width=0.05cm,
  draw=cBlack
]
\tikzstyle{FArrow2} = [
  FLine2,
  arrows={-Triangle[angle=50 :0pt  4]}
]
\tikzstyle{FArrow2inv} = [
  FLine2,
  arrows={Triangle[angle=50 :0pt  4]-}
]
\tikzstyle{AInstance} = [
    draw,
    inner sep=0.3cm,
    double
]

\tikzstyle{AResult} = [
    draw,
    inner sep=0.3cm
]
\newcommand{\Color}[2]{\ensuremath{\mathbf{\color{#1}{#2}}}}

\renewcommand{\a}{\Color{cA}{a}}
\renewcommand{\b}{\Color{cB}{b}}
\renewcommand{\c}{\Color{cC}{c}}

\newcommand{\x}{\Color{cX}{x}}
\newcommand{\y}{\Color{cY}{y}}
\newcommand{\z}{\Color{cZ}{z}}

\renewcommand{\v}{\Color{cBlack}{v}}
\renewcommand{\u}{\Color{cBlack}{u}}
\newcommand{\w}{\Color{cBlack}{w}}
\newcommand{\n}{\Color{cBlack}{?}}

\newcommand{\Colors}{\pgfutilensuremath{\mathit{Colors}}}
\newcommand{\ColorsTwo}{\pgfutilensuremath{\mathit{Colors}'}}

\newcommand{\Instance}[1]{\medskip\noindent\emph{Instance #1}.}

\newsavebox{\CaseOkBox}
\sbox{\CaseOkBox}{
\begin{tikzpicture}[scale=0.1*\scaleOfPictures]
    \path[fill=cB,rotate around={45:(0,0)}] (0,0) rectangle (3,1);
    \path[fill=cB,rotate around={45:(0,0)}] (0,0) rectangle (1,2);
\end{tikzpicture}
}
\newcommand{\CaseOk}{\usebox{\CaseOkBox}}

\newsavebox{\CaseEdgeExistsBox}
\sbox{\CaseEdgeExistsBox}{
\begin{tikzpicture}[scale=0.3*\scaleOfPictures]
    \node[ANode, scale=0.8] at (0, 0) (n1) {};
    \node[ANode, scale=0.8] at (1, 1) (n2) {};
    \draw[Edge, scale=0.8, line width=0.05cm] (n1) to[bend right=40] node[midway, auto] {\scriptsize$\exists$} (n2);
\end{tikzpicture}
}
\newcommand{\CaseEdgeExists}{\usebox{\CaseEdgeExistsBox}}

\newsavebox{\CaseAEdgeBox}
\sbox{\CaseAEdgeBox}{
\begin{tikzpicture}[scale=0.3*\scaleOfPictures]
    \node[ANode, scale=0.8] at (0, 0) (n1) {};
    \node[ANode, scale=0.8] at (1, 1) (n2) {};
    \draw[Edge, cA, scale=0.8, line width=0.05cm] (n1) to[bend right=40] node[midway, auto] {\a} (n2);
\end{tikzpicture}
}
\newcommand{\CaseAEdge}{\usebox{\CaseAEdgeBox}}

\newsavebox{\CaseBEdgeBox}
\sbox{\CaseBEdgeBox}{
\begin{tikzpicture}[scale=0.3*\scaleOfPictures]
    \node[ANode, scale=0.8] at (0, 0) (n1) {};
    \node[ANode, scale=0.8] at (1, 1) (n2) {};
    \draw[Edge, cB, scale=0.8, line width=0.05cm] (n1) to[bend right=40] node[midway, auto] {\b} (n2);
\end{tikzpicture}
}
\newcommand{\CaseBEdge}{\usebox{\CaseBEdgeBox}}

\newsavebox{\CaseEdgeDoesNotExistBox}
\sbox{\CaseEdgeDoesNotExistBox}{
\begin{tikzpicture}[scale=0.3*\scaleOfPictures]
    \node[ANode, scale=0.8] at (0, 0) (n1) {};
    \node[ANode, scale=0.8] at (1, 1) (n2) {};
    \draw[AEdge, scale=0.8, line width=0.05cm] (n1) to[bend right=40] node[midway, auto] {\scriptsize$\neg\exists$} (n2);
\end{tikzpicture}
}
\newcommand{\CaseEdgeDoesNotExist}{\usebox{\CaseEdgeDoesNotExistBox}}

\newcommand{\threeTokensLabeled}[6]{%
\begin{tikzpicture}[baseline=-0.35cm]
    \foreach \i/\l/\c in {1/#1/#2,2/#3/#4,3/#5/#6} {
        \node[PNToken,fill=\c] at (90+\i*120:0.22cm) {\textsf{\l}};
    }
\end{tikzpicture}}

\newcommand{\twoTokensLabeled}[4]{%
\begin{tikzpicture}
    \foreach \i/\l/\c in {1/#1/#2,2/#3/#4} {
        \node[PNToken,fill=\c] at (0+\i*180:0.19cm) {\textsf{\l}};
    }
\end{tikzpicture}}

\newcommand{\labeledTriangle}[6]{
\begin{tikzpicture}[scale=1*\scaleOfPictures]
    \begin{scope}[every node/.append style=Node]
        \node at (0, 0) (K1) {};
        \node at (1, 0) (K2) {};
        \node at (0.5, .866) (K3) {};
    \end{scope}
    \begin{pgfonlayer}{background}
        \lEdge[Edge, draw=#1, swap]{K1}{#2}{K2}
        \lEdge[Edge, draw=#3, swap]{K2}{#4}{K3}
        \lEdge[Edge, draw=#5]      {K1}{#6}{K3}
    \end{pgfonlayer}
\end{tikzpicture}
}

\let\originalforall=\forall
\renewcommand{\forall}{\mathop{\vcenter{\hbox{\Large$\originalforall$}}}}
\let\originalexists=\exists
\renewcommand{\exists}{\mathop{\vcenter{\hbox{\Large$\originalexists$}}}}

\let\originalunderline=\underline
\renewcommand{\underline}[1]{\originalunderline{\smash{#1}}}

\begin{document}

\begin{frontmatter}

\title{WQO dichotomy for 3-graphs\tnoteref{mytitlenote}}
\tnotetext[mytitlenote]{A preliminary shortened version of this paper appeared as~\cite{initial}.}

\author{S{\l}awomir Lasota\fnref{slfootnote}}
\author{Rados{\l}aw Pi{\'o}rkowski\fnref{rpfootnote}}
\address{Institute of Informatics, University of Warsaw}
\fntext[slfootnote]{Partially supported by the European Research Council (ERC) project Lipa
under the EU Horizon 2020 research and innovation programme (grant
agreement No. 683080).} 
\fntext[rpfootnote]{Partially supported by the Polish NCN grant 2016/21/B/ST6/01505.}

\begin{abstract}
We investigate data-enriched models, like Petri nets with data, 
where executability of a transition is conditioned by a relation between data values involved.
Decidability status of various decision problems in such models may depend on the structure of data domain. 
According to the WQO Dichotomy Conjecture, if a data domain is homogeneous then it either exhibits 
a well quasi-order (in which case decidability follows by standard arguments), or essentially all the decision problems are undecidable
for Petri nets over that data domain.

We confirm the conjecture for data domains being 3-graphs (graphs with 2-colored edges). 
On the technical level, this results is a significant step towards classification of homogeneous 3-graphs, 
going beyond known classification results for homogeneous structures.
\end{abstract}


\end{frontmatter}


\section{Introduction}


In Petri nets with data, tokens carry values from some data domain, and executability
of transitions is conditioned by a relation between data values involved.
One can consider \emph{unordered data}, like in~\cite{RF11}, i.e.,~an infinite data domain with the equality as the only  relation; or \emph{ordered data}, like in~\cite{LNORW07}, i.e.,~an infinite densely totally ordered data domain; or timed data, 
like in timed Petri nets~\cite{AN01} and timed-arc Petri nets~\cite{JJMS11}.
In~\cite{Las16} an abstract setting of Petri nets with an arbitrary fixed data domain $\A$ has been introduced, 
parametric in a relational structure $\A$.
The setting uniformly subsumes unordered, ordered and timed data (represented by $\A = (\N, =)$, $\A = (\Q, \leq)$ and $\A = (\Q, \leq, +1)$, respectively).

Following~\cite{Las16},
in order to enable finite presentation of Petri nets with data, 
and in particular to consider such models as input to algorithms, 
we restrict to relational structures $\A$ that are \emph{homogeneous}~\cite{survey} and \emph{effective} (the formal definitions are given in Section~\ref{sec:pn}). 
Certain standard decision problems (like the termination problem, the boundedness problem, or the coverability problem, 
jointly called from now on \emph{standard problems}) 
are all  decidable for Petri nets with ordered data~\cite{LNORW07} (and in consequence also for Petri nets with unordered data), 
as the model fits into the framework of well-structured transition systems of~\cite{FS01}. 
Most importantly, the structure $\A = (\Q, \leq)$ of ordered data \emph{admits well quasi-order} (\wqo) in the following sense:
for any \wqo $X$, the set of finite induced substructures of $(\Q, \leq)$ (i.e., finite total orders) labeled by elements of $X$, ordered naturally by embedding, is a \wqo
(this is exactly Higman's lemma).
Moreover, essentially the same argument can be used for any other homogeneous effective data domain which admits \wqo (see~\cite{Las16} for details).
On the other hand, for certain homogeneous effective data domains $\A$ the standard problems become all undecidable.
In the quest for understanding the decidability borderline, the following hypothesis has been formulated in~\cite{Las16}:
\begin{conjecture}[\Wqo Dichotomy Coinjecture~\cite{Las16}]
For an effective homogeneous structure $\A$, either $\A$ admits \wqo (in which case the standard problems are decidable 
for Petri nets with data $\A$), or all the standard problems are undecidable for Petri nets with data $\A$.
\end{conjecture}
According to~\cite{Las16}, the conjecture could have been equivalently stated for another data-enriched models,
e.g., for finite automata with one register~\cite{BBKL12}.
In this paper we consider, for the sake of presentation, only Petri nets with data.
\Wqo Dichotomy Conjecture holds in special cases when data domains $\A$ are undirected or directed graphs, due to the 
known classifications
of homogeneous graphs~\cite{Lachlan80,Cherlinbook}.

\para{Contributions}

We confirm the \Wqo Dichotomy Conjecture for data domains $\A$ being \emph{strongly}\footnote{
Strong homogeneity is a mild strengthening of homogeneity.} 
homogeneous \emph{3-graphs} (cf.~Thm.~\ref{thm:main} in Section~\ref{sec:results}).
A 3-graph is a logical structure with three irreflexive symmetric binary relations such that every pair of elements
of $\A$ belongs to exactly one of the relations
(essentially, a clique with 3-colored edges).

Our main technical contribution is a complex analysis of possible shapes of strongly homogeneous 3-graphs,
constituting the heart of the proof (cf.~Thm.~\ref{thm:core} in Section~\ref{sec:results}): we prove that a strongly homogeneous 3-graph either admits
\wqo (and thus its structure is very simple) or it embeds arbitrarily long paths. 
We believe that this result, being independent of a particular model of Petri nets for which the conjecture is formulated,
is a significant step towards full classification of homogeneous 3-graphs.
The classification of homogeneous structures is a well-known challenge in model theory, and has been only solved 
in some cases by now:
for undirected graphs~\cite{Lachlan80}, directed graphs (the proof of Cherlin spans a book~\cite{Cherlinbook}), 
multi-partite graphs~\cite{multipartite12}, and few others (the survey \cite{survey} is an excellent overview  
of homogeneous structures). 
Although the full classification of homogeneous 3-graphs was not our primary objective, 
we believe that our analysis significantly improves our understanding of these structures and can be helpful for classification.

Our result does not fully settle the status of  the \Wqo Dichotomy Conjecture.
Dropping the (mild) strong homogeneity assumption,
as well as extending the proof to arbitrarily many symmetric binary relations,
is left for future work.

%
%


\para{Related research}
Net models similar to Petri nets with data have been continuously proposed since the 80s, including,
among the others, high-level Petri nets~\cite{GL81}, colored Petri nets~\cite{J81},
unordered and ordered data nets~\cite{LNORW07}, $\nu$-Petri nets~\cite{RF11},
and constraint multiset rewriting~\cite{CDLMS99,D02a,D02b}.
Petri nets with data can be also considered as a reinterpretation of the classical definition of Petri nets in sets with atoms~\cite{BKL11full,BKLT13}, 
where one allows for \emph{orbit-finite} sets of places and transitions instead of just finite ones.
The decidability and complexity of standard problems for Petri nets over various data domains has attracted a lot of attention recently, see for instance~\cite{HLLLST16,LNORW07,LazicSchmitz16,R17,RF11}.

{\Wqo}s are important for their wide applicability in many areas.
Studies of {\wqo}s similar to ours, in case of graphs, have been conducted by Ding~\cite{Ding92} and Cherlin~\cite{Cherlin11};
their framework is different though, as they concentrate on subgraph ordering while we investigate \emph{induced} subgraph (or substructure) ordering.

\para{Outline}
We start by defining the model of Petri nets with data (in Section~\ref{sec:pn}), formulate our results
(Theorems~\ref{thm:main} and \ref{thm:core} in Section~\ref{sec:results}) and argue how Theorem~\ref{thm:core} implies Theorem~\ref{thm:main}
thus confirming the \Wqo Dichotomy Conjecture (in Section~\ref{sec:core-implies-main}).
Then the main technical part of the paper, spanning over Sections~\ref{sec:proof-core}--\ref{sec:coreA}, is devoted exclusively to the proof
of Theorem~\ref{thm:core}. This part is independent of the model of Petri nets with data, and conducts a complex and delicate analysis of consequences of
the amalgamation property for strongly homogeneous 3-graphs.


\section{Petri nets with homogeneous data} \label{sec:pn}


In this section we provide all necessary preliminaries.
Our setting follows~\cite{Las16} and is parametric in the underlying logical structure $\A$, 
which constitutes a \emph{data domain}.
Here are some example data domains:

\begin{itemize}
\item \emph{Equality data domain}: natural numbers with equality $\Aeq = (\nat, =)$.
Note that any other countably infinite set 
could be used instead of natural numbers, as the only available relation is equality.
\item \emph{Total order data domain}: rational numbers with the standard order $\Ato = (\Q, \leq)$.
Again, any other countably infinite dense total order without extremal elements could be used instead.
\item \emph{Nested equality data domain}: $\Aneq = (\nat^2, =_1, =)$ where $=_1$ is equality on 
the first component: $(n, m) =_1 (n', m')$ if $n = n'$ and $m\neq m'$. Essentially, $\A$ is an equivalence relation with infinitely many
infinite equivalence classes.
\end{itemize}
Note that two latter structures essentially extend the first one: in each case the equality is either present explicitly,
or is definable.
From now on, we always assume a fixed countably infinite
relational structure $\A$ with equality over a finite vocabulary (signature) $\Sigma$.

\para{Petri nets with data}

Petri nets with data are exactly like classical place/tran\-si\-tion Petri nets, 
except that tokens carry data values and these data values must satisfy a
prescribed constraint when a transition is executed.
Formally, a \emph{Petri net with data $\atoms$} consists of two disjoint finite sets $P$ (places) and $T$ (transitions), 
the arcs $A \subseteq P{\times} T \cup  T{\times} P$, and two labelings:
\begin{itemize}
\item arcs are labelled by pairwise disjoint finite nonempty sets of variables; 
\item transitions are labelled by first-order formulas over the vocabulary $\Sigma$ of $\atoms$, such that free variables
of the formula labeling a transition $t$ belong to the union of labels of the arcs incident to $t$. 
\end{itemize}
\begin{example} \rm
For illustration consider a Petri net with equality data $\Aeq$, 
with two places $p_1, p_2$ and two transitions $t_1, t_2$ depicted on Fig.~\ref{fig-udpn}.
\begin{figure}[bpt]
\centering
\newcommand{\netBefore}{%
\begin{tikzpicture}[xscale=2]
    \node (P) [PNBigPlace, label=above:$p_1$] at (1, 0) {\threeTokensLabeled{3}{cA}{1}{cB}{1}{cB}};
    \node (Q) [PNBigPlace, label=above:$p_2$] at (3, 0) {\twoTokensLabeled{5}{cC}{3}{cA}};

    \node (T) [PNTransition, label=above:$t_1$] at (0, 0) {};
    \node (S) [PNTransition, label=above:$t_2$] at (2, 0) {};

    \node[rectangle,draw, above=0.7cm of T] (phi1) {$x_1 \neq x_2$};
    \draw[-] ($(T.north)+(0,0.4)$) -- (phi1);
    \node[rectangle,draw, above=0.7cm of S] (phi2) {$y_1 = y_2 \neq z_3 \ \land \ z_1 = z_2$};
    \draw[-] ($(S.north)+(0,0.4)$) -- (phi2);

    \begin{pgfonlayer}{background}
        \coordinate (PEast) at (P.east);
        \draw[FArrow2, transform canvas={yshift= 0.1cm}] (P.east) to[bend left ] node[above]{$y_1$}  (S.west);
        \draw[FArrow2, transform canvas={yshift=-0.1cm}] (S.west) to[bend left ] node[below]{$z_1$}  (P.east);
        \draw[FArrow2, transform canvas={yshift= 0.1cm}] (Q.west) to[bend right] node[above]{$y_2$}  (S.east);
        \draw[FArrow2, transform canvas={yshift=-0.1cm}] (S.east) to[bend right] node[below]{$z_2, z_3$}  (Q.west);
        \draw[FArrow2] (T.east) to node[above]{$x_1, x_2$}  (P.west);
    \end{pgfonlayer}
\end{tikzpicture}%
}
\newcommand{\netAfter}{%
\begin{tikzpicture}[xscale=2]
    \node (P) [PNBigPlace, label=above:$p_1$] at (1, 0) {\threeTokensLabeled{4}{cY}{1}{cB}{1}{cB}};
    \node (Q) [PNBigPlace, label=above:$p_2$] at (3, 0) {\threeTokensLabeled{4}{cY}{5}{cC}{6}{cX}};

    \node (T) [PNTransition, label=above:$t_1$] at (0, 0) {};
    \node (S) [PNTransition, label=above:$t_2$] at (2, 0) {};

    \begin{pgfonlayer}{background}
        \coordinate (PEast) at (P.east);
        \draw[FArrow2, transform canvas={yshift= 0.1cm}] (P.east) to[bend left ] node[above]{$y_1$}  (S.west);
        \draw[FArrow2, transform canvas={yshift=-0.1cm}] (S.west) to[bend left ] node[below]{$z_1$}  (P.east);
        \draw[FArrow2, transform canvas={yshift= 0.1cm}] (Q.west) to[bend right] node[above]{$y_2$}  (S.east);
        \draw[FArrow2, transform canvas={yshift=-0.1cm}] (S.east) to[bend right] node[below]{$z_2, z_3$}  (Q.west);
        \draw[FArrow2] (T.east) to node[above]{$x_1, x_2$}  (P.west);
    \end{pgfonlayer}
\end{tikzpicture}%
}
\begin{tikzpicture}
    \node (before) at (0, 0) {\netBefore};
    \node (after) at (3, -2) {\netAfter};
    \draw[FlowArrow2,draw=cGray] (4,-0.505) to[out=0, in=90] node[auto] {execution of $t_2$} (5,-1.5);
\end{tikzpicture}
\caption{A Petri net with equality data, with places $P = \set{p_1, p_2}$ and transitions $T = \set{t_1, t_2}$.  
In the shown configuration, $t_2$ can be fired: consume two tokens carrying $3$, and put, e.g., token carrying $4$ on $p_1$ 
and tokens carrying $4, 6$ on $p_2$.
\label{fig-udpn}}
\end{figure}
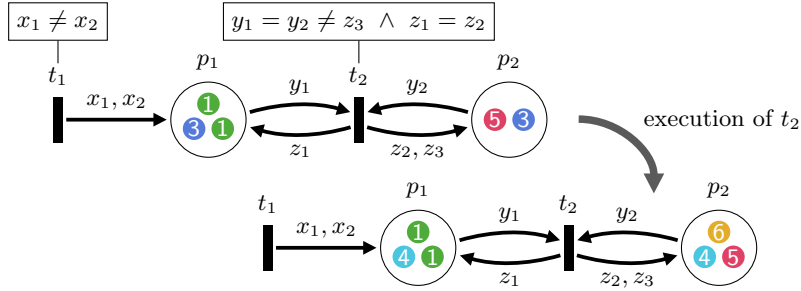
%
Transition $t_1$ outputs two tokens with arbitrary but distinct data values onto place $p_1$. Transition $t_2$ inputs two tokens with the same data value, say $a$, one from $p_1$ and one from $p_2$, and outputs 3 tokens: two tokens with arbitrary but equal data values, say $b$, one onto $p_1$ and the other onto $p_2$; and one token with a data value $c \neq a$ onto $p_2$.
Note that the transition $t_2$ does not specify whether $b=a$, or $b = c$, or $b \neq a,c$,
and therefore all three options are allowed.
Variables $y_1, y_2$ can be considered as input variables of $t_2$, while variables $z_1, z_2, z_3$ can be considered as output ones; analogously, $t_1$ has no input variables, and two output ones $x_1, x_2$.
\end{example}

The formal semantics of Petri nets with data is given by translation to multiset rewriting.
Given a set $X$, finite or infinite, a finite multiset over $X$ is a finite (possibly empty) partial function from $X$ to positive integers.
In the sequel let $\multisets{X}$ stand for the set of all finite multisets over $X$. 
A \emph{multiset rewriting system} $(\places, \transitions)$ consists of a set $\places$ together with a set of rewriting rules:
\[
\transitions \quad \subseteq \quad \multisets{\places} \times \multisets{\places}.
\]
Configurations $C \in \multisets{\places}$ are finite multisets over $\places$, and 
the step relation $\trans$ between configurations is defined as follows: for every $(I, O) \in \transitions$ 
and every $M \in \multisets{\places}$, there is the step ($+$ stands for multiset union)  
\[
M + I \ \trans \ M + O.
\]
For instance, a classical Petri net induces a multiset rewriting system where $\places$ is the set of places, and $\transitions$ is essentially the set of transitions, both $\places$ and $\transitions$ being finite. Configurations correspond to markings.

A Petri net with data $\atoms$ induces a multiset rewriting system $(\places, \transitions)$, 
where $\places = P \times \atoms$ and thus is infinite.
Configurations are finite multisets over $P\times \atoms$ (cf.~a configuration depicted in Fig.~\ref{fig-udpn}).
The rewriting rules $\transitions$ are defined as
\[
\transitions \quad = \quad \bigcup_{t\in T} \transitions_t,
\]
where the relation
$
\transitions_t \subseteq  \multisets{\places} \times \multisets{\places}
$
is defined as follows:
Let $\phi$ denote the formula labeling the transition $t$,  and let $X_i$, $X_o$ be the sets of input and output variables of $t$.
Every valuation $v_i : X_i \to \atoms$ gives rise to a multiset $M_{v_i}$ over $\places$, where $M_{v_i}(p, a)$ is the
(positive) number of variables $x$ labeling the arc $(p, t)$ with $v_i(x) = a$. Likewise for valuations $v_o : X_o \to \atoms$.
Then let 
\[
\transitions_t = \setof{(M_{v_i}, M_{v_o})}{v_i : X_i \to \atoms, \ v_o : X_o \to \atoms, \ v_i, v_o \vDash \phi}.
\]
Like $\places$, the set of rewriting rules $\transitions$ is infinite in general.

As usual, for a net $N$ and its configuration $C$, a run of $(N, C)$ is a maximal, finite or infinite, 
sequence of steps starting in $C$. 

\begin{remark}
As for classical Petri nets, an essentially equivalent definition can be given in terms of vector addition systems
(such a variant has been used in~\cite{HLLLST16} for equality data).
Petri nets with equality data are equivalent to (even if defined differently than) unordered data Petri nets of~\cite{LNORW07},
and Petri nets with total ordered data are equivalent to ordered data Petri nets of~\cite{LNORW07}.
\end{remark}

\para{Effective homogeneous structures}

For two relational $\Sigma$-structures $\str{A}$ and $\str{B}$ we say that $\str{A}$ \emph{embeds} in $\str{B}$, written $\str{A} \embedsin \str{B}$, if $\str{A}$ is isomorphic to an induced substructure of $\str{B}$,
i.e.,~to a structure obtained by restricting $\str{B}$ to a subset of its domain. This is witnessed by an injective function\footnote{We deliberately 
do not distinguish a structure $\str{A}$ from its domain set.} $h : \str{A} \to \str{B}$, 
which we call \emph{embedding}.
We write $\age{\A} = \setof{\str{A} \text{ a finite structure}}{\str{A} \embedsin \A}$ for the class of all finite structures that embed into $\A$,
and call it \emph{the age of $\A$}.

Homogeneous structures are defined through their automorphisms: $\A$ is homogeneous if every isomorphism of 
two of its finite induced substructures
extends to an automorphism of $\A$. In the sequel we will also need an equivalent definition using amalgamation.
An \emph{amalgamation instance} consists of  three structures $\str{A}, \str{B}_1, \str{B}_2 \in \age{\A}$ and two embeddings 
$h_1 : \str{A}\to\str{B}_1$ and $h_2 : \str{A} \to \str{B}_2$.
A solution of such instance is a structure $\str{C} \in \age{\A}$ and two embeddings 
$g_1 : \str{B}_1 \to \str{C}$ and $g_2 : \str{B}_2 \to \str{C}$ such that
$g_1 \circ h_1 = g_2 \circ h_2$ (we refer the reader to~\cite{Fraissebook} for further details). 
Intuitively, $\str{C}$ represents 'gluing' of $\str{B}_1$ and $\str{B}_2$ along the partial bijection
$h_2 \circ ({h_1}^{-1})$.
In this paper we will restrict ourselves to \emph{singleton} amalgamation instances, where only one element of $\str{B}_1$ is outside
of $h_1(\str{A})$, and likewise for $\str{B}_2$.
\begin{wrapfigure}[4]{!h}{3.1cm}
    \vspace{-7mm}
    \centering
    \begin{tikzpicture}[scale=\scaleOfPictures]
        \node (N1) at (0, 0) {
        \begin{tikzpicture}[scale=1.4*\scaleOfPictures]
            \path (-1,-0.2) rectangle (1, 1.2);
            \node[ Node] at (0, 0) (n1) {};
            \node[ Node] at (0, 1) (n2) {};
            \node[ANode] at (-0.866, 0.5) (n3) {};
            \node[ANode] at ( 0.866, 0.5) (n4) {};

            \begin{pgfonlayer}{background}
                \lEdge[cA]{n1}{}{n2}
                \lEdge[cA]{n1}{}{n3}
                \lEdge[cA, swap]{n2}{}{n3}
                \lEdge[cA, swap]{n1}{}{n4}
                \lEdge[cA]{n2}{}{n4}
                \lEdge[AEdge, bend left=20]{n3}{}{n4}
            \end{pgfonlayer}
        \end{tikzpicture}
        };
    \end{tikzpicture}
\end{wrapfigure}

An example singleton amalgamation instance is shown on the right,
where the graph $\str{A}$ consists of the single edge connecting two middle black nodes, 
$\str{B}_1$ is the left triangle, and $\str{B}_2$ the right one. The dashed line represents an edge that may 
(but does not have to) appear in a solution.
$\A$ is homogeneous if, and only if 
every amalgamation instance has a solution; in such case we say that $\age{\A}$ has the \emph{amalgamation property}.
See~\cite{survey} for further details.

A solution $\str{C}$ necessarily satisfies $g_1(h_1(\str{A})) = g_2(h_2(\str{A})) \subseteq g_1(\str{B}_1) \cap g_2(\str{B}_2)$;
a solution is \emph{strong} if $g_1(h_1(\str{A})) = g_1(\str{B}_1) \cap g_2(\str{B}_2)$.
Intuitively, this forbids additional gluing of $\str{B}_1$ and $\str{B}_2$ not specified by
the partial  bijection $h_2 \circ ({h_1}^{-1})$. 
If every amalgamation instance has a strong solution we call $\A$ \emph{strongly homogeneous}.
This is a mild restriction, as homogeneous structures are typically strongly homogeneous.

The equality, nested equality, and total order data domains are strongly homogeneous structures. For instance, in the latter case 
finite induced substructures are just finite total orders, which satisfy the strong amalgamation property.  
Many other natural classes of structures have the amalgamation property: finite graphs, finite directed graphs, finite partial orders, finite tournaments, etc. 
Each of these classes is the age of a strongly homogeneous relational structure, namely
the \emph{universal graph} (called also random graph),  
the universal directed graph, the universal partial order, the universal tournament, respectively.
Examples of homogeneous structures abound~\cite{survey}.

Homogeneous structures admit quantifier elimination: every first-order formula is equivalent to (i.e., defines the same set as) a quantifier-free one~\cite{survey}.
Thus it is safe to assume that formulas labeling transitions are quantifier-free.

\para{Admitting \wqo}

A \emph{well quasi-order} (\wqo) is a well-founded quasi-order with no infinite antichains.
For instance, finite multisets $\multisets{P}$ over a finite set $P$, ordered by multiset inclusion $\submultiseteq$, are a \wqo.
Another example is the embedding quasi-order $\embedsin$ in $\age{\Ato}$ (= all finite total orders) isomorphic to the ordering of natural numbers.
Finally, the embedding quasi-order in $\age{\A}$ can be lifted from finite structures to finite structures \emph{labeled}  
by elements of some ordered set $(X, \leq)$: for two such labeled structures $a : \str{A} \to X$ and $b: \str{B} \to X$ we define
$
a \embedsinpar{X} b
$
if some embedding $h : \str{A} \to \str{B}$ satisfies $a(x) \leq b(h(x))$ for every $x\in\str{A}$.
We say that $\A$ \emph{admits \wqo} when for every \wqo $(X, \leq)$, the lifted embedding order $\embedsinpar{X}$
is a \wqo too.
For instance, $\Ato$ admits \wqo by Higman's lemma.

Note the natural correspondence between configurations of a Petri net with data $\A$, and structures $\str{A} \in \age{\A}$
labeled by finite multisets over the set $P$ of places:
\[
\multisets{P \times \A} \quad \equiv \quad \setof{m : \str{A} \to \multisets{P}}{\str{A}\in\age{\A}}.
\]
Thus the lifted embedding quasi-order $\embedsinpar{\multisets{P}}$ is a quasi-order on configurations. 

\para{Standard decision problems}

A Petri net with data $N$ can be finitely represented by finite sets $P, T, A$ and appropriate labelings with variables and formulas.
Due to the homogeneity of $\A$, 
a configuration $C$ can be represented (up to automorphism of $\A$) by a structure $\str{A} \in \age{A}$ labeled by 
$\multisets{P}$.
We can thus consider the classical decision problems that input Petri nets with data $\A$, like the \emph{termination problem}:
does a given $(N, C)$ have only finite runs? The data domain is considered as a parameter, 
and hence itself does not constitute part of input.
Another classical problem is the \emph{place non-emptiness problem} (markability): given $(N, C)$ and a place $p$ of $N$, 
does $(N, C)$ admit a run that puts at least one token on place $p$?
One can also define the appropriate variants of the coverability problem (equivalent to the place non-emptiness problem),
the boundedness problem, the inevitability problem, etc. (see~\cite{Las16} for details).
All the decision problems mentioned above we jointly call \emph{standard problems}.

A $\Sigma$-structure $\A$ is called \emph{effective} if
the following \emph{age problem} for $\A$ is decidable: given a finite $\Sigma$-structure $\str{A}$, decide whether $\str{A} \embedsin \atoms$.
If $\A$ admits \wqo then application of the framework of well-structured transition systems~\cite{FS01}
to the lifted embedding order $\embedsinpar{\multisets{P}}$ yields:
\begin{theorem}[\cite{Las16}]
If an effective homogeneous structure $\A$ admits \wqo then
all the standard problems are decidable for Petri nets with data $\A$.
\end{theorem}
For homogeneous undirected (and also directed) graphs, the \Wqo Dichotomy Conjecture is easily shown 
by inspection of the classifications thereof~\cite{Lachlan80,Cherlinbook}.
We state in Theorem~\ref{thm:graphs} a core fact underlying the dichotomy, for future use.
A \emph{path} is a finite graph with  nodes $\set{v_1, \ldots, v_n}$ whose only edges are pairs $\set{v_i, v_{i+1}}$.
The nodes $v_1, v_n$ are \emph{ends} of the path, and $n$ is its length. 
\begin{theorem}[follows by~\cite{Lachlan80,Cherlinbook}] \label{thm:graphs}
A homogeneous graph $\A$ either admits \wqo, or $\A$ contains arbitrarily long paths as induced subgraphs, 
or the complement of $\A$ contains arbitrarily longh paths as induced subgraphs.
\end{theorem}
\noindent
Theorem~\ref{thm:graphs} implies the conjecture for graphs (the proof is in Section~\ref{sec:core-implies-main}):
\begin{corollary}
A homogeneous graph $\A$ either admits \wqo, or 
all standard problems are undecidable for Petri nets with data $\A$.
\end{corollary}


\section{Results} \label{sec:results}

A 3-graph $\G = (V, C_1, C_2, C_3)$ consists of a set $V$ 
and three irreflexive symmetric binary relations $C_1, C_2, C_3 \subseteq V^2$
such that every pair of distinct elements of $V$ belongs to exactly one of the three relations.
Any graph, including $\Aeq$ and $\Aneq$, can be seen as a 3-graph. 
In the sequel we treat a 3-graph as a clique with 3-colored edges.  

\begin{example} \label{ex:strongly} \rm
We provide an example of a strongly homogeneous 3-graph.
A 3-vertex 3-graph we call a \emph{triangle}; here are three triangles, where colors red, green and blue correspond to
relations $C_1, C_2$ and $C_3$, respectively:
    \begin{center}
        \begin{tikzpicture}[scale=\scaleOfPictures]
            \node[anchor=east, inner sep=0] at ( 0, 0) {\labeledTriangle{cA}{}{cA}{}{cC}{}};
            \node[anchor=east, inner sep=0] at ( 3, 0) {\labeledTriangle{cA}{}{cB}{}{cB}{}};
            \node[anchor=east, inner sep=0] at (6, 0) {\labeledTriangle{cA}{}{cA}{}{cA}{}};
        \end{tikzpicture}
    \end{center}
The three triangles, treated as forbidden patterns, define an infinite homogeneous 3-graph $\G$ as follows.
Consider the class $\cal C$ of all finite 3-graphs that do not embed any of the three triangles shown above.
The class has the amalgamation property (cf.~Appendix in \cite{survey}) -- it is not difficult to see that 
every singleton amalgamation instance $\str{A}, \str{B}_1, \str{B}_2$ can be solved using a green or red edge.
If the common part $\str A$ contains at most one element this follows by inspection of forbidden triangles. 
Otherwise,
supposing towards contradiction that the instance disallows either red or green edge as a solution, which means that
the instance contains necessarily the following pattern:
%
\begin{center}
    \begin{tikzpicture}[scale=\scaleOfPictures]
        \node (N1) at (0, 0) {
        \begin{tikzpicture}[scale=1.4*\scaleOfPictures]
            \path (-1,-0.2) rectangle (1, 1.2);
            \node[ Node] at (0, 0) (n1) {};
            \node[ Node] at (0, 1) (n2) {};
            \node[ANode] at (-0.866, 0.5) (n3) {};
            \node[ANode] at ( 0.866, 0.5) (n4) {};

            \begin{pgfonlayer}{background}
                \lEdge[cGray]{n1}{}{n2}
                \lEdge[cA]{n1}{}{n3}
                \lEdge[cA, swap]{n2}{}{n3}
                \lEdge[cB, swap]{n1}{}{n4}
                \lEdge[cA]{n2}{}{n4}
                \lEdge[AEdge, bend left=20]{n3}{}{n4}
            \end{pgfonlayer}
        \end{tikzpicture}
        };
    \end{tikzpicture}
\end{center}
%
with black representing some unknown color, we observe that any choice of color for the black edge leads to a forbidden triangle,
a contradiction. 
In consequence, there is a homogenoeus 3-graph $\G$ with $\age{\G} = \cal C$ (cf.~\cite{Fraissebook}).
One easily verifies that $\G$ is \emph{strongly} homogeneous; indeed, every singleton amalgamation instance
that admits a glueing solution admits also a solution where the green color is used instead of glueing, simply because all the following
triangles are not forbidden:
    \begin{center}
        \begin{tikzpicture}[scale=\scaleOfPictures]
            \node[anchor=east, inner sep=0] at ( 0, 0) {\labeledTriangle{cB}{}{cB}{}{cB}{}};
            \node[anchor=east, inner sep=0] at ( 3, 0) {\labeledTriangle{cB}{}{cA}{}{cA}{}};
            \node[anchor=east, inner sep=0] at (6, 0) {\labeledTriangle{cB}{}{cC}{}{cC}{}};
        \end{tikzpicture}
    \end{center}
\end{example}
Our main result confirms the \Wqo Dichotomy Conjecture for strongly homogeneous 3-graphs:
\begin{theorem}  \label{thm:main}
A strongly homogeneous 3-graph $\G$ either admits \wqo, or 
all standard problems are undecidable for Petri nets with data $\G$.
\end{theorem}
The core technical result of the paper is Theorem~\ref{thm:core} below.
\begin{theorem}  \label{thm:core}
A strongly homogeneous 3-graph $\G$ either admits \wqo, or for some $i, j \in \set{1, 2, 3}$ (not necessarily distinct)
the graph $(V, C_i \cup C_j)$ contains arbitrarily long paths as induced subgraphs. 
\end{theorem}
We prove that Theorem~\ref{thm:core} implies Theorem~\ref{thm:main} in the next section.  
Then, in the rest of the paper we concentrate solely on the proof of Theorem~\ref{thm:core}.
%

\begin{wrapfigure}[11]{H}{5cm}
    \vspace{-0.7cm}
    \centering
    \tikzstyle{Vert} = [Edge, draw=cC]
    \tikzstyle{Horiz} = [Edge, draw=cB]
    \begin{tikzpicture}[scale=1.2*\scaleOfPictures]
        \foreach \x in {1,2,3,4} {
            \foreach \y in {1,2,3,4} {
            \node[Node, anchor=center] at (\x, \y) {};
            }
        }
        \foreach \n in {1,2,3,4} {
        \node[anchor=center] at (\n, 0.6) {$\vdots$};
        \node[anchor=center] at (4.5, \n) {$\dots$ };
        }

        \begin{pgfonlayer}{background}
			\newcounter{cnt}
			\newcounter{cntb}
			\foreach \x in {1,2,3,4} {
				\foreach \y in {4,3,2,1} {
					\foreach \n in {1,2,3,4} {
						\setcounter{cnt}{\y-\x}
						
						\ifthenelse{\value{cnt} > 1}{
							\draw[Vert ] (\n,\x) to[bend left=30] (\n,\y);
							\draw[Horiz] (\x,\n) to[bend left=30] (\y,\n);
						}{}
						\ifthenelse{\value{cnt} = 1}{
							\draw[Vert ] (\n,\x) to[bend left =0] (\n,\y);
							\draw[Horiz] (\x,\n) to[bend right=0] (\y,\n);
						}{}
					}
				}
			}
        \end{pgfonlayer}

        \begin{pgfonlayer}{background}
            \draw[Selection, double distance=0.35cm, opacity=0.5] (1,1)
            \foreach \n in {1,2,3} {
            to[bend right=0] ++(1,0)  to[bend left=0]  ++(0, 1)
            }
            ;
            \foreach \n in {1,2,3} {
				\draw[Vert] (\n+1,\n) to[bend left =0] (\n+1,\n+1);
				\draw[Horiz] (\n,\n) to[bend right=0] (\n+1,\n);
			}
        \end{pgfonlayer}
    \end{tikzpicture}
    \label{fig:longpaths}
\end{wrapfigure}

\begin{example} \label{ex:main} \rm
For a quasi-order $(X, \leq)$, the multiset inclusion is defined as follows for $m, m' \in \multisets{X}$:
$m'$ is included in $m$ if $m'$ is obtained from $m$ by a sequence of operations, where each operation 
either removes some element, or replaces some element by a smaller one wrt.~$\leq$.
The structure $\Aeq = (\nat, =)$ admits \wqo.
Indeed, $\age{\Aeq}$ contains just finite pure sets, thus $\embedsinpar{X}$ is quasi-order-isomorphic to the multiset inclusion  
on $\multisets{X}$, and is therefore a \wqo whenever the underlying quasi-order $(X, \leq)$ is.
Similarly, $\Aneq = (\nat^2, =_1, =)$ also admits \wqo, as
$\embedsinpar{X}$ is quasi-order-isomorphic to the multiset inclusion on 
$\multisets{\multisets{X}}$.
\end{example}
%
\vspace{-1.2mm}
On the other hand, consider a 3-graph $(\nat^2, =_1, =_2, \neq_{12})$ where
$=_2$ is symmetric to $=_1$ and $(n, m) \neq_{12} (n', m')$ if $n \neq n'$ and $m \neq m'$.  
It refines $\Aneq$ and does not admit \wqo and hence satisfies the second case of Theorem~\ref{thm:core}.
Indeed, the graph $(\nat^2, =_1 \cup =_2)$ contains arbitrarily long paths 
of the shape presented on the right,
where the two colors depict $=_1$ and $=_2$, respectively, and lack of color corresponds to $\neq_{12}$.
Note that $(\nat^2, =_1, =_2, \neq_{12})$ is homogeneous but not strongly so.

Finally, the strongly homogenous 3-graph exhibited in Example~\ref{ex:strongly} also belongs to the second case of
Theorem~\ref{thm:core}. Indeed, no red-green triangle is forbidden and hence
the 3-graph contains a infinite red-green clique which contains,
as induced subgraphs, both arbitrarily long red paths and arbitrarily long green paths.
%

%


\section{Theorem~\ref{thm:core} implies Theorem~\ref{thm:main}}  \label{sec:core-implies-main}

Assume Theorem~\ref{thm:core} holds.
Towards proving Theorem~\ref{thm:main} consider an effective strongly homogeneous 3-graph 
$\A = (V, C_1, C_2, C_3)$ 
that does not admit \wqo and let $E = C_i \cup C_j \subseteq V^2$ given by Theorem~\ref{thm:core}.
%
Thus we know that the graph $(V, E)$ contains arbitrarily long paths.
We will demonstrate that Petri nets with data domain $\A$ can faithfully simulate computations of 2-counter machines.
To this aim we fix an arbitrary 
deterministic counter machine $\masz$ with two counters $c_1, c_2$, and states $Q$;
and construct a Petri net $N_\masz$ with data $\A$ that simulates the computation of 
$\masz$ starting in the initial configuration: initial state $q_\text{init}$ and the counter values $c_1 = c_2 = 0$. 
Places of the net will include
\[ \set{b_1, m_1, e_1, b_2, m_2, e_2, q, r} \cup Q \ \subseteq \ P \]
plus some further auxiliary ones.
In particular, every state of $\masz$ will have a corresponding place in $N$.
The idea is to represent a value $c_j = n$ by storing $n+2$ tokens carrying, as its values, nodes of a path of length $n+2$ 
in the graph $(V, E)$.
The two tokens carrying the ends of the path will be stored on places $b_j$ and $e_j$, respectively, while 
the remaining $n$ tokens will be stored on place $m_j$.
Simulation of a zero test amounts then to checking if the ends are related by an edge.
Simulation of a decrement amounts to replacing one end (say from place $e_j$) by its only neighbor from place $m_j$.
And simulation of an increment amounts to moving the token from $e_j$ to $m_j$, accompanied by 
production of a new token on place
$e_j$ carrying an arbitrary (guessed nondeterministically) value $v\in V$ 
not related by $E$ to any of the other tokens on places $b_j$ and $m_j$.

\paraemph{Zero test and decrement}
If $\masz$ does \emph{zero test} for $c_j$ in state $q$ and goes to $q'$,
the net $N_\masz$ has a transition $z_{j, q, q'}$ that inputs one token from $b_j$
and one token from $e_j$, checks that data values they carry are related by $E$, and puts back the same tokens to
the two places  (cf.~Fig.~\ref{fig-z}).
In addition, the transition $z_{j,q,q'}$ moves one token from place $q$ to $q'$, irrespectively of the data values it carries.
Similarly, \emph{decrement} of $c_j$ is performed by a transition $d_{j, q, q'}$ that inputs one token from $m_j$ and one token from
$e_j$, checks that data values they carry are related by $E$, and then puts back the former token to $e_j$ 
while discarding the latter one.
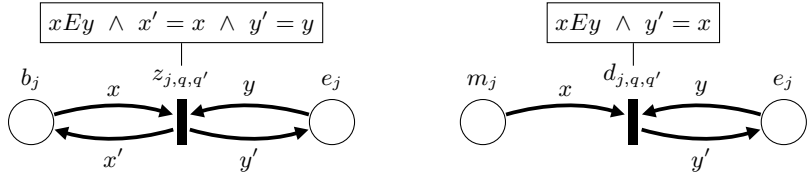
\begin{figure}[tbp]
\centering
\newcommand{\netOne}{%
\begin{tikzpicture}[xscale=2]
  \node (P) [PNPlace, label=above:$b_j$] at (1, 0) {};
  \node (Q) [PNPlace, label=above:$e_j$] at (3, 0) {};

  \node (S) [PNTransition, label={above:$z_{j,q,q'}$}] at (2, 0) {};
  \node[rectangle,draw, above=0.7cm of S] (phi2) {$x E y \ \land \ x' = x \ \land \ y' = y$};
  \draw[-] ($(S.north)+(0,0.4)$) -- (phi2);

  \begin{pgfonlayer}{background}
    \coordinate (PEast) at (P.east);
    \draw[FArrow2, transform canvas={yshift= 0.1cm}] (P.east) to[bend left ] node[above]{$x$}  (S.west);
    \draw[FArrow2, transform canvas={yshift=-0.1cm}] (S.west) to[bend left ] node[below]{$x'$}  (P.east);
    \draw[FArrow2, transform canvas={yshift= 0.1cm}] (Q.west) to[bend right] node[above]{$y$}  (S.east);
    \draw[FArrow2, transform canvas={yshift=-0.1cm}] (S.east) to[bend right] node[below]{$y'$}  (Q.west);
  \end{pgfonlayer}
\end{tikzpicture}%
}
\newcommand{\netTwo}{%
\begin{tikzpicture}[xscale=2]
  \node (P) [PNPlace, label=above:$m_j$] at (1, 0) {};
  \node (Q) [PNPlace, label=above:$e_j$] at (3, 0) {};

  \node (S) [PNTransition, label={above:$d_{j,q,q'}$}] at (2, 0) {};
  \node[rectangle,draw, above=0.7cm of S] (phi2) {$x E y \ \land \ y' = x$};
  \draw[-] ($(S.north)+(0,0.4)$) -- (phi2);

  \begin{pgfonlayer}{background}
    \coordinate (PEast) at (P.east);
    \draw[FArrow2, transform canvas={yshift= 0.1cm}] (P.east) to[bend left ] node[above]{$x$}  (S.west);
    \draw[FArrow2, transform canvas={yshift= 0.1cm}] (Q.west) to[bend right] node[above]{$y$}  (S.east);
    \draw[FArrow2, transform canvas={yshift=-0.1cm}] (S.east) to[bend right] node[below]{$y'$}  (Q.west);
  \end{pgfonlayer}
\end{tikzpicture}%
}
\begin{tikzpicture}
  \node (before) at (0, 0) {\netOne};
  \node (after) at (6, 0) {\netTwo};
\end{tikzpicture}
\caption{Transition $z_{j,q,q'}$ and $d_{j,q,q'}$ simulating zero test and decrement of counter $c_j$, respectively.
Places corresponding to control states of $\masz$ are omitted for simplicity.
\label{fig-z}}
\end{figure}

\paraemph{Increment}
Slightly more complicated is the simulation of increment of a counter $c_j$, 
as it involves creating a fresh value that must correctly extend, by one vertex, the path currently stored on places
$b_j, m_j, e_j$.
In the first step of the simulation, the net executes a transition $i_j$ that guesses a data value $v\in V$ related
by $E$ to the value $v_e$ carried by the single token on place $e_j$ but not to the value $v_b$ carried 
by the single token on place $b_j$;
the token from $e_j$ is moved to $m_j$ (and its copy is additionally put to an auxiliary place
$p$ for future use), and a new token carrying $v$ is put on $e_j$
(and its copy is additionally put to an auxiliary place $r$ for future use).
What remains to be checked in that $v$ has been guessed correctly by $i_j$, namely that 
$v$ is related by $E$ to none of the data values carried by tokens on $m_j$ except for $v_e$.
To this end the net performs a traversal through the path, in the direction from $v_e$ to $v_b$, 
in order to check the correctness of $v$.
The traversal is done by iterative execution of the transition $t_j$, depicted on Fig.~\ref{fig-c}, which uses
the places $p, r$ to store the current edge of the path in the course of traversal.
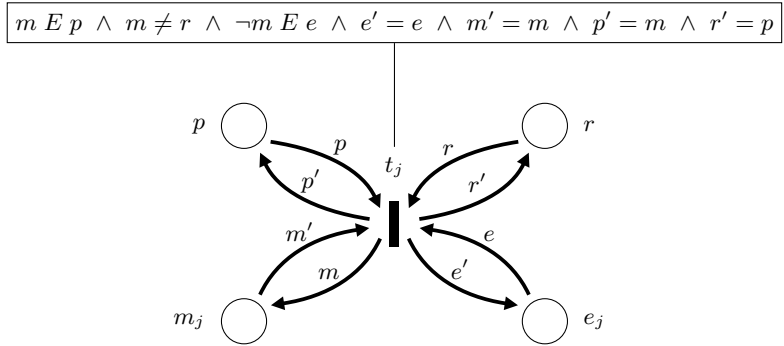
\begin{figure}[tbp]
\centering
\newcommand{\netThree}{%
\begin{tikzpicture}[xscale=2, yscale=1.3]
  \node (P1) [PNPlace, outer sep=0.1cm, label=left:$p$] at (0, 2) {};
  \node (P2) [PNPlace, outer sep=0.1cm, label=right:$r$] at (2, 2) {};
  \node (P3) [PNPlace, outer sep=0.1cm, label=left:$m_j$] at (0, 0) {};
  \node (P4) [PNPlace, outer sep=0.1cm, label=right:$e_j$] at (2, 0) {};

  \node (T) [PNTransition, label={above:$t_j$}, outer sep=0.2cm] at (1, 1) {};
  \node[rectangle,draw, above=1.9cm of T] (phi2) {$m \mathrel{E} p \ \land \ m \neq r  \ \land \  \neg m \mathrel{E} e  \ \land \  e' = e  \ \land \  m' = m  \ \land \  p' = m  \ \land \ r' = p$};
  \draw[-] ($(T.north)+(0,0.4)$) -- (phi2);

  \begin{pgfonlayer}{background}
    \coordinate (T1) at ($(T.west) + (0,0.1)$);
    \coordinate (T3) at ($(T.west) - (0,0.1)$);
    \coordinate (T2) at ($(T.east) + (0,0.1)$);
    \coordinate (T4) at ($(T.east) - (0,0.1)$);
    \draw[FArrow2, transform canvas={shift=(45:0.1cm)}] (P1.south east) to[bend left ] node[above]{$p$}  (T1);
    \draw[FArrow2, transform canvas={shift=(45:-0.1cm)}] (T1) to[bend left] node[above,pos=0.4]{$p'$}  (P1.south east);
    \draw[FArrow2, transform canvas={shift=(135:0.1cm)}] (P2.south west) to[bend right ] node[above]{$r$}  (T2);
    \draw[FArrow2, transform canvas={shift=(135:-0.1cm)}] (T2) to[bend right] node[above,pos=0.4]{$r'$} (P2.south west);
    \draw[FArrow2, transform canvas={shift=(135:0.1cm)}] (P3.north east) to[bend left ] node[above]{$m$}  (T3);
    \draw[FArrow2, transform canvas={shift=(135:-0.1cm)}] (T3) to[bend left] node[above, pos=0.6]{$m'$}  (P3.north east);
    \draw[FArrow2, transform canvas={shift=(45:0.1cm)}] (P4.north west) to[bend right ] node[above]{$e$}  (T4);
    \draw[FArrow2, transform canvas={shift=(45:-0.1cm)}] (T4) to[bend right] node[above, pos=0.6]{$e'$} (P4.north west);
  \end{pgfonlayer}
\end{tikzpicture}%
}
\begin{tikzpicture}
  \node (before) at (0, 0) {\netThree};
\end{tikzpicture}
\caption{Transition $t_j$ used in the simulation of increment on counter $c_j$.
\label{fig-c}}
\end{figure}
The condition $m \mathrel{E} p \ \land \ m \neq r$ checks that the value of variable $m$ is the other neighbour of $p$ along the path;
the condition $\neg m \mathrel{E} e$ checks that the guessed value $v$, stored on place $e_j$, is indeed not related by $E$ to
the value of $m$;
the condition $e' {=} e  \ \land \  m' {=} m$ ensures that the same value returns to places $m_j$ and $e_j$;
and finally the condition $p' {=} m  \ \land \ r' {=} p$ ensured that the current edge is moved along the path.

Finally, the simulation of increment of $c_j$ finishes with a transition $i'_j$ that is enabled when the value on place $p$ is 
related by $E$ to the value on place $b_j$; transition $i'_j$ removes the tokens from places $p$ and $r$.

Initial configuration $C_\masz$ of $N_\masz$ puts one token on each of the places $b_1, b_2, e_1, e_2$, using two 
arbitrary data values related by $E$, to encounter for
$c_1 = c_2 = 0$; and one token on the place corresponding to the initial state $q_\text{init}$.

We have thus sketched a construction of a net $N_\masz$ and the initial configuration $C_\masz$.
Observe that consecutive steps of $N_\masz$ faithfully simulate consecutive steps of $\masz$, using a path of 
sufficient length.
$N$ can however get stuck at some point of simulation, if the currently used path can not be
extended to a longer one; a priori, this could happen if
the fresh data values $v$ used in the simulation of increments are not guessed appropriately.
Nevertheless, 
since the net $N$ stops when a token is put on $p_\text{halt}$ (i.e., when no token is stored on places
in $Q\setminus \set{p_\text{halt}}$), we have:
\begin{claim}
The place $p_\text{halt}$ corresponding to the halting state of $\masz$ is nonempty in some run of $(N_\masz, C_\masz)$ if, and only if the machine $\masz$ halts. 
\end{claim}
In one direction, a run of $(N_\masz, C_\masz)$ putting a token on $p_\text{halt}$ simulates the halting run of $\masz$
from the initial configuration.
In the other direction, if $\masz$ halts then the net $N_\masz$ can use a sufficiently long path in $(V, E)$ for values $v$ guessed
in the simulation of increments to be able to simulate the whole computation of $\masz$ and finally put a token on place $p_\text{halt}$.
Thus the claim directly entails undecidability of the place non-emptiness problem, and hece also of the coverability problem.
To treat other decision problems, we notice that $(V, E)$ contains, in addition to arbitrarily long finite paths,
also an infinite $\omega$-path:
\begin{claim}
The graph $(V, E)$ contains an $\omega$-path.
\end{claim}
Indeed, treat finite paths as finite words over a 2-letter alphabet, and arrange all finite paths into a tree.
The tree contains arbitrarily long branches, thus it necessarily contains an infinite branch.
Using homogeneity of $\A$ one argues (see e.g.~Lemma 6.1.3 in \cite{Hodges}) 
that every infinite branch realizes as an $\omega$-path in $(V, E)$.
With the last claim we obtain:
\begin{claim}
$(N_\masz, C_\masz)$ terminates if and only if the machine $\masz$ halts. 
\end{claim}
Indeed, when the computation of $\masz$ from the initial configuration halts then $N_\masz$ necessarily terminates.
On the other hand, if the computation of $\masz$ from the initial configuration is infinite, an infinite $\omega$-path
in $(V, E)$ can be used for the simulation thus constituting an infinite run of $(N_\masz, C_\masz)$.
This entails undecidability of the termination problem, and hence also of the boundedness problem and the inevitability problem.




\section{Proof of Theorem~\ref{thm:core}}  \label{sec:proof-core}

%

%
From now on we consider a fixed 3-graph $\G = (V, C_1, C_2, C_3)$ as data domain,
assuming $\G$ to be countably infinite and strongly homogeneous. We treat $\G$ as a clique with
3-colored edges: we call $C_1, C_2$ and $C_3$
\emph{colors} and put $\Colors = \{ C_1, C_2, C_3\} \subset \mathcal{P}(V \times V)$. To denote individual
colors from this set, we will use variables $\a, \b, \c$ and $\x, \y, \z$.
%
    A path in the graph $(V, \a \cup \b)$ we call \emph{\a\b-path} ($\a\b \in \Colors$);
    for simplicity, we will write \emph{\a-path} instead of \a\a-path.
    Likewise we speak of $\a\b$-cliques, $\a$-cliques, $\a\b$-cycles, etc.
A \emph{triangle}
    $\triangle\a\b\c$ is a 3-clique with edges colored by $\a, \b, \c$.
    (Note that $\triangle\a\b\c = \triangle\b\c\a = \triangle\c\b\a$).

\vspace{-1mm}
\paragraph{Sketch of the proof}
Lemma~\ref{lemma:withBranches} below states that any 3-graph $\G$ has to meet one of the four listed cases.
It splits the proof into four separate paths:
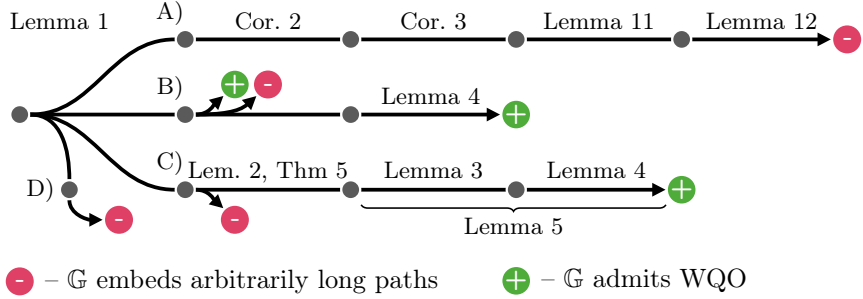
\begin{center}
	\newcommand{\cminus}{\begin{tikzpicture}[scale=0.02]
			\draw[fill=white] (-5,-0.6) rectangle (5, 0.6);
		\end{tikzpicture}}
	\newcommand{\cplus}{\begin{tikzpicture}[scale=0.02]
			\draw[fill=white] (-5,-0.6) rectangle (5, 0.6);
			\draw[fill=white] (-0.6,-5) rectangle (0.6, 5);
		\end{tikzpicture}}
	\tikzstyle{nminus} = [Node,fill=cC,minimum size=0.4cm,align=center]
	\tikzstyle{nplus} = [Node,fill=cB,minimum size=0.4cm,align=center]
    \begin{tikzpicture}[xscale=2.2, yscale=2]
        \node[Node]  (s) at (0,  0  ) {};

        \node[Node] (a0) at (1  ,  0.5) {};
        \node[inner sep=0.05cm, anchor=south, xshift=-0.2cm] (a0label) at (a0.north) {\footnotesize A)};
        \node[Node] (a1) at (2  ,  0.5) {};
        \node[Node] (a2) at (3  ,  0.5) {};
        \node[Node] (a3) at (4  ,  0.5) {};
        \node[nminus] (a4) at (5  ,  0.5) {\cminus};

        \node[Node] (b0) at (1  ,  0  ) {};
        \node[inner sep=0.05cm, anchor=south, xshift=-0.2cm] (b0label) at (b0.north) {\footnotesize B)};
        \node[nplus] (b0p) at (1.3  , 0.2) {\cplus};
        \node[nminus] (b0m) at (1.5,0.2) {\cminus};
        \node[Node] (b1) at (2  ,  0  ) {};
        \node[nplus] (b2) at (3  ,  0  ) {\cplus};

        \node[Node] (c0) at (1  , -0.5) {};
        \node[inner sep=0.05cm, anchor=south, xshift=-0.2cm] (c0label) at (c0.north) {\footnotesize C)};
        \node[nminus] (c0m) at (1.3,-0.7) {\cminus};
        \node[Node] (c1) at (2  , -0.5) {};
        \node[Node] (c2) at (3  , -0.5) {};
        \node[nplus] (c3) at (4  ,  -0.5) {\cplus};

        \node[Node] (d0) at (0.3, -0.5) {};
        \node[inner sep=0.05cm, anchor=east, xshift=0cm] (d0label) at (d0.west) {\footnotesize D)};
        \node[nminus] (d1) at (0.6,-0.7) {\cminus};


        \begin{pgfonlayer}{background}

            \draw[FLine2] (s) edge[out=0,in=180] (a0);
            \draw[FArrow2]
                (a0) -- node[midway,above](label1){\footnotesize Cor.~\ref{corollary:addingVertices}}
                (a1) -- node[midway,above]{\footnotesize Cor.~\ref{corollary:canAddP2xA}}
                (a2) -- node[midway,above]{\footnotesize Lemma~\ref{lemma:addingPPP}}
                (a3) -- node[midway,above]{\footnotesize Lemma~\ref{lemma:longPaths}}
                (a4);

            \draw[FLine2] (s) -- (b0) -- (b1);
            \draw[FArrow2]
                (b1) -- node[midway,above]{\footnotesize Lemma~\ref{lemma:FamilyOfIdenticalSlices}}
                (b2);
            \draw[FArrow2] (b0) edge[out=0,in=225] (b0p);
            \draw[FArrow2] (b0) edge[out=0,in=225] (b0m);

            \draw[FLine2] (s) edge[out=0,in=180] (c0);
            \draw[FArrow2]
                (c0) -- node[midway,above,anchor=base,yshift=1.5mm]{\footnotesize Lem.~\ref{lemma:kPartiteRestrictedIsHomoheneous}, Thm~\ref{theorem:bipartiteClassification}}
                (c1) -- node[midway,above,anchor=base,yshift=1.5mm]{\footnotesize Lemma~\ref{lemma:compatibleBijections}}
                (c2) -- node[midway,above,anchor=base,yshift=1.5mm]{\footnotesize Lemma~\ref{lemma:FamilyOfIdenticalSlices}}
                (c3);
            \draw[FArrow2] (c0) edge[out=0,in=135] (c0m);
            \draw [FLine2, semithick,
                    decoration={
                    brace,
                    mirror,
                    raise=0.2cm
                    },
                    decorate
                ] (c1) -- (c3)
            node [pos=0.5,anchor=north,yshift=-0.25cm] {\footnotesize Lemma~\ref{lemma:nonGenericMultipartiteGraphsFormWQO}};

            \draw[FLine2] (s) edge[out=0,in=90] (d0);
            \draw[FArrow2] (d0) edge[out=270,in=180] (d1);
        \end{pgfonlayer}
        \node[nminus, anchor=base] (l1) at (0,-1.15) {\cminus};
        \node[anchor=west] at (l1.east) {-- $\G$ embeds arbitrarily long paths};
        \node[nplus,anchor=base] (l2) at (3,-1.15) {\cplus};
        \node[anchor=west] at (l2.east) {-- $\G$ admits WQO};
        \node[anchor=west,inner sep=0] at (s.west |- label1.west) {\footnotesize Lemma~\ref{lemma:withBranches}};
    \end{tikzpicture}
\end{center}
\vspace{-2mm}
\noindent
After stating and proving Lemma~\ref{lemma:withBranches} we proceed with the proofs of Cases A), B) and C).
%
Case A) constitutes the most difficult part of the proof and involves a complex and delicate analysis of consequences of the amalgamation property.
It consists of four steps that deduce extension of the assumed induced substructures by individual vertices (cf.~Cor.~\ref{corollary:addingVertices}), 
individual edges (cf.~Cor.~\ref{corollary:canAddP2xA}), paths of length 2 (cf.~Lemma~\ref{lemma:addingPPP}), resp., 
culminating in derivation of arbitrarily long paths (cf.~Lemma~\ref{lemma:longPaths}).
Thus in case A) only the second condition of Theorem~\ref{thm:core} is possible, while in the other two cases
both conditions of Theorem~\ref{thm:core} may hold true.




\newcommand{\lemmaWithBranchesI}{
\begin{tikzpicture}[scale=1*\scaleOfPictures]
    \begin{scope}[every node/.append style=Node]
        \foreach \angle in {0, 60, ..., 299} {
        \path (0, 1) node (N\angle) at +(\angle:1cm) {}; }
    \end{scope}
    \path (0, 1) node (dots) at +(-50:1) {\dots};
    \begin{pgfonlayer}{background}
        \foreach \angleI in {0, 60, ..., 299} {
        \foreach \angleII in {0, 60, ..., 299} {
        \lEdge[Edge, cC]{N\angleI}{}{N\angleII} } }
        \lEdge[Edge, draw=cC, swap]{N120}{\c}{N180}
    \end{pgfonlayer}
\end{tikzpicture}
}
\newcommand{\lemmaWithBranchesII}{
\begin{tikzpicture}[scale=1*\scaleOfPictures]
    \begin{scope}[every node/.append style=Node]
        \node at (0, 0) (M1) {};
        \node at (1, 0) (M2) {};
        \node at (0.5, .866) (M3) {};
    \end{scope}
    \begin{pgfonlayer}{background}
        \lEdge[Edge, draw=cA, swap]{M1}{\a}{M2}
        \lEdge[Edge, draw=cX, swap]{M2}{\x}{M3}
        \lEdge[Edge, draw=cC]{M1}{\c}{M3}
    \end{pgfonlayer}
\end{tikzpicture}
}
\newcommand{\lemmaWithBranchesIII}{
\begin{tikzpicture}[scale=1*\scaleOfPictures]
    \begin{scope}[every node/.append style=Node]
        \node at (0, 0) (K1) {};
        \node at (1, 0) (K2) {};
        \node at (0.5, .866) (K3) {};
    \end{scope}
    \begin{pgfonlayer}{background}
        \lEdge[Edge, draw=cA, swap]{K1}{\a}{K2}
        \lEdge[Edge, draw=cC, swap]{K2}{\c}{K3}
        \lEdge[Edge, draw=cC]{K1}{\c}{K3}
    \end{pgfonlayer}
\end{tikzpicture}
}

\begin{lemma}\label{lemma:withBranches}
    Every homogeneous 3-graph $\G = (V, C_1, C_2, C_3)$ satisfies one of the following conditions:
    \nopagebreak[4]
    \begin{enumerate}
        \item[A)] for some color $\c\in\Colors$, $\G$ contains the following induced substructures:
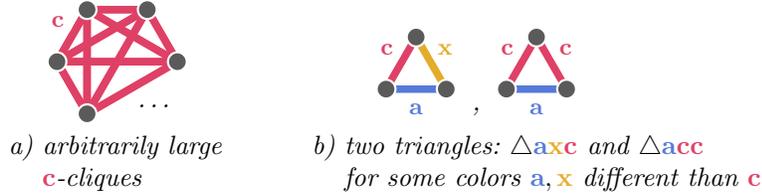
\vspace{1mm}

        \begin{tikzpicture}[scale=\scaleOfPictures]
            \node[anchor=south, inner sep=0] at (0, 0) {\lemmaWithBranchesI};
            \node[anchor=south, inner sep=0] at (5, 0) {\lemmaWithBranchesII};
            \node[anchor=south, ] at (6, 0) {,};
            \node[anchor=south, inner sep=0] at (7, 0) {\lemmaWithBranchesIII};
            \node[anchor=north, inner xsep=0, align=left] at (0, 0)
                {a) arbitrarily large \\ \phantom{a) }$\c$-cliques};
            \node[anchor=north, inner xsep=0, align=left] at (7, 0)
                {b) two triangles: $\triangle\a\x\c$ and $\triangle\a\c\c$ \\  \phantom{b) }for some colors $\a, \x$ different than $\c$};
        \end{tikzpicture}

        \item[B)] for some colors $\x \neq \y$, $(V, \x \cup \y)$ is a union of disjoint cliques,
        \item[C)] for some color $\x$, $(V, \x)$ is a union of finitely many disjoint infinite cliques,
        \item[D)] for some colors $\x \neq \y$, $(V, \x \cup \y)$ contains arbitrarily long paths.
    \end{enumerate}
\end{lemma}

\begin{proof} 
    By Ramsey theorem, $\G$ contains an arbitrarily large monochromatic cliques.
    Let us state a bit stronger requirement:

    \smallskip

    \noindent
    \textbf{Condition $\spadesuit$} \ 
    For some $\a, \c \in \Colors$, $\G$ contains arbitrarily large $\c$-cliques and
    a triangle $\triangle\a\c\c$ with exactly two $\c$-edges ($\a \neq \c$).

    \noindent
    Consider two cases, depending on whether the condition $\spadesuit$ is satisfied or not.

    \medskip

    \noindent
    \textbf{Case $1^\circ$}
    Assume that $\G$ contains both arbitrarily large $\c$-cliques and a triangle $\triangle\a\c\c$
    for some $\a, \c \in \Colors$. Let $\b$ be the third, remaining color.
    Our goal will be to show that either A) or B) holds.

    If the graph $(V, \a \cup \b)$ is a disjoint sum of cliques, we immediately obtain B).
    Suppose the contrary.
    We get that $\G$ has to contain one of the three possible counterexamples for transitivity
    of relation $\a \cup \b$:
    
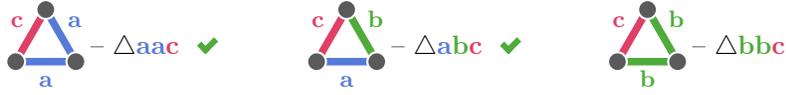
\vspace{-1mm}
    \begin{center}
        \begin{tikzpicture}[scale=\scaleOfPictures]
            \node[anchor=east, inner sep=0] at ( 0, 0) {\labeledTriangle{cA}{\a}{cA}{\a}{cC}{\c}};
            \node[anchor=east, inner sep=0] at ( 5, 0) {\labeledTriangle{cA}{\a}{cB}{\b}{cC}{\c}};
            \node[anchor=east, inner sep=0] at (10, 0) {\labeledTriangle{cB}{\b}{cB}{\b}{cC}{\c}};
            \node[anchor=west, inner sep=0] at ( 0, 0) {-- $\triangle\a\a\c$ \CaseOk};
            \node[anchor=west, inner sep=0] at ( 5, 0) {-- $\triangle\a\b\c$ \CaseOk};
            \node[anchor=west, inner sep=0] at (10, 0) {-- $\triangle\b\b\c$ \phantom{\CaseOk}};
        \end{tikzpicture}
    \end{center}
    \vspace{-3mm}
    If it contains the triangle $\triangle\a\a\c$ or $\triangle\a\b\c$, case A) holds.

    Suppose we got $\triangle\b\b\c$. Let us check this time
    whether colors $\a$ and $\c$ form a union of disjoint cliques.
    Again, if it is so, we easily get B), so we assume the contrary. Similarly, we necessarily obtain one
    of the following triangles:
    
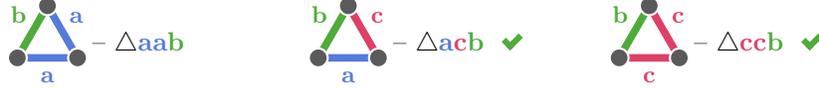
\vspace{-1mm}
    \begin{center}
        \begin{tikzpicture}[scale=\scaleOfPictures]
            \node[anchor=east, inner sep=0] at ( 0, 0) {\labeledTriangle{cA}{\a}{cA}{\a}{cB}{\b}};
            \node[anchor=east, inner sep=0] at ( 5, 0) {\labeledTriangle{cA}{\a}{cC}{\c}{cB}{\b}};
            \node[anchor=east, inner sep=0] at (10, 0) {\labeledTriangle{cC}{\c}{cC}{\c}{cB}{\b}};
            \node[anchor=west, inner sep=0] at ( 0, 0) {-- $\triangle\a\a\b$};
            \node[anchor=west, inner sep=0] at ( 5, 0) {-- $\triangle\a\c\b$ \CaseOk};
            \node[anchor=west, inner sep=0] at (10, 0) {-- $\triangle\c\c\b$ \CaseOk};
        \end{tikzpicture}
    \end{center}
    \vspace{-3mm}
    This time case A) also holds for two out of the three triangles above:
    \begin{itemize}
        \item for $\triangle\a\c\b$, because together with subgraphs resulting from
        assumption $\spadesuit$ (i.e. with triangle $\triangle\a\c\c$ and the $\c$-cliques)
        we get all graphs required by A).
        \item for $\triangle\c\c\b$ paired with the triangle $\triangle\b\b\c$ we just obtained,
        using color $\b$ appearing in those triangles in place of $\a$ in condition A).  
    \end{itemize}
    It only remains to consider the situation when we got $\triangle\a\a\b$. We use it together with
    previously obtained triangle $\triangle\b\b\c$ to build the following instance of singleton amalgamation:

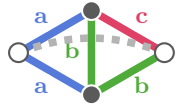
\vspace{-4mm}
    \begin{center}
        \begin{tikzpicture}[scale=\scaleOfPictures]
            \node (N1) at (0, 0) {
            \begin{tikzpicture}[scale=1.4*\scaleOfPictures]
                \path (-1,-0.2) rectangle (1, 1.2);
                \node[ Node] at (0, 0) (n1) {};
                \node[ Node] at (0, 1) (n2) {};
                \node[ANode] at (-0.866, 0.5) (n3) {};
                \node[ANode] at ( 0.866, 0.5) (n4) {};

                \begin{pgfonlayer}{background}
                    \lEdge[cB]{n1}{\b}{n2}
                    \lEdge[cA]{n1}{\a}{n3}
                    \lEdge[cA, swap]{n2}{\a}{n3}
                    \lEdge[cB, swap]{n1}{\b}{n4}
                    \lEdge[cC]{n2}{\c}{n4}
                    \lEdge[AEdge, bend left=20]{n3}{}{n4}
                \end{pgfonlayer}
            \end{tikzpicture}
            };
        \end{tikzpicture}
    \end{center}
\vspace{-3mm}
\noindent
    Depending on the color of the dashed edge, in the solution we get
    one of the following triangles:
    
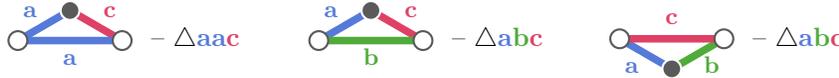
\vspace{-2mm}
    \begin{center}
        \begin{tikzpicture}[scale=\scaleOfPictures]
            \node[anchor=east] (N2) at (0,0) {
            \begin{tikzpicture}[scale=1*\scaleOfPictures]
                \path (-1,-0.2) rectangle (1, 1.2);
                \node[ Node] at (0, 1) (n1) {};
                \node[ANode] at (-0.866, 0.5) (n3) {};
                \node[ANode] at ( 0.866, 0.5) (n4) {};

                \begin{pgfonlayer}{background}
                    \lEdge[cA, swap]{n1}{\a}{n3}
                    \lEdge[cC]{n1}{\c}{n4}
                    \lEdge[cA, swap]{n3}{\a}{n4}
                \end{pgfonlayer}
            \end{tikzpicture}
            };
            \node[anchor=east] (N3) at (5,0) {
            \begin{tikzpicture}[scale=1*\scaleOfPictures]
                \path (-1,-0.2) rectangle (1, 1.2);
                \node[ Node] at (0, 1) (n2) {};
                \node[ANode] at (-0.866, 0.5) (n3) {};
                \node[ANode] at ( 0.866, 0.5) (n4) {};

                \begin{pgfonlayer}{background}
                    \lEdge[cA, swap]{n2}{\a}{n3}
                    \lEdge[cC]{n2}{\c}{n4}
                    \lEdge[cB, swap]{n3}{\b}{n4}
                \end{pgfonlayer}
            \end{tikzpicture}
            };
            \node[anchor=east] (N4) at (10,0) {
            \begin{tikzpicture}[scale=1*\scaleOfPictures]
                \path (-1,-0.2) rectangle (1, 1.2);
                \node[ Node] at (0, 0) (n1) {};
                \node[ANode] at (-0.866, 0.5) (n3) {};
                \node[ANode] at ( 0.866, 0.5) (n4) {};

                \begin{pgfonlayer}{background}
                    \lEdge[cA]{n1}{\a}{n3}
                    \lEdge[cB, swap]{n1}{\b}{n4}
                    \lEdge[cC]{n3}{\c}{n4}
                \end{pgfonlayer}
            \end{tikzpicture}
            };
            \node[anchor=west] at ( 0, 0) {-- $\triangle\a\a\c$};
            \node[anchor=west] at ( 5, 0) {-- $\triangle\a\b\c$};
            \node[anchor=west] at (10, 0) {-- $\triangle\a\b\c$ \phantom{\CaseOk}};
        \end{tikzpicture}
    \end{center}
    \vspace{-3mm}
    and each one alone completes the requirements of A). This closes case $1^\circ$.

    \smallskip
    \noindent
    \textbf{Case $2^\circ$}
    Suppose condition $\spadesuit$ is false. Remind that $\G$ contains arbitrarily large $\c$-cliques
    for some $\c \in \G$. Since $\spadesuit$ does not hold, the graph does not contain a triangle
    $\triangle\c\c\a$ for any $\a\neq\c$; in other words, the color $\c$ appears only within cliques. We conclude
    that $(V, \c)$ is a union of disjoint cliques.
    Clearly at least one of such cliques has to be infinite. By homogeneity we get that all the
    cliques in $(V, \c)$ have to be infinite.
    Now our target is to show that either C) or D) holds.

    The case C) is fulfilled when there are only finitely many $\c$-cliques. Let us assume the contrary.
    In each of the $\c$-cliques we chose one vertex.
    Edges between the chosen vertices form an infinite $\a\b$-clique $K$.
    Using Ramsey theorem again, we conclude that in $K$ one of the colors $\a, \b$ forms
    arbitrarily large monochromatic cliques. W.l.o.g. suppose that this is color $\b$.

    If the graph $\G$ contained $\triangle\y\b\b$ for some $\y \neq \b$, then the assumptions
    of $\spadesuit$ would be met, leading to a contradiction.
    Therefore we conclude that $(V, \b)$
    is a union of disjoint infinite $\b$-cliques.

    When there are only finitely many $\b$-cliques, condition C) is fulfilled. Otherwise
    we know that $\G$ is a union of infinitely many $\x$-cliques for both $\x=\c$ and $\x=\b$.
    Using homogenity, it is easy to show that then every pair of differently colored cliques has \emph{exactly one}
    common vertex, so the graph $\G$ takes the form as depicted in Example~\ref{ex:main}.  
%
%
%
A graph of such form contains arbitrarily long $\b\c$-path, so the requirements of D) are met. 
\end{proof} 
	
\section{Case C) in the proof of Theorem~\ref{thm:core}} \label{sec:WQOtwo}



Let $\c$ be the color that satisfies condition C), and $\a$, $\b$ --- the
remaining two colors. 
In this section we often treat $\G$ as the
$k$-partite graph $(V, \a \cup \b)$ (for some $k\in\mathbb{N}$): $k$ cliques of color $\c$ allow to distinguish
$k$ groups of vertices
$V_1 \cup V_2 \cup \dots \cup V_k = V$ (from now on we will refer to them as layers). The
remaining two colors can be interpreted as existence ($\a$) and nonexistence ($\b$) of edges between
these groups.

\paragraph{Remark $\bigstar$} 
We observe that the special color $\c$ between vertices
within each layer $V_i$ ensures that the automorphisms of $\G$ will not 'mix' those layers:
when two vertices $u, v$ belong to a common layer $V_i$, then their images $f(u), f(v)$ will
also belong to some common layer $V_j$, no matter what automorphism $f \in \aut{\G}$ we choose.
Obviously, the automorphisms can switch positions of whole layers, e.g. move all vertices from $V_i$
to some $V_j$ and vice versa --- in this respect the layers are undistinguishable.

\begin{changemargin}{0cm}{5.7cm}
\begin{lemma}\label{lemma:kPartiteRestrictedIsHomoheneous}
\noindent\tikz[remember picture] \node[inner sep=0, minimum size=0] (marker) {};
For every $i, j \in \{1, 2, \dots, k\}$, 
the bipartite graph $\G_{i,j} = ( V_i \cup
V_j, \a \cap (V_i\cup V_j)^2$, $V_i, V_j)$ (with two distinguishable sides $V_i, V_j$) is homogeneous.
\end{lemma}

\begin{tikzpicture}[overlay, remember picture]
    \node[anchor=north east] at (marker.south west -| current page text area.east)
    {\begin{tikzpicture}[scale=0.75*\scaleOfPictures]
         \foreach \x/\y in {3/0,0/3,3/3} {
         \foreach \d in {1,2,3} {
         \node (x\x-y\y-d\d) [Node, anchor=center,thin,minimum size=0.2cm] at (\x+\d*0.5*1.6-0.5*0.5, \y+\d*0.25*1.6-0.25*0.5) {};
         }
         \node[anchor=center] [rotate=-11, inner sep=0] at (\x+5.5*0.5+0.15, \y+5.5*0.25+0.27) {\reflectbox{$\ddots$}};
         }
         \node[align=center,anchor=center] at (0-0.2, 0+0.2) {\dots\\ \scriptsize remaining \\\scriptsize\raisebox{0.15cm}{$(k-3)$ layers}};
         \node[anchor=center] at (0-0.2, 3+0.2) {$V_1$};
         \node[anchor=center] at (3+2.4, 3+0.2) {$V_2$};
         \node[anchor=center] at (3+2.4, 0+0.2) {$V_3$};
         \begin{pgfonlayer}{background}
             \foreach \d in {1,2,3} {
             \foreach \dd in {1,2,3} {
             \ifthenelse{\d=\dd}{}{
             \draw[Edge, cA, opacity=0.5] (x0-y3-d\d) -- (x3-y0-d\dd);
             \draw[Edge, cA, opacity=0.5] (x0-y3-d\d) -- (x3-y3-d\dd);
             }
             }
             }
         \end{pgfonlayer}
         \begin{pgfonlayer}{background}
             \foreach \d in {1,2,3} {
             \foreach \dd in {1,2,3} {
             \ifthenelse{\d=\dd}{}{
             \draw[Edge, cA] (x3-y3-d\d) -- (x3-y0-d\dd);
             }
             }
             }
         \end{pgfonlayer}
         \draw
         (3+1*0.5-0.5, 0+1*0.25-0.8) --
         (3+8*0.5-0.5, 0+8*0.25-0.8) --
         (3+8*0.5-0.5, 3+8*0.25+0.4) --
         (3+1*0.5-0.5, 3+1*0.25+0.4) -- cycle;
         \node[anchor=center] at (4.5, 5) {$\G_{2,3}$};
    \end{tikzpicture}};
\end{tikzpicture}

\medskip
The vertex sets $V_i$ and $V_j$ 
are used here as unary relations that allow to tell the two layers of $\G_{i,j}$ (sides of $\G_{i,j}$) apart.
An example is shown on the right, with three layers $V_1, V_2$ and $V_3$, and three bipartite graphs
$\G_{1,2}$, $\G_{2,3}$ and $\G_{1,3}$.  
%
\end{changemargin}


\begin{proof}  
Fix $\G_{i,j}$ a bipartite graph.
Note that $\G_{i,j}$  has \emph{distinguishable sides}.
To prove its homogeneity we have to show that each isomorphism of two of its finite induced subgraphs
may be extended to some automorphism of $\G_{i,j}$.
Let us then take some given isomorphism $f: G_1 \rightarrow G_2$ for some finite induced subgraphs $G_1, G_2$
of $\G_{i,j}$. It is easy to extend it to a full automorphism when it 'touches' both
layers of $G_{i,j}$, i.e.:
$$
V(G_1) \cap V_i \neq \emptyset \ \ \wedge\ \  V(G_1) \cap V_j \neq \emptyset
$$
where $V(G_1)$ is the set of vertices of $G_1$. In this case, by homogeneity of $\G$,
we construct a full automorphism $f' : \G \rightarrow \G$, which extends $f$.
It is easy to see that in this case $f'$ has to fix the layers $V_i$ and $V_j$
($f'(V_i) = V_i$ and $f'(V_j) = V_j$), and hence 
$f'$ restricted to the graph $\G_{i,j}$ is a correct automorphism of this graph.

Things get more complicated when $f$ operates only on some single layer of $\G_{i,j}$. W.l.o.g.
suppose that it 'touches' only $V_i$, so $V(G_1) \cap V_j = \emptyset$. Now the above construction
will not work out of the box --- if we were unlucky, the automorphism of $\G$ we get by homogeneity
moves the whole layer $V_j$ to some $V_n$
located 'outside' the graph $\G_{i,j}$ ($n \notin \{i, j\}$).

It will be handy to make the following observation: when $f$ 'touches' only $V_i$, which is an infinite $\c$-clique, we may
assume that $V(G_1) \cap V(G_2) = \emptyset$. Indeed, every function $g: G_1 \rightarrow G_2$
that violates this condition may be decomposed as $g = f_2 \circ f_1$ for some
$f_1, f_2$:
$$
G_1 \xrightarrow{f_1} H \xrightarrow{f_2} G_2
$$
such that $H$ is disjoint both with $G_1$ and with $G_2$.

Now, let $N = |V(G_1)| = |V(G_2)|$ be the size of the domain
of isomorphism $f$. Let us take an arbitrary infinite family $(S_n)_{n \in \mathbb{N}}$ of subgraphs of $\G$ with
disjoint vertex sets, such that the following conditions are met:
\begin{itemize}
    \item $|V(S_n) \cap V_m| = 1$ for $m \neq i$ (and this single vertex will be denoted as $v_m^{(n)}$),
    \item $|V(S_n) \cap V_i| = N$ (denote these vertices as $s^{(n)}_1, s^{(n)}_2, s^{(n)}_3, \dots, s^{(n)}_N$).
\end{itemize}
We define a \emph{connection type} of a layer $V_i$ with $V_m$ in the graph $S_n$ as the $N$-element
sequence of colors of edges from the list bellow:
$$
(\{s^{(n)}_1, v_m^{(n)}\}, \{s^{(n)}_2, v_m^{(n)}\}, \dots, \{s^{(n)}_N, v_m^{(n)}\})
$$
E.g. in the graph bellow, the connection type of layer $V_i = V_3$ with $V_1$ is $\a\b\b\a$,
and with $V_2$ --- $\a\a\b\a$ (remembering that $\b$ is treated as lack of an edge):
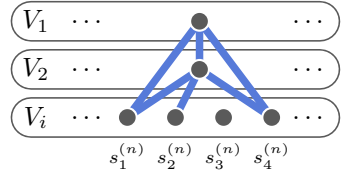
\begin{center}
    \begin{tikzpicture}[scale=0.8*\scaleOfPictures]
        \node (up) [Node] at (1.5, 1) {};
        \node (up2) [Node] at (1.5, 2) {};
        \node (n1) [Node] at (0, 0) {};
        \node (n2) [Node] at (1, 0) {};
        \node (n3) [Node] at (2, 0) {};
        \node (n4) [Node] at (3, 0) {};
        \node [] at (-1.9, 0) {$V_i$};
        \node [] at (-1.9, 1) {$V_2$};
        \node [] at (-1.9, 2) {$V_1$};
        \node [] at (-0.8, 0) {\dots};
        \node [] at (-0.8, 1) {\dots};
        \node [] at (-0.8, 2) {\dots};
        \node [] at (3.8, 0) {\dots};
        \node [] at (3.8, 1) {\dots};
        \node [] at (3.8, 2) {\dots};
        \foreach \x in {1,2,3,4} {
        \node at (\x-1, -0.8) {\scriptsize$s^{(n)}_{\x}$};
        }
        \begin{pgfonlayer}{background}
            \lEdge[cA]{n1}{}{up}
            \lEdge[cA]{n2}{}{up}
            \lEdge[cA]{n4}{}{up}
            \lEdge[cA]{n1}{}{up2}
            \lEdge[cA]{n4}{}{up2}
            \lEdge[cA]{up}{}{up2}
        \end{pgfonlayer}
        \begin{pgfonlayer}{prebackground}
            \draw[Selection] (-2, 0) -- (4, 0);
            \draw[Selection] (-2, 1) -- (4, 1);
            \draw[Selection] (-2, 2) -- (4, 2);
        \end{pgfonlayer}
    \end{tikzpicture}
\end{center}
\vspace{-3mm}
Furthermore, we define the type of graph $S_n$ to be the sequence of types arising between $V_i$ and other layers
plus the list of edge-colors between all pairs of vertices $v_\bullet^{(n)}$ (enumerated in some consistent way).
As there are only finitely many such types, by pigeonhole principle there
exists a pair of graphs $S_a$ and $S_b$ with the same type.

Let us fix some order on vertices of $G_1$: $V(G_1) = \{g_1, g_2, \dots, g_N\}$.
Let $h$ be the partial isomorphism that moves the vertices as follows:
\begin{align*}
    s^{(a)}_1 &\rightarrow g_1 & s^{(b)}_1 &\rightarrow f(g_1)\\
    &\dots &  &\dots\\
    s^{(a)}_N &\rightarrow g_N & s^{(b)}_N &\rightarrow f(g_N)
\end{align*}
By homogeneity, it has to extend to a full automorphism $h'\in \aut{\G}$. In particular, in the
neighbourhood of $G_1$ and $G_2$ there will be images of all vertices $v_\bullet^{(\alpha)}$ of graphs $S_a$ and
$S_b$:
$$
h'\!\left(v_1^{(\alpha)}\right),
h'\!\left(v_2^{(\alpha)}\right), \ \dots\ ,
h'\!\left(v_{i-1}^{(\alpha)}\right),
h'\!\left(v_{i+1}^{(\alpha)}\right), \ \dots\ ,
h'\!\left(v_k^{(\alpha)}\right)
$$
(for $\alpha$ in $\{a, b\}$). What follows is that $G_1$ with added vertices $h'(v_\bullet^{(a)})$ has
the same type as $G_2$ with $h'(v_\bullet^{(b)})$ respectively (that type
may differ from the type of $S_a$ and $S_b$ though!).
It is best illustrated on a picture:

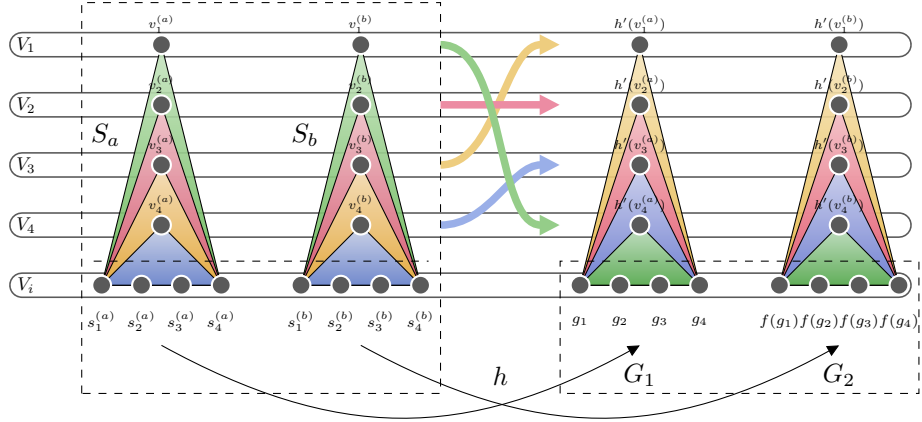
\begin{center}
    \begin{tikzpicture}[xscale=0.53, yscale=0.8]
        \begin{pgfonlayer}{background}
            \foreach \copy in {0, 4.5} {
            \foreach \color/\y in {cB/4,cC/3,cX/2,cA/1} {
            \shade[top color=\color,bottom color=\color,opacity=1]
            (\copy, 0) -- (\copy+3, 0) -- (\copy+1.5, \y) -- cycle;
            \draw[Edge, draw=\color] (\copy, 0) -- (\copy+3, 0) -- (\copy+1.5, \y) -- cycle;
            }
            }
            \foreach \copy in {13.5, 18} {
            \foreach \color/\y in {cX/4,cC/3,cA/2,cB/1} {
            \shade[top color=\color,bottom color=\color,opacity=1] (\copy, 0) -- (\copy+3, 0) -- (\copy+1.5, \y) -- cycle;
            \draw[Edge, draw=\color] (\copy, 0) -- (\copy+3, 0) -- (\copy+1.5, \y) -- cycle;
            }
            }
            \draw[FArrow2,cA,out=0, in = 180] (9.2,1) to (11.8,2);
            \draw[FArrow2,cX,out=0, in = 180] (9.2,2) to (11.8,4);
            \draw[FArrow2,cC,out=0, in = 180] (9.2,3) to (11.8,3);
            \draw[FArrow2,cB,out=0, in = 180] (9.2,4) to (11.8,1);
        \end{pgfonlayer}
        \foreach \copy/\oldname in {0/0, 4.5/5, 13.5/13, 18/18} {
        \foreach \x in {0,1,2,3} {
        \node (c\oldname-n\x) [Node] at (\copy+\x, 0) {};
        }
        \foreach \y in {1,2,...,4} {
        \node (c\oldname-up\y) [Node] at (\copy+1.5, \y) {};
        }
        }
        \foreach \copy/\id in {0/a, 4.5/b} {
        \foreach \y in {1,2,...,4} {
        \node (current) [Basic, scale=1,anchor=center] at (\copy+2+\y*0.375, 5-\y) {\footnotesize $v_{\y}^{(\id)}$};
        \begin{pgfonlayer}{background}
        \draw ($(current.west) + (-0.1,0)$) -- (\copy+1.5, 5-\y);
        \end{pgfonlayer}
        }
        }
        \foreach \copy/\id in {13.5/a, 18/b} {
        \foreach \y/\idx in {1/3,2/2,3/4,4/1} {
        \node (current) [Basic, scale=1,anchor=center] at (\copy+0.7-\y*0.375, 5-\y) {\footnotesize $h'\left(v_{\idx}^{(\id)}\right)$};
        \begin{pgfonlayer}{background}
        \draw ($(current.east) + (0.1,0)$) -- (\copy+1.5, 5-\y);
        \end{pgfonlayer}
        }
        }
        \foreach \x/\d in {1/1,2/0,3/1,4/0} {
        	\node[scale=1] (current1) at (\x   -1, -0.6-\d*0.4) {\footnotesize $s^{(a)}_{\x}$};
        	\node[scale=1] (current2) at (\x+4.5 -1, -0.6-\d*0.4) {\footnotesize $s^{(b)}_{\x}$};
        	\node[scale=1] (current3) at (\x+13.5-1, -0.6-\d*0.4) {\footnotesize $g_{\x}$};
        	\node[scale=1] (current4) at (\x+18-1, -0.6-\d*0.4) {\footnotesize $f(g_{\x})$};
        	\begin{pgfonlayer}{background}
        		\draw (current1.north) -- (\x   -1, 0);
        		\draw (current2.north) -- (\x+4.5   -1, 0);
        		\draw (current3.north) -- (\x+13.5   -1, 0);
        		\draw (current4.north) -- (\x+18   -1, 0);
        	\end{pgfonlayer}
        }
        \node [scale=1] at (-0.9, 4) {$V_1$};
        \node [scale=1] at (-0.9, 3) {$V_2$};
        \node [scale=1] at (-0.9, 2) {$V_3$};
        \node [scale=1] at (-0.9, 1) {$V_4$};
        \node [scale=1] at (-0.9, 0) {$V_i$};

        \node[] at (0+1.5, 4.7) {$S_a$};
        \node[] at (4.5+1.5, 4.7) {$S_b$};
        \node[] at (10.5, -1.8) {$h$};

        \node[] at (1.5+13, -1.9) {$G_1$};
        \node[] at (1.5+18, -1.9) {$G_2$};

        \draw[FArrow2, transform canvas={yshift=-1.0cm}] ($(c0-n0)!0.5!(c0-n3)$) to[bend right=18] ($(c13-n0)!0.5!(c13-n3)$);
        \draw[FArrow2, transform canvas={yshift=-1.0cm}] ($(c5-n0)!0.5!(c5-n3)$) to[bend right=18] ($(c18-n0)!0.5!(c18-n3)$);

        \begin{pgfonlayer}{prebackground}
            \draw[Selection,draw=cLightGray, double distance=0.5cm] (-0, 0) -- (21, 0);
            \draw[Selection,draw=cLightGray, double distance=0.5cm] (-0, 1) -- (21, 1);
            \draw[Selection,draw=cLightGray, double distance=0.5cm] (-0, 2) -- (21, 2);
            \draw[Selection,draw=cLightGray, double distance=0.5cm] (-0, 3) -- (21, 3);
            \draw[Selection,draw=cLightGray, double distance=0.5cm] (-0, 4) -- (21, 4);
        	\fill[cGray, fill=cLightGray, opacity=0.3, rounded corners=0.1cm] (0-0.5, -2.2) rectangle (5+3.5, 0.5);
        	\fill[cGray, fill=cLightGray, opacity=0.3, rounded corners=0.1cm] (13-0.5, -2.2) rectangle (18+3.5, 0.5);
        \end{pgfonlayer}
    \end{tikzpicture}
\end{center}
\vspace{2mm}
Above, the colored triangles represent the types of connections. The order of those types
may get permuted when applying $h'$, but still --- in line with the remark $\bigstar$ --- for each
$\beta \in \{1,2,\dots, k\} \setminus \{i\}$
the vertex $h'\!\left(v_\beta^{(a)}\right)$ must stay in the same layer
as $h'\!\left(v_\beta^{(b)}\right)$, furthermore their type of connection with layer $V_i$ is preserved.

Extending the isomorphism $f$ in a natural way (thanks to the compatibility of types) on those newly
obtained vertices:
$$
h'\left(v_\bullet^{(a)}\right) \xrightarrow{\hspace{0.5cm}f\hspace{0.5cm}} h'\left(v_\bullet^{(b)}\right)
$$
we get an isomorphism that this time 'operates' on all layers $V_\bullet$. If we now extend it to
an automorphism of the whole $\G$, we will get a function that fixes all layers $V_\bullet$. This function
may be safely restricted to $V_i \cup V_j$, staying a correct automorphism of our initial bipartite
graph $\G_{i,j}$, which completes the proof.
\end{proof}
We are going to apply to graphs $\G_{i,j}$ the following classification result:
%
\begin{theorem}[\cite{multipartite12}]\label{theorem:bipartiteClassification}
A countably infinite homogeneous bipartite graph (with distinguishable sides) is either empty, or full, or a perfect matching, or 
the complement of a perfect matching, or a \emph{universal} graph.
\end{theorem}
From our point of view, all we need to know about the universal graph is that it
contains arbitrarily long paths which -- translated to our notation -- would mean that $\G_{i,j}$ contains
arbitrarily long $\a$-paths. Therefore in our further considerations we assume that
$\G_{i,j}$ is not universal which, in our notation, leaves two types of $\G_{i,j}$:
\begin{enumerate}
    \item all edges of $\G_{i, j}$ have the same color $\x \in \{ \a, \b\}$, i.e. $\G_{i,j}$
    is a full or empty bipartite graph,
    \item one of the colors $\x \in \{\a, \b\}$ forms a perfect matching in $\G_{i,j}$, the second
    one ($\y \neq \x$) is then the complement of this matching.
\end{enumerate}
%

%
%
Graphs of type 2. 
may be seen as bijections between their sets of vertices (layers).
Lemma~\ref{lemma:compatibleBijections} states that those bijections
have to preserve other graphs, up to complement. 
%

\begin{lemma}\label{lemma:compatibleBijections}
%
Let $V_i, V_j, V_k$ be some arbitrary pairwise
different layers, such that $\G_{i, j}$ is of type 2 and $\psi : V_i \rightarrow V_j$ is
the bijection it determines. 
Then $\psi$ takes $\a \cap(V_i \cup V_k)^2$ to $\a\cap (V_j \cup V_k)^2$, or to its complement. Formally:
\vspace{-1mm}
\begin{align*}
    & \left(\forall_{u \in V_i} \forall_{v \in V_k}
    \underbrace{u \mathrel{\a} v}_{\clubsuit}
    \Leftrightarrow
    \underbrace{\psi(u) \mathrel{\a} v}_{\spadesuit} \right)
    \vee \left(\forall_{u \in V_i} \forall_{v \in V_k}
    \underbrace{\neg u \mathrel{\a} v}_{\heartsuit}
    \Leftrightarrow
    \underbrace{\psi(u) \mathrel{\a} v}_{\diamondsuit} \right)
\end{align*}
\end{lemma}

\begin{proof}  
We head towards a contradiction. Negating the claim we get:  
\begin{align*}
    \Big( \exists_{u \in V_i} \exists_{v \in V_k}
    \neg \clubsuit \wedge \spadesuit
    \ \ \vee\ \
                \clubsuit \wedge \neg \spadesuit
    \Big)
    \ \wedge \ 
    \Big( \exists_{u \in V_i} \exists_{v \in V_k}
    \neg \heartsuit \wedge \diamondsuit
    \ \ \vee\ \
                \heartsuit \wedge \neg \diamondsuit
    \Big)
\end{align*}
which leads to four cases with similar proofs. We will consider one of them
(corresponding to $\neg\heartsuit \wedge \diamondsuit$ and $\clubsuit \wedge \neg \spadesuit$) and omit the other.
Let us then assume that there exist $x, x' \in V_i$ and $y, y' \in V_k$ such that:
$$
x\mathrel{\a}y
\ \ \wedge\ \
        x'\mathrel{\a}y'
\ \ \wedge\ \
        \psi(x)\mathrel{\a}y
\ \ \wedge\ \
        \neg \psi(x')\mathrel{\a}y'.
$$
Let $g$ be a partial isomorphism of the form $g = \{ x \rightarrow x', y \rightarrow y'\}$. By
homogeneity of $\G$, there is some full automorphism $g' \in \aut{\G}$ extending $g$.
If additionally we were able to force $g$ to fix the layer $V_j$, we would be almost done.
Let us try to achieve that property.

For that purpose, in $V_j$ we choose a vertex $v$ such that:
\begin{enumerate}
    \item[I.] $v \notin \psi(\{x, x'\})$,
    \item[II.] if $\G_{j,k}$ is a graph of type 2. defining a bijection
    $\phi:V_k \rightarrow V_j$, then also $v \notin \phi(\{y, y'\})$.
\end{enumerate}
Clearly such vertex must exist -- two above conditions  exclude at most 4 different vertices
from the infinite set of candidates.
The function $g$ extended with $v \xrightarrow{g} v$ stays a correct isomorphism, because:
\begin{itemize}

    \item in $\G_{i,j}$ by definition of isomorphism we need the edges $\{x,
    v\}$ and $\{g(x), g(v)\}$ to be equally colored, and, in fact, they are.
    We get this thanks to the condition I.: $x$ is connected with all vertices from
    $V_j \setminus \{\psi(x)\}$ by $\x$-edges, $\x \in
    \{\a, \b\}$. We similarly handle $x'$.

    \item in turn in $\G_{j,k}$ --- if it is a graph of type 1., the needed equality of colors of
    edges $\{y, v\}$ and $\{g(y), g(v)\}$ trivially holds.
    If it is a graph of type 2., the equality of colors is derived similarly as in $\G_{i,j}$,
    using the condition II.

\end{itemize}
Presence of the vertex $v$ ensures that layer $V_j$ is preserved by
the full automorphism $g' \in \aut{\G}$ we get by
homogeneity.

Since $\G_{i,j}$ is of type 2., the vertex $\psi(x')$ is the only possible choice
for the image of $\psi(x)$ under $g'$ --- this is the only vertex $x'$ is connected to
by an appropriately colored edge. Because $g'$ is an automorphism, we get that $\psi(x') \mathrel{\a} y'$,
which leads us to the contradiction.
%
\end{proof}

\noindent
From the lemma we have just proved one easily derives the following corollary:

\begin{corollary}\label{corollary:bipartiteAreCompatible}
The following relation $\equiv$ on layers is transitive: 
$$
V_i \equiv V_j \Leftrightarrow \text{ the graph $\G_{i,j}$ is of type 2.}
$$
Furthermore, if $V_i \equiv V_j$ and $V_j \equiv V_k$ then
$f_{j,k}\circ f_{i,j} = f_{i,k}$, where $f_{i,j}, f_{i,k}, f_{j,k}$ are the bijections
determined by graphs $\G_{i,j}, \G_{i,k}$ and $\G_{j,k}$.

\end{corollary}
In Lemma~\ref{lemma:nonGenericMultipartiteGraphsFormWQO} below, which is the last step of the proof of case C), we will apply the following fact (a special case thereof is shown in Theorem 2 in~\cite{KL11}):
\begin{lemma}\label{lemma:FamilyOfIdenticalSlices}
Consider a homogeneous 3-graph $\G$ and a partition of its vertex set 
$
V \ = \ \bigcup_{n \in \nat} U_n
$
into sets $U_\bullet$ of equal finite cardinality.
Suppose further that for every $n\in\nat$, there is an automorphism $\pi_n$ of $\G$ that swaps 
    $U_0$ with $U_n$ and is identity elsewhere.
Then $\G$ admits \wqo.
\end{lemma}

\begin{proof}   
Let $\G = (V, \a, \b, \c)$ be a 3-graph.
Define for $u \in U_0$ the sets $V_u \subseteq V$, which we call \emph{layers}: 
$$
V_u = \setof{\pi_n(u)}{n\in\nat}.
$$
We will prove that the structure 
$\G' = (V, \a, \b, \c, (V_u)_{u\in U_0})$ admits \wqo. 
This will imply that $\G$ admits \wqo as well; indeed,
compared to $\G$, structure $\G'$ is equipped with
additional unary relations $V_\bullet$, which only makes the order $\embedsin$ in $\age{\G'}$ finer than the analogous order in $\age{\G}$.

Let $G_n$ denote the induced substructure of $\G'$ on vertex set $U_n$. 
By the assumptions, for every $n,m \in\nat$ there is a swap of $U_n$ and $U_m$ that, extended with identity elsewhere,
is an automorphism of $\G'$. 
In consequence, all structures $G_\bullet$ are isomorphic, and the embedding order $\embedsin$
of induced substructures of $\G'$ is isomorphic to finite multisets over $\age{G_0}$, ordered by multiset inclusion.
Thus $(\age{\G'}, \embedsin)$ is isomorphic to the multiset inclusion in $\multisets{\age{G_0}}$, which is a \wqo as
$U_0$ is finite.
For any \wqo $(X, \leq)$, analogous isomorphism holds between
the lifted embedding order $(\age{\G'}, \embedsinpar{X})$ and  the multiset inclusion in multisets over 
induced substructures of $G_0$ labeled by elements of $X$, and again the latter order is a \wqo.
Thus $\G'$ admits \wqo.
\end{proof}
\begin{lemma}\label{lemma:nonGenericMultipartiteGraphsFormWQO}
The 3-graph $\G$ 
admits \wqo.
\end{lemma}

\begin{proof}  
We are going to prepare the ground for the use of Lemma~\ref{lemma:FamilyOfIdenticalSlices}. 
By Corollary~\ref{corollary:bipartiteAreCompatible}. the vertex set $V$ partitions into
$
V = \bigcup_{n \in \mathbb{N}} U_n
$ so that
\begin{enumerate}
    \item[a)] every layer $V_i$ shares with every set $U_n$ exactly one vertex: 
    $U_n \cap V_i = \{ v^{(n)}_i \}$,
    \item[b)] if $f_{i,j}$ is the bijection determined by $\G_{i,j}$ (a graph of type 2.), then
    $f_{i,j}(v^{(n)}_i) \in U_n$, so all the bijections preserve every set $U_\bullet$.
\end{enumerate}
Intuitively, $\G$ can by cut into thin 'slices' perpendicular to the layers
$V_\bullet$. By thin we mean that the slices have exactly one vertex in each layer.  
The cut is made along the bijections dictated by the graphs of type 2. 
as in the picture bellow:

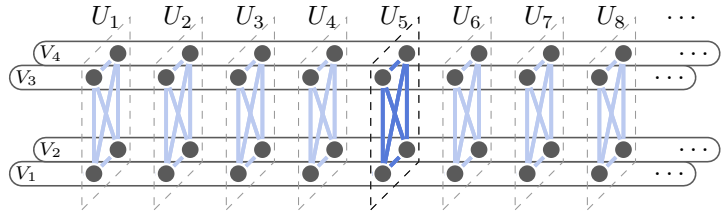
\begin{center}
    \begin{tikzpicture}[scale=0.4*\scaleOfPictures]
        \foreach \x/\y in {0/0,1/1,0/4,1/5} {
        \foreach \d in {1,2,...,8} {
        \node (y\y-d\d) [Node, anchor=center] at (\x+\d*3, \y) {};
        }
        \node [inner sep=0] at (\x+9*3, \y) {$\dots$};
        \begin{pgfonlayer}{prebackground}
            \draw[Selection, double distance=0.3cm] (\x, \y) -- (\x+9*3+0.5, \y);
        \end{pgfonlayer}
        }
        \node at (0+0.2, 0) {\scriptsize$V_1$};
        \node at (1+0.2, 1) {\scriptsize$V_2$};
        \node at (0+0.2, 4) {\scriptsize$V_3$};
        \node at (1+0.2, 5) {\scriptsize$V_4$};
        \begin{pgfonlayer}{background}
            \foreach \d in {1,2,3,4,5,6,7,8} {
            \node at (3*\d+0.5, 6.5) {$U_{\d}$};
            \foreach \p in {0,1,4,5} {
            \foreach \q in {0,1,4,5} {
            \draw[ThinEdge, cA!40!white] (y\p-d\d) -- (y\q-d\d);
            }
            }

            \draw[dashed, cBlack!40!white]
            (3*\d-0.5, 0-0.5-1) --
            (3*\d+1.5, 1+0.5-1) --
            (3*\d+1.5, 5+0.5+1) --
            (3*\d-0.5, 4-0.5+1) -- cycle
            ;
            }
            \foreach \p in {0,1,4,5} {
            \foreach \q in {0,1,4,5} {
            \draw[ThinEdge, cA] (y\p-d5) -- (y\q-d5);
            }
            }
            \draw[dashed, cBlack]
            (3*5-0.5, 0-0.5-1) --
            (3*5+1.5, 1+0.5-1) --
            (3*5+1.5, 5+0.5+1) --
            (3*5-0.5, 4-0.5+1) -- cycle
            ;
            \node at (3*9+0.5, 6.5) {$\dots$};
        \end{pgfonlayer}

    \end{tikzpicture}
\end{center}
We observe that for every $n$, the bijection $h_n : V \to V$ that swaps $U_1$ and $U_n$ 
along the only bijection $U_1 \rightarrow U_n$ that preserves layers, and is identity elsewhere, is an automorphism of $\G$.
Indeed, for any three slices $U_a, U_b, U_c$ we have that:
$$
v^{(a)}_i \mathrel{\a} v^{(c)}_j \Leftrightarrow v^{(b)}_i \mathrel{\a} v^{(c)}_j
$$
so the edges $\left\{v^{(a)}_i, v^{(c)}_j\right\}$ and $\left\{v^{(b)}_i, v^{(c)}_j\right\}$
are colored the same way.
The above equivalence is obvious in case when $\G_{i,j}$ is a graph of type 1. In the case of graph of type
2., the vertex $v^{(c)}_i$ is connected with all vertices from $V_j$ but one by $\x$-edges
for some $\x \in \{\a, \b\}$. However, the special vertex $f_{i,j}(v_i^{(c)})$ that is not connected by a $\x$-edge,
by the condition  b), also belongs to $U_c$, so it does not interfere with above equivalence.

By Lemma~\ref{lemma:FamilyOfIdenticalSlices} we deduce that $\G$ admits \wqo,
which completes the proof.
\end{proof}

\section{Case B) in the proof of Theorem~\ref{thm:core}}  \label{sec:coreB}


Let $\a$, $\b$ be the two colors such that the graph $H = (V, \a \cup \b)$ is a sum of disjoint cliques.
The color appearing between the cliques we mark as $\c$.
Since the set of vertices $V$ is infinite, the graph 
$H$ cannot be a finite sum of finite cliques.
Furthermore, by homogeneity we have that all $\a\b$-cliques in $\G$ are isomorphic, so
their sizes are equal. We then have three cases to investigate:
\begin{enumerate}
    \item $H$ is a sum of \emph{infinite} number of \emph{infinite} $\a\b$-cliques,
    \item $H$ is a sum of \emph{\phantom{in}finite} number of \emph{infinite} $\a\b$-cliques,
    \item $H$ is a sum of \emph{infinite} number of \emph{\phantom{in}finite} $\a\b$-cliques.
\end{enumerate}

Let us concentrate on the first case.
Because each $\a\b$-clique $K\trianglelefteq \G$ maximal in terms of relation '$\trianglelefteq$'
is homogeneous, we can apply Theorem~\ref{thm:graphs} 
to deduce that either
$K$ admits \wqo, or it contains arbitrarily long $\x$-paths for some $\x \in \{ \a, \b\}$.
We only need to consider the former case.

The crucial observation is that 
the embedding order on induced substructures of $\G$ is isomorphic to the multiset inclusion in $\multisets{\age{K}}$.
Indeed, any induced substructure $\str{X} \embedsin \G$ splits into the $\a\b$-cliques, and as there are only $\c$-edges between the cliques,
this split of $\str{X}$ determines $\str{X}$ uniquely. Finally, the choice of particular $\a\b$-cliques is irrelevant, as they are all isomorphic.

As the multiset inclusion in $\multisets{\age{K}}$ admits \wqo by assumption, being itself a \wqo in particular, 
we deduce that $(\age{\G}, \embedsin)$ is a \wqo too.
Similarly one observes that the lifted order $\embedsinpar{X}$ is a \wqo, for any underlying \wqo $(X, \leq)$.

The second case, when $H$ is a sum of $k$ infinite $\a\b$-cliques, is dealt analogously with the only difference that
multisets  over $\age{K}$ of size at most $k$ are considered instead of multisets of unbounded size.

Finally the third case, when $H$ is a sum of finite $\a\b$-cliques, follows immediately by Lemma~\ref{lemma:FamilyOfIdenticalSlices}.

\section{Case A) in the proof of Theorem~\ref{thm:core}}  \label{sec:coreA}



This is the most extensive part of the proof. 
Now we assume that case A) 
of Lemma~\ref{lemma:withBranches} holds and analyze the consequences. We are going to present
a chain of lemmas that eventually gives us the existence of arbitrarily long paths in $\G$.

From now on we fix the color $\c$ appearing in case A) of lemma~\ref{lemma:withBranches} and
consider it as the no-edge relation. Consequently, we will treat $\G$ as a 2-edge-colored graph.
For that reason we define $\Colors' = \Colors \setminus \{\c\}$.
In all pictures in this section, the lack of an edge between some two vertices of graph
will mean that they are connected by a $\c$-edge.

Let us introduce a few new notations:
\begin{itemize}
    \item  $\underline{\x\y\z\dots}$ will denote an $\a\b$-path with consecutive edges colored by
    $\x$, $\y$, $\z$, etc. ($\a,\b \in \ColorsTwo$). E.g., $\underline{\a\b\a}$ corresponds to the
    following path:
    \tikz{
    \node[Node, minimum size=0.125cm] at (0,0) (a) {};
    \node[Node, minimum size=0.125cm] at (0.5,0) (b) {};
    \node[Node, minimum size=0.125cm] at (1.0,0) (c) {};
    \node[Node, minimum size=0.125cm] at (1.5,0) (d) {};
    \draw[ThinEdge, draw=cA ] (a) edge node[midway, auto] {\scriptsize$\a$} (b);
    \draw[ThinEdge, draw=cB] (b) edge node[midway, auto] {\scriptsize$\b$} (c);
    \draw[ThinEdge, draw=cA ] (c) edge node[midway, auto] {\scriptsize$\a$} (d);
    }
    The single-vertex path will be written as $\bullet$.

    \item For cycles we will use similar notation: $\circ\x\y\dots\z$ stands for a $\a\b$-cycle with
    consecutive edges painted $\x, \y, \dots, \z$.

    \item For two given graphs $G_1$ and $G_2$, a graph $G_1 + G_2$ is built as follows:
    We take disjoint copies of $G_1$ and $G_2$ and connect the two parts with $\c$-edges.
    E.g., $\underline{\a\a} + \bullet$ denotes the graph:
    ''
    \tikz{
    \node[Node, minimum size=0.125cm] at (0,0) (a) {};
    \node[Node, minimum size=0.125cm] at (0.5,0) (b) {};
    \node[Node, minimum size=0.125cm] at (1.0,0) (c) {};
    \node[Node, minimum size=0.125cm] at (1.5,0) (d) {};
    \draw[ThinEdge, draw=cA ] (a) edge node[midway, auto] {\scriptsize$\a$} (b);
    \draw[ThinEdge, draw=cA] (b) edge node[midway, auto] {\scriptsize$\a$} (c);
    }
    ''.

    \item For a given graph $G_1$, a sum of its $k$ copies (in the above sense) is written as
    $k \cdot G_1$, e.g. $3\cdot\bullet = \bullet+\bullet+\bullet$.

    \item Discrete graph $D_k$ is a graph $k\cdot\bullet$. 

\end{itemize}

\noindent
Now we can reformulate the case A) of Lemma~\ref{lemma:withBranches} using the new convention:
\paragraph{Lemma \ref{lemma:withBranches} (new formulation of case A)}

\medskip
\noindent
$\G$ contains the following induced subgraphs

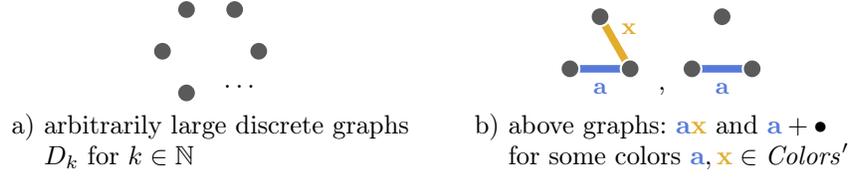
\newcommand{\lemmaWithBranchesIb}{
\begin{tikzpicture}[scale=0.8*\scaleOfPictures]
    \begin{scope}[every node/.append style=Node]
        \foreach \angle in {0, 60, ..., 299} {
        \path (0, 1) node (N\angle) at +(\angle:1cm) {}; }
    \end{scope}
    \path (0, 1) node (dots) at +(-50:1) {\dots};
\end{tikzpicture}
}
\newcommand{\lemmaWithBranchesIIb}{
\begin{tikzpicture}[scale=1*\scaleOfPictures]
    \begin{scope}[every node/.append style=Node]
        \node at (0, 0) (M1) {};
        \node at (1, 0) (M2) {};
        \node at (0.5, .866) (M3) {};
    \end{scope}
    \begin{pgfonlayer}{background}
        \lEdge[Edge, draw=cA, swap]{M1}{\a}{M2}
        \lEdge[Edge, draw=cX, swap]{M2}{\x}{M3}
    \end{pgfonlayer}
\end{tikzpicture}
}
\newcommand{\lemmaWithBranchesIIIb}{
\begin{tikzpicture}[scale=1*\scaleOfPictures]
    \begin{scope}[every node/.append style=Node]
        \node at (0, 0) (K1) {};
        \node at (1, 0) (K2) {};
        \node at (0.5, .866) (K3) {};
    \end{scope}
    \begin{pgfonlayer}{background}
        \lEdge[Edge, draw=cA, swap]{K1}{\a}{K2}
    \end{pgfonlayer}
\end{tikzpicture}
}
\begin{center}
    \begin{tikzpicture}[scale=\scaleOfPictures]
        \node[anchor=south, inner sep=0] at (-1, 0) {\lemmaWithBranchesIb};
        \node[anchor=south, inner sep=0] at (5.5, 0) {\lemmaWithBranchesIIb};
        \node[anchor=south, ] at (6.5, 0) {,};
        \node[anchor=south, inner sep=0] at (7.5, 0) {\lemmaWithBranchesIIIb};
        \node[anchor=north, inner xsep=0, align=left] at (-1, 0)
        {a) arbitrarily large discrete graphs \\ \phantom{a) }$D_k$ for $k \in \mathbb{N}$};
        \node[anchor=north, inner xsep=0, align=left] at (6.5, 0)
        {b) above graphs: $\underline{\a\x}$ and $\underline{\a} + \bullet$ \\  \phantom{b) }for some colors $\a, \x \in \ColorsTwo$};
    \end{tikzpicture}
\end{center}

\newcommand{\circled}[1]{
\tikz[baseline=(char.base)]{
\node[shape=rounded rectangle,draw,inner ysep=2pt,inner xsep=1pt] (char) {\ensuremath{#1}};
}
}
\newcommand{\circledb}[1]{
\tikz[baseline=(char.base)]{
\node[shape=rounded rectangle,draw,inner ysep=2pt,inner xsep=1pt] (char) {#1};
}
}

\newcommand{\kb}{\circledb{\scriptsize\ensuremath{k\cdot\bullet}}}


\subsection{Adding isolated vertices}\label{subsec:addingVertices}

Our first goal is to show that $\G$ embeds a graph $\underline{\a\x} + k\cdot\bullet$
for each $k \in \mathbb{N}$. The proof will be inductive. The induction base follows easily
by the assumed condition A). 
Two coming lemmas, when combined, will form the inductive step. From now on, the expression
$k\cdot\bullet$ will appear many times, so for readability we will emphasize it as \kb.

\begin{lemma}\label{lemma:addOneA}
Let $\G$ be a strongly homogeneous, 2-edge colored graph that embeds arbitrarily large discrete graphs
and also the subgraphs
$\underline{\a\x} + \kb$ and
$\underline{\a} + \bullet + \kb$ for some $\a, \x \in \ColorsTwo$ and $k \in \mathbb{N}$.
Then $\G$ embeds graphs:
\begin{enumerate}
    \item $\underline{\a_2\y} + \kb$,
    \item $\underline{\a_2} + \bullet + \bullet + \kb$
\end{enumerate}
for some colors $\a_2, \y \in \ColorsTwo$.
\end{lemma}
\noindent
(It is important to note that $\a$ does not have to be equal to $\a_2$.)
\begin{lemma}\label{lemma:addOneB}
Let $\G$ be a strongly homogeneous, 2-edge colored graph that embeds graphs
$\underline{\a\y} + \kb$ and
$\underline{\a} + \bullet + \bullet + \kb$ for some $\a, \y \in \ColorsTwo$ and $k \in \mathbb{N}$.
Then $\G$ embeds graph $\underline{\a\z} + \bullet + \kb$ for some color $\z \in \ColorsTwo$.
\end{lemma}

\noindent
Juxtaposition of those lemmas allows us to 'add' arbitrarily many isolated vertices:
\begin{align*}
    \left\{
    \begin{array}{l}
        \a\x + \kb
        \\
        \a + \bullet + \kb
    \end{array}
    \right\}
    \xrightarrow{\text{Lem. \ref{lemma:addOneA}.}}
    \left\{
    \begin{array}{l}
        \a_2\y + \kb
        \\
        \a_2 + \bullet + \bullet + \kb
    \end{array}
    \right\}
    \xrightarrow{\text{Lem. \ref{lemma:addOneB}.}}
    \left\{
    \begin{array}{l}
        \a_2\z + \tikz[baseline=(n1.base)]{
        \node[shape=rounded rectangle, inner ysep=2pt, inner xsep=1pt, draw] (n1) {\ensuremath{\bullet + \kb}};
        }
        \\
        \a_2 + \bullet + \tikz[baseline=(n1.base)]{
        \node[shape=rounded rectangle, inner ysep=2pt, inner xsep=1pt, draw] (n1) {\ensuremath{\bullet + \kb}};
        }
    \end{array}
    \right\}
\end{align*}
Similar scheme will emerge also in subsequent parts of the proof: in analogous way
we will later be adding isolated edges and
two-edge paths.

Now, let us move on to the proof of Lemmas \ref{lemma:addOneA} and \ref{lemma:addOneB}. They
will be the first from a group of lemmas making a heavy use of the amalgamation property.

\newcommand{\addOneAxIxQUESTION}{
\begin{tikzpicture}[scale=\scaleOfPictures]
    \begin{pgfonlayer}{foreground}
        \node[ANode] at (0, 0) (n1) {};
        \node[ Node] at (1, 0) (n2) {};
        \node[ANode] at (0, 1) (n3) {};
        \node[ Node] at (1, 1) (n4) {};
        \node[Basic] at (1.8, 0.5) {+ $\kb$};
    \end{pgfonlayer}
    \begin{pgfonlayer}{background}
        \lEdge[cA]{n1}{\a}{n2}
        \lEdge[AEdge]{n1}{}{n3}
    \end{pgfonlayer}
\end{tikzpicture}
}
\newcommand{\addOneAxIxcaseX}{
\begin{tikzpicture}[scale=\scaleOfPictures]
    \begin{pgfonlayer}{foreground}
        \node[ANode] at (0, 0) (n1) {};
        \node[ Node] at (1, 0) (n2) {};
        \node[ANode] at (0, 1) (n3) {};
        \node[ Node] at (1, 1) (n4) {};
        \node[Basic] at (1.8, 0.5) {+ $\kb$};
    \end{pgfonlayer}
    \begin{pgfonlayer}{background}
        \lEdge[cA]{n1}{\a}{n2}
        \lEdge[cY]{n1}{\y}{n3}
    \end{pgfonlayer}
\end{tikzpicture}
}
\newcommand{\addOneAxIxcaseNO}{
\begin{tikzpicture}[scale=\scaleOfPictures]
    \begin{pgfonlayer}{foreground}
        \node[ANode] at (0, 0) (n1) {};
        \node[ Node] at (1, 0) (n2) {};
        \node[ANode] at (0, 1) (n3) {};
        \node[ Node] at (1, 1) (n4) {};
        \node[Basic] at (1.8, 0.5) {+ $\kb$};
    \end{pgfonlayer}
    \begin{pgfonlayer}{background}
        \lEdge[cA]{n1}{\a}{n2}
    \end{pgfonlayer}
\end{tikzpicture}
}

\begin{proof}[Proof of Lemma~\ref{lemma:addOneA}]
By assumptions we know that $G_1 = \underline{\a\x} + \kb \trianglelefteq \G$ as well as
$G_2 = \underline{\a} + \bullet + \kb \trianglelefteq \G$ for some given colors $\a, \x \in \ColorsTwo$.
The set $\ColorsTwo$ has two elements --- let $\b$ be the second of its elements, different from $\a$.

\paragraph{Current target} To prove the lemma, it suffices to show one of the following statements:
\begin{enumerate}
    \item[\CaseOk 1)] $\G$ embeds a graph $\underline{\a} + \bullet + \bullet + \kb$ \\(paired with
    $G_1$, it will give us the thesis of lemma),
    \item[\CaseOk 2)] $\G$ embeds graphs $\underline{\b} + \bullet + \bullet + \kb$ and $\underline{\b\y} + \kb$
    \\(here $G_1$ would not help, since lemma requires compatibility of edge colors,
    yet $G_1$ may not contain $\b$-edge if $\x = \a$).
\end{enumerate}

\Instance{\ref{lemma:addOneA}.1}
We begin by considering the following amalgamation instance:

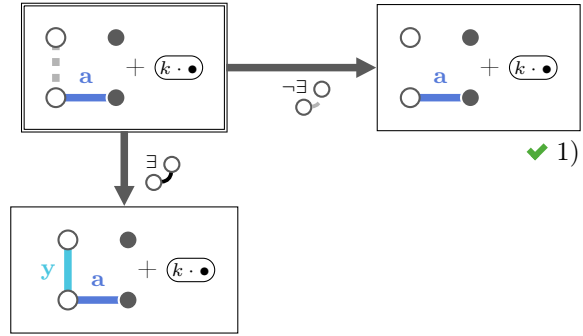
\begin{center}
    \begin{tikzpicture}[scale=\scaleOfPictures]
        \node (Q) [AInstance] at (0, 0) {\addOneAxIxQUESTION};
        \node (PREV) [AResult, below=1cm of Q] {\addOneAxIxcaseX };
        \node (D) [AResult, right=2cm of Q] {\addOneAxIxcaseNO};
        \node[anchor=north east, inner xsep=0] at (D.south east) (ok1) {\CaseOk 1)};
        \draw[CaseArrow]
        (Q)
        -- node[midway, auto,     , inner sep=0cm] {\CaseEdgeExists}
        (PREV);
        \draw[CaseArrow]
        (Q)
        -- node[midway, auto, swap, inner sep=0cm] {\CaseEdgeDoesNotExist}
        (D);
    \end{tikzpicture}
\end{center}
If in its solution the edge is not present, we get graph $\underline{\a} +
\bullet + \bullet + \kb$, so \CaseOk 1) is obtained immediately.
Assume the contrary --- that some $\y$-edge appeared $\y \in \ColorsTwo$.

\Instance{\ref{lemma:addOneA}.2} Using the obtained graph, we build
a new instance:

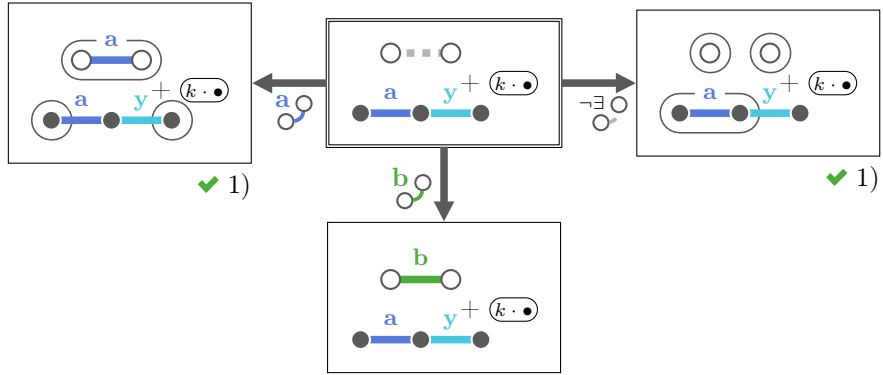
\newcommand{\addOneAxIIxQUESTION}{
\begin{tikzpicture}[scale=\scaleOfPictures]
    \begin{pgfonlayer}{foreground}
        \node[ Node] at (0, 0) (n1) {};
        \node[ Node] at (1, 0) (n2) {};
        \node[ Node] at (2, 0) (n3) {};

        \node[ANode] at (0.5, 1) (n4) {};
        \node[ANode] at (1.5, 1) (n5) {};
        \node[Basic] at (2.3, 0.5) {+ $\kb$};
    \end{pgfonlayer}
    \begin{pgfonlayer}{background}
        \lEdge[cA]{n1}{\a}{n2}
        \lEdge[cY]{n2}{\y}{n3}
        \lEdge[AEdge]{n4}{}{n5}
    \end{pgfonlayer}
\end{tikzpicture}
}
\newcommand{\addOneAxIIxA}{
\begin{tikzpicture}[scale=\scaleOfPictures]
    \begin{pgfonlayer}{foreground}
        \node[ Node] at (0, 0) (n1) {};
        \node[ Node] at (1, 0) (n2) {};
        \node[ Node] at (2, 0) (n3) {};

        \node[ANode] at (0.5, 1) (n4) {};
        \node[ANode] at (1.5, 1) (n5) {};
        \node[Basic] at (2.3, 0.5) {+ $\kb$};
    \end{pgfonlayer}
    \begin{pgfonlayer}{background}
        \lEdge[cA]{n1}{\a}{n2}
        \lEdge[cY]{n2}{\y}{n3}
        \lEdge[cA]{n4}{\a}{n5}
    \end{pgfonlayer}
    \begin{pgfonlayer}{prebackground}
        \draw[Selection] (n1.center) -- (n1.center);
        \draw[Selection] (n3.center) -- (n3.center);
        \draw[Selection] (n4.center) -- (n5.center);
    \end{pgfonlayer}
\end{tikzpicture}
}
\newcommand{\addOneAxIIxNO}{
\begin{tikzpicture}[scale=\scaleOfPictures]
    \begin{pgfonlayer}{foreground}
        \node[ Node] at (0, 0) (n1) {};
        \node[ Node] at (1, 0) (n2) {};
        \node[ Node] at (2, 0) (n3) {};

        \node[ANode] at (0.5, 1) (n4) {};
        \node[ANode] at (1.5, 1) (n5) {};
        \node[Basic] at (2.3, 0.5) {+ $\kb$};
    \end{pgfonlayer}
    \begin{pgfonlayer}{background}
        \lEdge[cA]{n1}{\a}{n2}
        \lEdge[cY]{n2}{\y}{n3}
    \end{pgfonlayer}
    \begin{pgfonlayer}{prebackground}
        \draw[Selection] (n4.center) -- (n4.center);
        \draw[Selection] (n5.center) -- (n5.center);
        \draw[Selection] (n1.center) -- (n2.center);
    \end{pgfonlayer}
\end{tikzpicture}
}
\newcommand{\addOneAxIIxB}{
\begin{tikzpicture}[scale=\scaleOfPictures]
    \begin{pgfonlayer}{foreground}
        \node[ Node] at (0, 0) (n1) {};
        \node[ Node] at (1, 0) (n2) {};
        \node[ Node] at (2, 0) (n3) {};

        \node[ANode] at (0.5, 1) (n4) {};
        \node[ANode] at (1.5, 1) (n5) {};
        \node[Basic] at (2.3, 0.5) {+ $\kb$};
    \end{pgfonlayer}
    \begin{pgfonlayer}{background}
        \lEdge[cA]{n1}{\a}{n2}
        \lEdge[cY]{n2}{\y}{n3}
        \lEdge[cB]{n4}{\b}{n5}
    \end{pgfonlayer}
\end{tikzpicture}
}

\begin{center}
\centering
\makebox[\linewidth]{%
    \begin{tikzpicture}[scale=\scaleOfPictures]
        \node (Q) [AInstance] at (0, 0) {\addOneAxIIxQUESTION};
        \node (B) [AResult, left=1cm of Q] {\addOneAxIIxA};
        \node (D) [AResult, right=1cm of Q] {\addOneAxIIxNO};
        \node (E) [AResult, below=1cm of Q] {\addOneAxIIxB};
        \node[anchor=north east, inner xsep=0] at (B.south east) (ok1) {\CaseOk 1)};
        \node[anchor=north east, inner xsep=0] at (D.south east) (ok2) {\CaseOk 1)};
        \draw[CaseArrow]
        (Q)
        -- node[midway, auto,     , inner sep=0cm] {\CaseAEdge}
        (B);
        \draw[CaseArrow]
        (Q)
        -- node[midway, auto, swap, inner sep=0cm] {\CaseEdgeDoesNotExist}
        (D);
        \draw[CaseArrow]
        (Q)
        -- node[midway, auto, swap, inner sep=0cm] {\CaseBEdge}
        (E);
    \end{tikzpicture}
}
\end{center}

The above instance is one of the few that actually use the strong amalgamation property.
As shown on the picture, in cases when we get an $\a$-edge or we do not get an edge at all,
condition \CaseOk 1). is easily met. Let us assume we obtained a $\b$-edge.

At this point we have to notice that $\underline{\b} + \underline{\a\y} + \kb$ embeds a graph
$\underline{\b} + \bullet + \bullet + \kb$, so from now on to prove \CaseOk 2), it suffices to obtain
$\underline{\a\b} + \kb$.
Hence, if $\y = \b$, we would have the missing graph $\underline{\a\b} + \kb$ as a subgraph
of $\underline{\b} + \underline{\a\y} + \kb$. It then only remains to consider the case $\y = \a$.

For later use, from $\underline{\b} + \underline{\a\a} + \kb$ we take the following subgraph $G_3$:

\newcommand{\addOneAxIIIx}{
\begin{tikzpicture}[scale=\scaleOfPictures]
    \begin{pgfonlayer}{foreground}
        \node at (-2, 0) {$G_3 =$};
        \node[ Node] at (0, 0) (n1) {};
        \node[ Node] at (1, 0) (n2) {};
        \node[ Node] at (2, 0) (n3) {};

        \node[ Node] at (-1, 0) (n4) {};
        \node[Basic] at (2.8, 0) {+ $\kb$};
    \end{pgfonlayer}
    \begin{pgfonlayer}{background}
        \lEdge[cA]{n1}{\a}{n2}
        \lEdge[cA]{n2}{\a}{n3}
    \end{pgfonlayer}
\end{tikzpicture}
}

\begin{center}
    \addOneAxIIIx
\end{center}

\Instance{\ref{lemma:addOneA}.3}
We use it to construct a new instance of amalgamation:

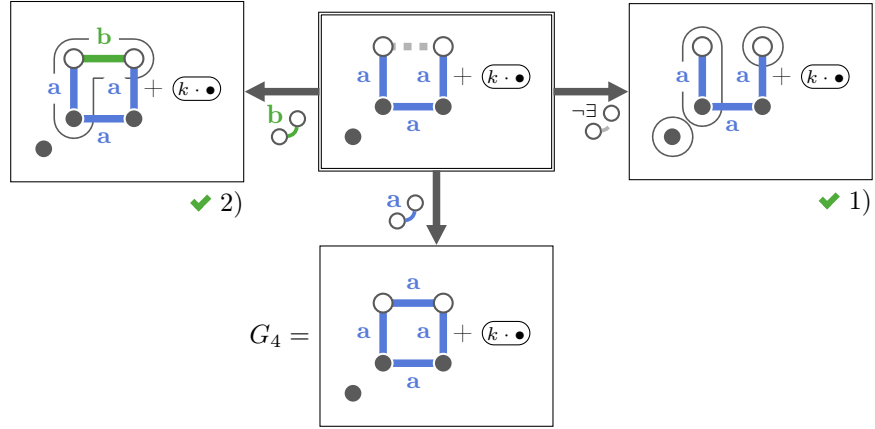
\newcommand{\addOneAxIVxQ}{
\begin{tikzpicture}[scale=\scaleOfPictures]
    \begin{pgfonlayer}{foreground}
        \node[ Node] at (0, 0) (n1) {};
        \node[ Node] at (1, 0) (n2) {};
        \node[ANode] at (0, 1) (n4) {};
        \node[ANode] at (1, 1) (n3) {};

        \node[ Node] at (-0.5, -0.5) (n5) {};
        \node[Basic] at (1.8, 0.5) {+ $\kb$};
    \end{pgfonlayer}
    \begin{pgfonlayer}{background}
        \lEdge[cA,swap]{n1}{\a}{n2}
        \lEdge[cA]{n2}{\a}{n3}
        \lEdge[cA]{n1}{\a}{n4}
        \lEdge[AEdge]{n3}{}{n4}
    \end{pgfonlayer}
\end{tikzpicture}
}

\newcommand{\addOneAxIVxNO}{
\begin{tikzpicture}[scale=\scaleOfPictures]
    \begin{pgfonlayer}{foreground}
        \node[ Node] at (0, 0) (n1) {};
        \node[ Node] at (1, 0) (n2) {};
        \node[ANode] at (0, 1) (n4) {};
        \node[ANode] at (1, 1) (n3) {};

        \node[ Node] at (-0.5, -0.5) (n5) {};
        \node[Basic] at (1.8, 0.5) {+ $\kb$};
    \end{pgfonlayer}
    \begin{pgfonlayer}{background}
        \lEdge[cA,swap]{n1}{\a}{n2}
        \lEdge[cA]{n2}{\a}{n3}
        \lEdge[cA]{n1}{\a}{n4}
    \end{pgfonlayer}
    \begin{pgfonlayer}{prebackground}
        \draw[Selection] (n4.center) -- (n1.center) (n3.center) -- (n3.center) (n5.center) -- (n5.center);
    \end{pgfonlayer}
\end{tikzpicture}
}

\newcommand{\addOneAxIVxB}{
\begin{tikzpicture}[scale=\scaleOfPictures]
    \begin{pgfonlayer}{foreground}
        \node[ Node] at (0, 0) (n1) {};
        \node[ Node] at (1, 0) (n2) {};
        \node[ANode] at (0, 1) (n4) {};
        \node[ANode] at (1, 1) (n3) {};

        \node[ Node] at (-0.5, -0.5) (n5) {};
        \node[Basic] at (1.8, 0.5) {+ $\kb$};
    \end{pgfonlayer}
    \begin{pgfonlayer}{background}
        \lEdge[cA,swap]{n1}{\a}{n2}
        \lEdge[cA]{n2}{\a}{n3}
        \lEdge[cA]{n1}{\a}{n4}
        \lEdge[cB, swap]{n3}{\b}{n4}
    \end{pgfonlayer}
    \begin{pgfonlayer}{prebackground}
        \draw[Selection] (n3.center) -- (n4.center) -- (n1.center);
    \end{pgfonlayer}
\end{tikzpicture}
}

\newcommand{\addOneAxIVxA}{
\begin{tikzpicture}[scale=\scaleOfPictures]
    \begin{pgfonlayer}{foreground}
        \node[ Node] at (0, 0) (n1) {};
        \node[ Node] at (1, 0) (n2) {};
        \node[ANode] at (0, 1) (n4) {};
        \node[ANode] at (1, 1) (n3) {};

        \node[ Node] at (-0.5, -0.5) (n5) {};
        \node[Basic] at (1.8, 0.5) {+ $\kb$};
    \end{pgfonlayer}
    \begin{pgfonlayer}{background}
        \lEdge[cA,swap]{n1}{\a}{n2}
        \lEdge[cA]{n2}{\a}{n3}
        \lEdge[cA]{n1}{\a}{n4}
        \lEdge[cA, swap]{n3}{\a}{n4}
    \end{pgfonlayer}
\end{tikzpicture}
}

\begin{center}
    \begin{tikzpicture}[scale=\scaleOfPictures]
        \node (Q) [AInstance] at (0, 0) {\addOneAxIVxQ};
        \node (B) [AResult, left=1cm of Q] {\addOneAxIVxB};
        \node (D) [AResult, right=1cm of Q] {\addOneAxIVxNO};
        \node (E) [AResult, below=1cm of Q] {\addOneAxIVxA};
        \node[anchor=north east, inner xsep=0] at (B.south east) (ok1) {\CaseOk 2)};
        \node[anchor=north east, inner xsep=0] at (D.south east) (ok2) {\CaseOk 1)};
        \node[anchor=east] at (E.west) (name) {$G_4 =$};
        \draw[CaseArrow]
        (Q)
        -- node[midway, auto,     , inner sep=0cm] {\CaseBEdge}
        (B);
        \draw[CaseArrow]
        (Q)
        -- node[midway, auto, swap, inner sep=0cm] {\CaseEdgeDoesNotExist}
        (D);
        \draw[CaseArrow]
        (Q)
        -- node[midway, auto, swap, inner sep=0cm] {\CaseAEdge}
        (E);
    \end{tikzpicture}
\end{center}
%
Again, two cases immediately lead us to the end of the proof (see the picture), so only one
needs further examination: If an $\a$-edge is present in the solution, we
have the graph $G_4 = \circ\a\a\a\a + \bullet + \kb$. It will come useful in a moment (at the end of the proof), but first we
have to 'construct' yet another one. The construction will take three upcoming amalgamations,
then we will return to $G_4$.\\


\Instance{\ref{lemma:addOneA}.4} To build the instance we again use graph $G_3$, this
time paired with the discrete graph $D_{k+4}$ --- we can afford to do that, since in $\G$ embeds
arbitrarily large discrete graphs.

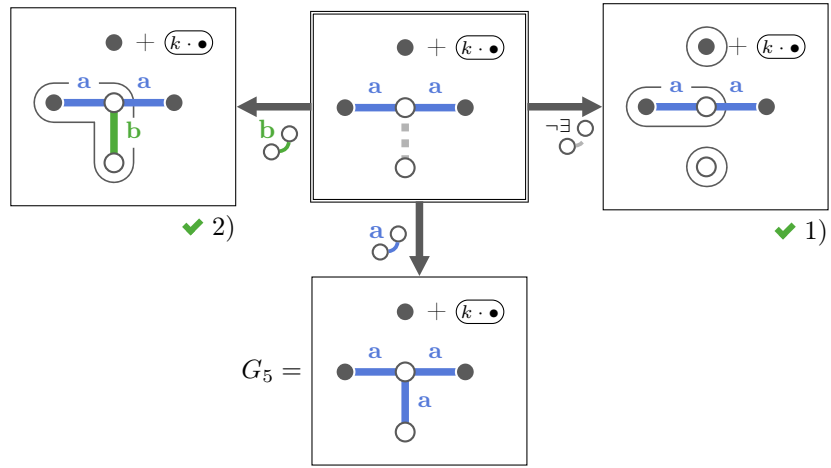
\newcommand{\addOneAxVxQ}{
\begin{tikzpicture}[scale=\scaleOfPictures]
    \begin{pgfonlayer}{foreground}
        \node[ANode] at (0, 0) (n0) {};
        \node[ Node] at (1, 0) (n1) {};
        \node[ Node] at (0, 1) (n2) {};
        \node[ Node] at (-1, 0) (n3) {};
        \node[ANode] at (0, -1) (n4) {};

        \node[Basic] at (1, 1) {+ $\kb$};
    \end{pgfonlayer}
    \begin{pgfonlayer}{background}
        \lEdge[cA]{n0}{\a}{n1}
        \lEdge[cA,swap]{n0}{\a}{n3}
        \lEdge[AEdge]{n0}{}{n4}
    \end{pgfonlayer}
\end{tikzpicture}
}
\newcommand{\addOneAxVxB}{
\begin{tikzpicture}[scale=\scaleOfPictures]
    \begin{pgfonlayer}{foreground}
        \node[ANode] at (0, 0) (n0) {};
        \node[ Node] at (1, 0) (n1) {};
        \node[ Node] at (0, 1) (n2) {};
        \node[ Node] at (-1, 0) (n3) {};
        \node[ANode] at (0, -1) (n4) {};

        \node[Basic] at (1, 1) {+ $\kb$};
    \end{pgfonlayer}
    \begin{pgfonlayer}{background}
        \lEdge[cA]{n0}{\a}{n1}
        \lEdge[cA,swap]{n0}{\a}{n3}
        \lEdge[cB]{n0}{\b}{n4}
    \end{pgfonlayer}
    \begin{pgfonlayer}{prebackground}
        \draw[Selection] (n3.center) -- (n0.center) -- (n4.center);
    \end{pgfonlayer}
\end{tikzpicture}
}
\newcommand{\addOneAxVxNO}{
\begin{tikzpicture}[scale=\scaleOfPictures]
    \begin{pgfonlayer}{foreground}
        \node[ANode] at (0, 0) (n0) {};
        \node[ Node] at (1, 0) (n1) {};
        \node[ Node] at (0, 1) (n2) {};
        \node[ Node] at (-1, 0) (n3) {};
        \node[ANode] at (0, -1) (n4) {};

        \node[Basic] at (1, 1) {+ $\kb$};
    \end{pgfonlayer}
    \begin{pgfonlayer}{background}
        \lEdge[cA]{n0}{\a}{n1}
        \lEdge[cA,swap]{n0}{\a}{n3}
    \end{pgfonlayer}
    \begin{pgfonlayer}{prebackground}
        \draw[Selection] (n3.center) -- (n0.center) (n4.center) -- (n4.center) (n2.center) -- (n2.center);
    \end{pgfonlayer}
\end{tikzpicture}
}
\newcommand{\addOneAxVxA}{
\begin{tikzpicture}[scale=\scaleOfPictures]
    \begin{pgfonlayer}{foreground}
        \node[ANode] at (0, 0) (n0) {};
        \node[ Node] at (1, 0) (n1) {};
        \node[ Node] at (0, 1) (n2) {};
        \node[ Node] at (-1, 0) (n3) {};
        \node[ANode] at (0, -1) (n4) {};

        \node[Basic] at (1, 1) {+ $\kb$};
    \end{pgfonlayer}
    \begin{pgfonlayer}{background}
        \lEdge[cA]{n0}{\a}{n1}
        \lEdge[cA,swap]{n0}{\a}{n3}
        \lEdge[cA]{n0}{\a}{n4}
    \end{pgfonlayer}
\end{tikzpicture}
}

\begin{center}
    \begin{tikzpicture}[scale=\scaleOfPictures]
        \node (Q) [AInstance] at (0, 0) {\addOneAxVxQ};
        \node (B) [AResult, left=1cm of Q] {\addOneAxVxB};
        \node (D) [AResult, right=1cm of Q] {\addOneAxVxNO};
        \node (E) [AResult, below=1cm of Q] {\addOneAxVxA};
        \node[anchor=north east, inner xsep=0] at (B.south east) (ok1) {\CaseOk 2)};
        \node[anchor=north east, inner xsep=0] at (D.south east) (ok2) {\CaseOk 1)};
        \node[anchor=east] at (E.west) (name) {$G_5 =$};
        \draw[CaseArrow]
        (Q)
        -- node[midway, auto,     , inner sep=0cm] {\CaseBEdge}
        (B);
        \draw[CaseArrow]
        (Q)
        -- node[midway, auto, swap, inner sep=0cm] {\CaseEdgeDoesNotExist}
        (D);
        \draw[CaseArrow]
        (Q)
        -- node[midway, auto, swap, inner sep=0cm] {\CaseAEdge}
        (E);
    \end{tikzpicture}
\end{center}
The acquired graph $G_5$ will be used in Instance \ref{lemma:addOneA}.6. To complete the proof of lemma, we still
need one more graph --- namely $\underline{\a\a\a} + \kb$. We will get it quickly in the
following instance of amalgamation:

\Instance{\ref{lemma:addOneA}.5} This time we put together two copies of $G_3$:

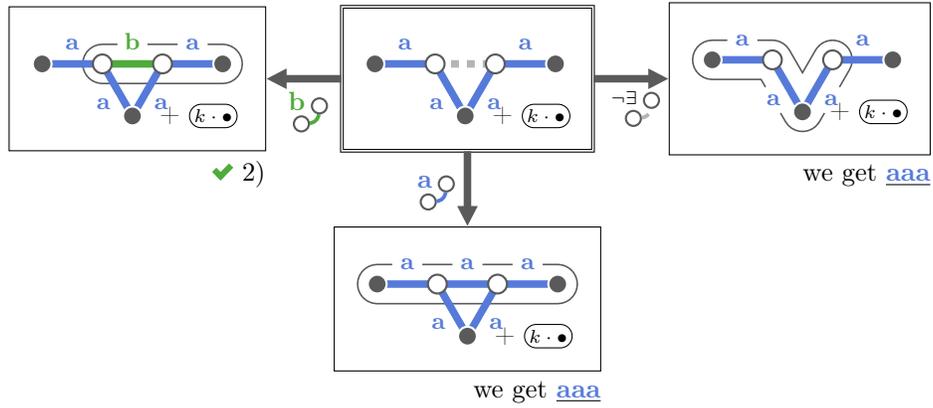
\newcommand{\addOneAxVIxQ}{
\begin{tikzpicture}[scale=\scaleOfPictures]
    \begin{pgfonlayer}{foreground}
        \node[ Node] at (0, 0) (n0) {};
        \node[ANode] at (1, 0) (n1) {};
        \node[ANode] at (2, 0) (n2) {};
        \node[ Node] at (3, 0) (n3) {};
        \node[ Node] at (1.5, -0.866) (nn) {};

        \node[Basic] at (2.6, -0.866) {+ $\kb$};
    \end{pgfonlayer}
    \begin{pgfonlayer}{background}
        \lEdge[cA]{n0}{\a}{n1}
        \lEdge[cA, swap]{n1}{\a}{nn}
        \lEdge[cA, swap]{nn}{\a}{n2}
        \lEdge[cA]{n2}{\a}{n3}
        \lEdge[AEdge]{n1}{}{n2}
    \end{pgfonlayer}
\end{tikzpicture}
}

\newcommand{\addOneAxVIxB}{
\begin{tikzpicture}[scale=\scaleOfPictures]
    \begin{pgfonlayer}{foreground}
        \node[ Node] at (0, 0) (n0) {};
        \node[ANode] at (1, 0) (n1) {};
        \node[ANode] at (2, 0) (n2) {};
        \node[ Node] at (3, 0) (n3) {};
        \node[ Node] at (1.5, -0.866) (nn) {};

        \node[Basic] at (2.6, -0.866) {+ $\kb$};
    \end{pgfonlayer}
    \begin{pgfonlayer}{background}
        \lEdge[cA]{n0}{\a}{n1}
        \lEdge[cA, swap]{n1}{\a}{nn}
        \lEdge[cA, swap]{nn}{\a}{n2}
        \lEdge[cA]{n2}{\a}{n3}
        \lEdge[cB]{n1}{\b}{n2}
    \end{pgfonlayer}
    \begin{pgfonlayer}{prebackground}
        \draw[Selection] (n1.center) -- (n2.center) -- (n3.center);
    \end{pgfonlayer}
\end{tikzpicture}
}

\newcommand{\addOneAxVIxA}{
\begin{tikzpicture}[scale=\scaleOfPictures]
    \begin{pgfonlayer}{foreground}
        \node[ Node] at (0, 0) (n0) {};
        \node[ANode] at (1, 0) (n1) {};
        \node[ANode] at (2, 0) (n2) {};
        \node[ Node] at (3, 0) (n3) {};
        \node[ Node] at (1.5, -0.866) (nn) {};

        \node[Basic] at (2.6, -0.866) {+ $\kb$};
    \end{pgfonlayer}
    \begin{pgfonlayer}{background}
        \lEdge[cA]{n0}{\a}{n1}
        \lEdge[cA, swap]{n1}{\a}{nn}
        \lEdge[cA, swap]{nn}{\a}{n2}
        \lEdge[cA]{n2}{\a}{n3}
        \lEdge[cA]{n1}{\a}{n2}
    \end{pgfonlayer}
    \begin{pgfonlayer}{prebackground}
        \draw[Selection] (n0.center) -- (n3.center);
    \end{pgfonlayer}
\end{tikzpicture}
}

\newcommand{\addOneAxVIxNO}{
\begin{tikzpicture}[scale=\scaleOfPictures]
    \begin{pgfonlayer}{foreground}
        \node[ Node] at (0, 0) (n0) {};
        \node[ANode] at (1, 0) (n1) {};
        \node[ANode] at (2, 0) (n2) {};
        \node[ Node] at (3, 0) (n3) {};
        \node[ Node] at (1.5, -0.866) (nn) {};

        \node[Basic] at (2.6, -0.866) {+ $\kb$};
    \end{pgfonlayer}
    \begin{pgfonlayer}{background}
        \lEdge[cA]{n0}{\a}{n1}
        \lEdge[cA, swap]{n1}{\a}{nn}
        \lEdge[cA, swap]{nn}{\a}{n2}
        \lEdge[cA]{n2}{\a}{n3}
    \end{pgfonlayer}
    \begin{pgfonlayer}{prebackground}
        \draw[Selection] (n0.center) -- (n1.center) -- (nn.center) -- (n2.center);
    \end{pgfonlayer}
\end{tikzpicture}
}

\begin{center}
    \begin{tikzpicture}[scale=\scaleOfPictures]
        \node (Q) [AInstance] at (0, 0) {\addOneAxVIxQ};
        \node (B) [AResult, left=1cm of Q] {\addOneAxVIxB};
        \node (D) [AResult, right=1cm of Q] {\addOneAxVIxNO};
        \node (E) [AResult, below=1cm of Q] {\addOneAxVIxA};
        \node[anchor=north east, inner xsep=0] at (B.south east) (ok1) {\CaseOk 2)};
        \node[anchor=north east, inner xsep=0] at (D.south east) (ok2) {we get $\underline{\a\a\a}$};
        \node[anchor=north east, inner xsep=0] at (E.south east) (ok3) {we get $\underline{\a\a\a}$};
        \draw[CaseArrow]
        (Q)
        -- node[midway, auto,     , inner sep=0cm] {\CaseBEdge}
        (B);
        \draw[CaseArrow]
        (Q)
        -- node[midway, auto, swap, inner sep=0cm] {\CaseEdgeDoesNotExist}
        (D);
        \draw[CaseArrow]
        (Q)
        -- node[midway, auto, swap, inner sep=0cm] {\CaseAEdge}
        (E);
    \end{tikzpicture}
\end{center}
If we obtained a $\b$-edge, we luckily end, having met the condition \CaseOk 2). In both remaining
cases from the resulting graph we derive a path $\underline{\a\a\a}$.

\Instance{\ref{lemma:addOneA}.6}
Using that path together with $G_5$ (from Instance \ref{lemma:addOneA}.4), we construct another instance of amalgamation.
Fortunately, it is the penultimate instance in the proof of the current lemma.

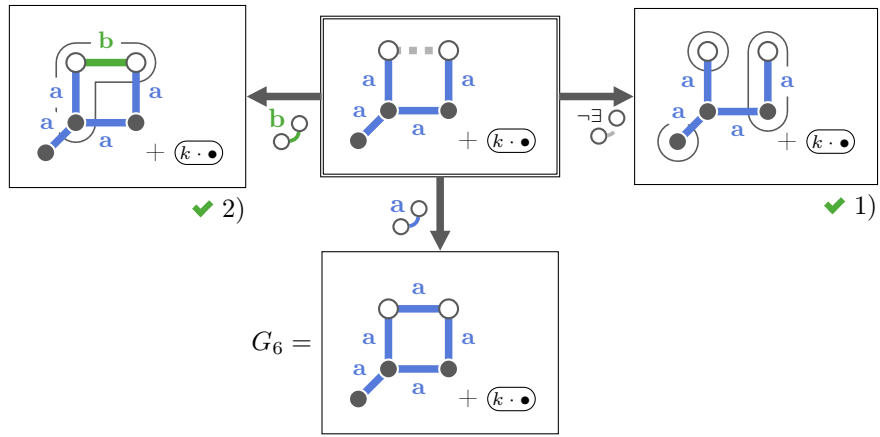
\newcommand{\addOneAxVIIxQ}{
\begin{tikzpicture}[scale=\scaleOfPictures]
    \begin{pgfonlayer}{foreground}
        \node[ Node] at (0, 0) (n0) {};
        \node[ Node] at (1, 0) (n1) {};
        \node[ANode] at (1, 1) (n2) {};
        \node[ANode] at (0, 1) (n3) {};
        \node[ Node] at (-0.5, -0.5) (ogon) {};

        \node[Basic] at (1.8, -0.5) {+ $\kb$};
    \end{pgfonlayer}
    \begin{pgfonlayer}{background}
        \lEdge[cA]{ogon}{\a}{n0}
        \lEdge[cA,swap]{n0}{\a}{n1}
        \lEdge[cA,swap]{n1}{\a}{n2}
        \lEdge[cA,swap]{n3}{\a}{n0}
        \lEdge[AEdge]{n2}{}{n3}
    \end{pgfonlayer}
\end{tikzpicture}
}

\newcommand{\addOneAxVIIxB}{
\begin{tikzpicture}[scale=\scaleOfPictures]
    \begin{pgfonlayer}{foreground}
        \node[ Node] at (0, 0) (n0) {};
        \node[ Node] at (1, 0) (n1) {};
        \node[ANode] at (1, 1) (n2) {};
        \node[ANode] at (0, 1) (n3) {};
        \node[ Node] at (-0.5, -0.5) (ogon) {};

        \node[Basic] at (1.8, -0.5) {+ $\kb$};
    \end{pgfonlayer}
    \begin{pgfonlayer}{background}
        \lEdge[cA]{ogon}{\a}{n0}
        \lEdge[cA,swap]{n0}{\a}{n1}
        \lEdge[cA,swap]{n1}{\a}{n2}
        \lEdge[cA,swap]{n3}{\a}{n0}
        \lEdge[cB,swap]{n2}{\b}{n3}
    \end{pgfonlayer}
    \begin{pgfonlayer}{prebackground}
        \draw[Selection] (n2.center) -- (n3.center) -- (n0.center);
    \end{pgfonlayer}
\end{tikzpicture}
}

\newcommand{\addOneAxVIIxNO}{
\begin{tikzpicture}[scale=\scaleOfPictures]
    \begin{pgfonlayer}{foreground}
        \node[ Node] at (0, 0) (n0) {};
        \node[ Node] at (1, 0) (n1) {};
        \node[ANode] at (1, 1) (n2) {};
        \node[ANode] at (0, 1) (n3) {};
        \node[ Node] at (-0.5, -0.5) (ogon) {};

        \node[Basic] at (1.8, -0.5) {+ $\kb$};
    \end{pgfonlayer}
    \begin{pgfonlayer}{background}
        \lEdge[cA]{ogon}{\a}{n0}
        \lEdge[cA,swap]{n0}{\a}{n1}
        \lEdge[cA,swap]{n1}{\a}{n2}
        \lEdge[cA,swap]{n3}{\a}{n0}
    \end{pgfonlayer}
    \begin{pgfonlayer}{prebackground}
        \draw[Selection] (n3.center) -- (n3.center) (ogon.center) -- (ogon.center) (n1.center) -- (n2.center);
    \end{pgfonlayer}
\end{tikzpicture}
}

\newcommand{\addOneAxVIIxA}{
\begin{tikzpicture}[scale=\scaleOfPictures]
    \begin{pgfonlayer}{foreground}
        \node[ Node] at (0, 0) (n0) {};
        \node[ Node] at (1, 0) (n1) {};
        \node[ANode] at (1, 1) (n2) {};
        \node[ANode] at (0, 1) (n3) {};
        \node[ Node] at (-0.5, -0.5) (ogon) {};

        \node[Basic] at (1.8, -0.5) {+ $\kb$};
    \end{pgfonlayer}
    \begin{pgfonlayer}{background}
        \lEdge[cA]{ogon}{\a}{n0}
        \lEdge[cA,swap]{n0}{\a}{n1}
        \lEdge[cA,swap]{n1}{\a}{n2}
        \lEdge[cA,swap]{n3}{\a}{n0}
        \lEdge[cA,swap]{n2}{\a}{n3}
    \end{pgfonlayer}
\end{tikzpicture}
}

\begin{center}
    \begin{tikzpicture}[scale=\scaleOfPictures]
        \node (Q) [AInstance] at (0, 0) {\addOneAxVIIxQ};
        \node (B) [AResult, left=1cm of Q] {\addOneAxVIIxB};
        \node (D) [AResult, right=1cm of Q] {\addOneAxVIIxNO};
        \node (E) [AResult, below=1cm of Q] {\addOneAxVIIxA};
        \node[anchor=north east, inner xsep=0] at (B.south east) (ok1) {\CaseOk 2)};
        \node[anchor=north east, inner xsep=0] at (D.south east) (ok2) {\CaseOk 1)};
        \node[anchor=east] at (E.west) (name) {$G_6 =$};
        \draw[CaseArrow]
        (Q)
        -- node[midway, auto,     , inner sep=0cm] {\CaseBEdge}
        (B);
        \draw[CaseArrow]
        (Q)
        -- node[midway, auto, swap, inner sep=0cm] {\CaseEdgeDoesNotExist}
        (D);
        \draw[CaseArrow]
        (Q)
        -- node[midway, auto, swap, inner sep=0cm] {\CaseAEdge}
        (E);
    \end{tikzpicture}
\end{center}
Similarly as in all previous instances, only one case does not end immediately by satisfying
one of the conditions \CaseOk 1) or \CaseOk 2). Let $G_6$ be the graph we get in the $\a$-edge--case.

\Instance{\ref{lemma:addOneA}.7}
We have nearly made it through to the end of the proof of Lemma~\ref{lemma:addOneA}. For
construction of the last amalgamation instance we need graphs $G_4$ (from Instance \ref{lemma:addOneA}.3) and $G_6$
(just created).

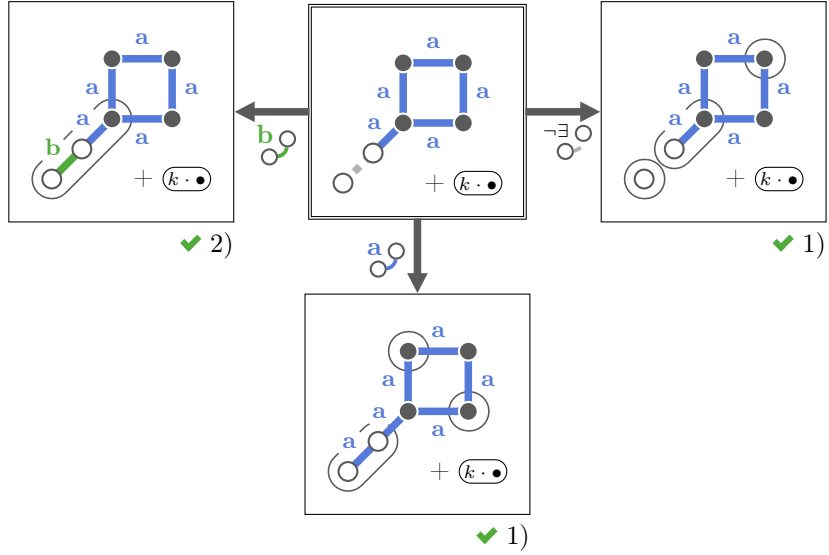
\newcommand{\addOneAxVIIIxQ}{
\begin{tikzpicture}[scale=\scaleOfPictures]
    \begin{pgfonlayer}{foreground}
        \node[ Node] at (0, 0) (n0) {};
        \node[ Node] at (1, 0) (n1) {};
        \node[ Node] at (1, 1) (n2) {};
        \node[ Node] at (0, 1) (n3) {};
        \node[ANode] at (-0.5, -0.5) (o1) {};
        \node[ANode] at (-1, -1) (o2) {};

        \node[Basic] at (1, -1) {+ $\kb$};
    \end{pgfonlayer}
    \begin{pgfonlayer}{background}
        \lEdge[cA]{o1}{\a}{n0}
        \lEdge[cA,swap]{n0}{\a}{n1}
        \lEdge[cA,swap]{n1}{\a}{n2}
        \lEdge[cA,swap]{n2}{\a}{n3}
        \lEdge[cA,swap]{n3}{\a}{n0}
        \lEdge[AEdge]{o2}{}{o1}
    \end{pgfonlayer}
\end{tikzpicture}
}
\newcommand{\addOneAxVIIIxA}{
\begin{tikzpicture}[scale=\scaleOfPictures]
    \begin{pgfonlayer}{foreground}
        \node[ Node] at (0, 0) (n0) {};
        \node[ Node] at (1, 0) (n1) {};
        \node[ Node] at (1, 1) (n2) {};
        \node[ Node] at (0, 1) (n3) {};
        \node[ANode] at (-0.5, -0.5) (o1) {};
        \node[ANode] at (-1, -1) (o2) {};

        \node[Basic] at (1, -1) {+ $\kb$};
    \end{pgfonlayer}
    \begin{pgfonlayer}{background}
        \lEdge[cA]{o1}{\a}{n0}
        \lEdge[cA,swap]{n0}{\a}{n1}
        \lEdge[cA,swap]{n1}{\a}{n2}
        \lEdge[cA,swap]{n2}{\a}{n3}
        \lEdge[cA,swap]{n3}{\a}{n0}
        \lEdge[cA]{o2}{\a}{o1}
    \end{pgfonlayer}
    \begin{pgfonlayer}{prebackground}
        \draw[Selection] (o1.center) -- (o2.center) (n1.center) -- (n1.center) (n3.center) -- (n3.center);
    \end{pgfonlayer}
\end{tikzpicture}
}
\newcommand{\addOneAxVIIIxB}{
\begin{tikzpicture}[scale=\scaleOfPictures]
    \begin{pgfonlayer}{foreground}
        \node[ Node] at (0, 0) (n0) {};
        \node[ Node] at (1, 0) (n1) {};
        \node[ Node] at (1, 1) (n2) {};
        \node[ Node] at (0, 1) (n3) {};
        \node[ANode] at (-0.5, -0.5) (o1) {};
        \node[ANode] at (-1, -1) (o2) {};

        \node[Basic] at (1, -1) {+ $\kb$};
    \end{pgfonlayer}
    \begin{pgfonlayer}{background}
        \lEdge[cA]{o1}{\a}{n0}
        \lEdge[cA,swap]{n0}{\a}{n1}
        \lEdge[cA,swap]{n1}{\a}{n2}
        \lEdge[cA,swap]{n2}{\a}{n3}
        \lEdge[cA,swap]{n3}{\a}{n0}
        \lEdge[cB]{o2}{\b}{o1}
    \end{pgfonlayer}
    \begin{pgfonlayer}{prebackground}
        \draw[Selection] (o2.center) -- (o1.center) -- (n0.center);
    \end{pgfonlayer}
\end{tikzpicture}
}
\newcommand{\addOneAxVIIIxNO}{
\begin{tikzpicture}[scale=\scaleOfPictures]
    \begin{pgfonlayer}{foreground}
        \node[ Node] at (0, 0) (n0) {};
        \node[ Node] at (1, 0) (n1) {};
        \node[ Node] at (1, 1) (n2) {};
        \node[ Node] at (0, 1) (n3) {};
        \node[ANode] at (-0.5, -0.5) (o1) {};
        \node[ANode] at (-1, -1) (o2) {};

        \node[Basic] at (1, -1) {+ $\kb$};
    \end{pgfonlayer}
    \begin{pgfonlayer}{background}
        \lEdge[cA]{o1}{\a}{n0}
        \lEdge[cA,swap]{n0}{\a}{n1}
        \lEdge[cA,swap]{n1}{\a}{n2}
        \lEdge[cA,swap]{n2}{\a}{n3}
        \lEdge[cA,swap]{n3}{\a}{n0}
    \end{pgfonlayer}
    \begin{pgfonlayer}{prebackground}
        \draw[Selection] (o2.center) -- (o2.center) (o1.center) -- (n0.center) (n2.center) -- (n2.center);
    \end{pgfonlayer}
\end{tikzpicture}
}

\begin{center}
    \begin{tikzpicture}[scale=\scaleOfPictures]
        \node (Q) [AInstance] at (0, 0) {\addOneAxVIIIxQ};
        \node (B) [AResult, left=1cm of Q] {\addOneAxVIIIxB};
        \node (D) [AResult, right=1cm of Q] {\addOneAxVIIIxNO};
        \node (E) [AResult, below=1cm of Q] {\addOneAxVIIIxA};
        \node[anchor=north east, inner xsep=0] at (B.south east) (ok1) {\CaseOk 2)};
        \node[anchor=north east, inner xsep=0] at (D.south east) (ok2) {\CaseOk 1)};
        \node[anchor=north east, inner xsep=0] at (E.south east) (ok3) {\CaseOk 1)};
        \draw[CaseArrow]
        (Q)
        -- node[midway, auto,     , inner sep=0cm] {\CaseBEdge}
        (B);
        \draw[CaseArrow]
        (Q)
        -- node[midway, auto, swap, inner sep=0cm] {\CaseEdgeDoesNotExist}
        (D);
        \draw[CaseArrow]
        (Q)
        -- node[midway, auto, swap, inner sep=0cm] {\CaseAEdge}
        (E);
    \end{tikzpicture}
\end{center}

Each of three possible outcomes of this instance allows to fulfill the conditions \CaseOk 1) or
\CaseOk 2), thus we finally completed the proof of Lemma~\ref{lemma:addOneA}.
\end{proof}

There is nothing left to do but to proceed with proving the next lemma. This proof will be a bit
shorter, as it consists only of four amalgamation instances.


\begin{proof}[Proof of Lemma~\ref{lemma:addOneB}]
The assumptions of the lemma require $\G$ to embed the following graphs:
\begin{itemize}
    \item graph $G_1 = \underline{\a\y} + \kb$,
    \item graph $G_2 = \underline{\a} + \bullet + \bullet + \kb$ obtained as the result of previous lemma.
\end{itemize}
for some colors $\a, \y \in \ColorsTwo$.
As before, let $\b$ denote the second (i.e. different than $\a$) color from two-element set $\ColorsTwo$.

\paragraph{Proof structure} Present lemma aims at showing that $\G$ embeds a graph of the form
$\underline{\a\z} + \bullet + \kb$. The structure of the proof has a slight subtlety: depending on color $\y$
two different cases may occur:

\begin{enumerate}
    \item if $\y = \b$, then we are bound to succeed with finding the required graph $\underline{\a\z} + \bullet + \kb$,
    \item however, if $\y = \a$, in some case we may not immediately find such graph. Instead of it,
    first we will find graph $G_1' = \underline{\a\b} + \kb$ --- a graph that looks like $G_1$
    we have in our assumptions, but with one edge recolored from $\y$ to $\b$. This graph
    allows us to repeat the whole reasoning, but now with the guarantee that we will end in the
    first case ($\y = \b$).
\end{enumerate}

Let us now move on to the proof --- even if the subtlety is not entirely clear now,
everything should get more evident, when we will get to the problematic point.

\Instance{\ref{lemma:addOneB}.1} The first amalgamation instance
is built using the graphs $G_1$ and $G_2$ following from the assumptions:

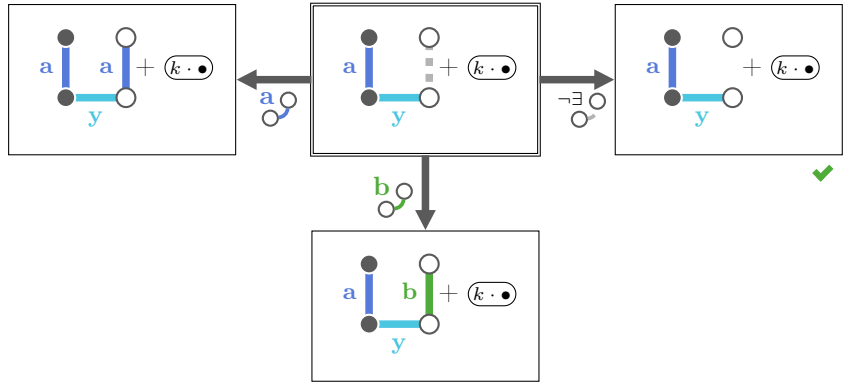
\newcommand{\addOneBxIxQ}{
\begin{tikzpicture}[scale=\scaleOfPictures]
    \begin{pgfonlayer}{foreground}
        \node[ Node] at (0, 0) (n0) {};
        \node[ANode] at (1, 0) (n1) {};
        \node[ANode] at (1, 1) (n2) {};
        \node[ Node] at (0, 1) (n3) {};

        \node[Basic] at (1.8, 0.5) {+ $\kb$};
    \end{pgfonlayer}
    \begin{pgfonlayer}{background}
        \lEdge[cA,swap]{n3}{\a}{n0}
        \lEdge[cY,swap]{n0}{\y}{n1}
        \lEdge[AEdge]{n1}{}{n2}
    \end{pgfonlayer}
\end{tikzpicture}
}
\newcommand{\addOneBxIxA}{
\begin{tikzpicture}[scale=\scaleOfPictures]
    \begin{pgfonlayer}{foreground}
        \node[ Node] at (0, 0) (n0) {};
        \node[ANode] at (1, 0) (n1) {};
        \node[ANode] at (1, 1) (n2) {};
        \node[ Node] at (0, 1) (n3) {};

        \node[Basic] at (1.8, 0.5) {+ $\kb$};
    \end{pgfonlayer}
    \begin{pgfonlayer}{background}
        \lEdge[cA,swap]{n3}{\a}{n0}
        \lEdge[cY,swap]{n0}{\y}{n1}
        \lEdge[cA]{n1}{\a}{n2}
    \end{pgfonlayer}
\end{tikzpicture}
}
\newcommand{\addOneBxIxB}{
\begin{tikzpicture}[scale=\scaleOfPictures]
    \begin{pgfonlayer}{foreground}
        \node[ Node] at (0, 0) (n0) {};
        \node[ANode] at (1, 0) (n1) {};
        \node[ANode] at (1, 1) (n2) {};
        \node[ Node] at (0, 1) (n3) {};

        \node[Basic] at (1.8, 0.5) {+ $\kb$};
    \end{pgfonlayer}
    \begin{pgfonlayer}{background}
        \lEdge[cA,swap]{n3}{\a}{n0}
        \lEdge[cY,swap]{n0}{\y}{n1}
        \lEdge[cB]{n1}{\b}{n2}
    \end{pgfonlayer}
\end{tikzpicture}
}
\newcommand{\addOneBxIxNO}{
\begin{tikzpicture}[scale=\scaleOfPictures]
    \begin{pgfonlayer}{foreground}
        \node[ Node] at (0, 0) (n0) {};
        \node[ANode] at (1, 0) (n1) {};
        \node[ANode] at (1, 1) (n2) {};
        \node[ Node] at (0, 1) (n3) {};

        \node[Basic] at (1.8, 0.5) {+ $\kb$};
    \end{pgfonlayer}
    \begin{pgfonlayer}{background}
        \lEdge[cA,swap]{n3}{\a}{n0}
        \lEdge[cY,swap]{n0}{\y}{n1}
    \end{pgfonlayer}
\end{tikzpicture}
}

\begin{center}
    \begin{tikzpicture}[scale=\scaleOfPictures]
        \node (Q) [AInstance] at (0, 0) {\addOneBxIxQ};
        \node (B) [AResult, left=1cm of Q] {\addOneBxIxA};
        \node (D) [AResult, right=1cm of Q] {\addOneBxIxNO};
        \node (E) [AResult, below=1cm of Q] {\addOneBxIxB};
        \node[anchor=north east, inner xsep=0] at (D.south east) (ok2) {\CaseOk};
        \draw[CaseArrow]
        (Q)
        -- node[midway, auto,     , inner sep=0cm] {\CaseAEdge}
        (B);
        \draw[CaseArrow]
        (Q)
        -- node[midway, auto, swap, inner sep=0cm] {\CaseEdgeDoesNotExist}
        (D);
        \draw[CaseArrow]
        (Q)
        -- node[midway, auto, swap, inner sep=0cm] {\CaseBEdge}
        (E);
    \end{tikzpicture}
\end{center}
In case where the solution does not contain a new edge, we directly get the graph we are looking for.
The case of and $\a$-edge is not much difficult -- to successfully deal with it,
we only need one additional amalgamation.
It turns out, that the appearance of a $\b$-edge is the most cumbersome case.
We will return to it in instance \ref{lemma:addOneB}.3.

\Instance{\ref{lemma:addOneB}.2} Here we use the graph $\underline{\a\x\a}$ we just obtained (in case of $\a$-edge)
together with $G_2$.

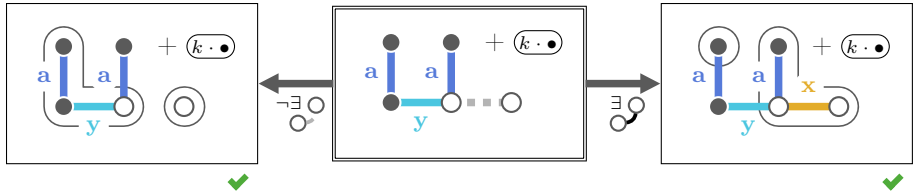
\newcommand{\addOneBxIIxQ}{
\begin{tikzpicture}[scale=\scaleOfPictures]
    \begin{pgfonlayer}{foreground}
        \node[ Node] at (0, 0) (n0) {};
        \node[ANode] at (1, 0) (n1) {};
        \node[ Node] at (1, 1) (n2) {};
        \node[ Node] at (0, 1) (n3) {};
        \node[ANode] at (2, 0) (n4) {};

        \node[Basic] at (2.2, 1) {+ $\kb$};
    \end{pgfonlayer}
    \begin{pgfonlayer}{background}
        \lEdge[cA,swap]{n3}{\a}{n0}
        \lEdge[cA]{n1}{\a}{n2}
        \lEdge[cY,swap]{n0}{\y}{n1}
        \lEdge[AEdge]{n1}{}{n4}
    \end{pgfonlayer}
\end{tikzpicture}
}
\newcommand{\addOneBxIIxZ}{
\begin{tikzpicture}[scale=\scaleOfPictures]
    \begin{pgfonlayer}{foreground}
        \node[ Node] at (0, 0) (n0) {};
        \node[ANode] at (1, 0) (n1) {};
        \node[ Node] at (1, 1) (n2) {};
        \node[ Node] at (0, 1) (n3) {};
        \node[ANode] at (2, 0) (n4) {};

        \node[Basic] at (2.2, 1) {+ $\kb$};
    \end{pgfonlayer}
    \begin{pgfonlayer}{background}
        \lEdge[cA,swap]{n3}{\a}{n0}
        \lEdge[cA]{n1}{\a}{n2}
        \lEdge[cY,swap]{n0}{\y}{n1}
        \lEdge[cX]{n1}{\x}{n4}
    \end{pgfonlayer}
    \begin{pgfonlayer}{prebackground}
        \draw[Selection] (n4.center) -- (n1.center) -- (n2.center) (n3.center) -- (n3.center);
    \end{pgfonlayer}
\end{tikzpicture}
}
\newcommand{\addOneBxIIxNO}{
\begin{tikzpicture}[scale=\scaleOfPictures]
    \begin{pgfonlayer}{foreground}
        \node[ Node] at (0, 0) (n0) {};
        \node[ANode] at (1, 0) (n1) {};
        \node[ Node] at (1, 1) (n2) {};
        \node[ Node] at (0, 1) (n3) {};
        \node[ANode] at (2, 0) (n4) {};

        \node[Basic] at (2.2, 1) {+ $\kb$};
    \end{pgfonlayer}
    \begin{pgfonlayer}{background}
        \lEdge[cA,swap]{n3}{\a}{n0}
        \lEdge[cA]{n1}{\a}{n2}
        \lEdge[cY,swap]{n0}{\y}{n1}
    \end{pgfonlayer}
    \begin{pgfonlayer}{prebackground}
        \draw[Selection] (n3.center) -- (n0.center) -- (n1.center) (n4.center) -- (n4.center);
    \end{pgfonlayer}
\end{tikzpicture}
}
\begin{center}
    \begin{tikzpicture}[scale=\scaleOfPictures]
        \node (Q) [AInstance] at (0, 0) {\addOneBxIIxQ};
        \node (B) [AResult, left=1cm of Q] {\addOneBxIIxNO};
        \node (D) [AResult, right=1cm of Q] {\addOneBxIIxZ};
        \node[anchor=north east, inner xsep=0] at (B.south east) (ok1) {\CaseOk};
        \node[anchor=north east, inner xsep=0] at (D.south east) (ok2) {\CaseOk};
        \draw[CaseArrow]
        (Q)
        -- node[midway, auto,     , inner sep=0cm] {\CaseEdgeDoesNotExist}
        (B);
        \draw[CaseArrow]
        (Q)
        -- node[midway, auto, swap, inner sep=0cm] {\CaseEdgeExists}
        (D);
    \end{tikzpicture}
\end{center}
In each of possible cases we get a graph that matches the pattern we look for --- a graph
$\underline{\a\z} + \bullet + \kb$ for some $\z \in \ColorsTwo$. Let us return to the
omitted $\b$-edge case of Instance \ref{lemma:addOneB}.1:

\Instance{\ref{lemma:addOneB}.3} Present instance differs from the previous one only
with the color of one edge, but it has substantial consequences for our proof.

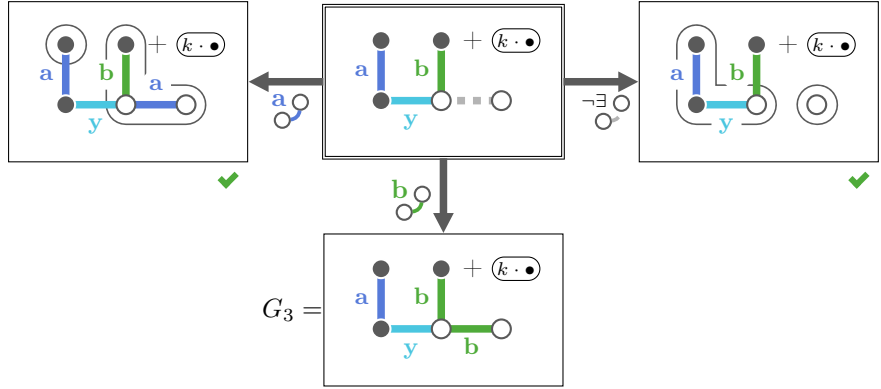
\newcommand{\addOneBxIIIxQ}{
\begin{tikzpicture}[scale=\scaleOfPictures]
    \begin{pgfonlayer}{foreground}
        \node[ Node] at (0, 0) (n0) {};
        \node[ANode] at (1, 0) (n1) {};
        \node[ Node] at (1, 1) (n2) {};
        \node[ Node] at (0, 1) (n3) {};
        \node[ANode] at (2, 0) (n4) {};

        \node[Basic] at (2, 1) {+ $\kb$};
    \end{pgfonlayer}
    \begin{pgfonlayer}{background}
        \lEdge[cA,swap]{n3}{\a}{n0}
        \lEdge[cB]{n1}{\b}{n2}
        \lEdge[cY,swap]{n0}{\y}{n1}
        \lEdge[AEdge]{n1}{}{n4}
    \end{pgfonlayer}
\end{tikzpicture}
}
\newcommand{\addOneBxIIIxNO}{
\begin{tikzpicture}[scale=\scaleOfPictures]
    \begin{pgfonlayer}{foreground}
        \node[ Node] at (0, 0) (n0) {};
        \node[ANode] at (1, 0) (n1) {};
        \node[ Node] at (1, 1) (n2) {};
        \node[ Node] at (0, 1) (n3) {};
        \node[ANode] at (2, 0) (n4) {};

        \node[Basic] at (2, 1) {+ $\kb$};
    \end{pgfonlayer}
    \begin{pgfonlayer}{background}
        \lEdge[cA,swap]{n3}{\a}{n0}
        \lEdge[cB]{n1}{\b}{n2}
        \lEdge[cY,swap]{n0}{\y}{n1}
    \end{pgfonlayer}
    \begin{pgfonlayer}{prebackground}
        \draw[Selection] (n3.center) -- (n0.center) -- (n1.center) (n4.center) -- (n4.center);
    \end{pgfonlayer}
\end{tikzpicture}
}
\newcommand{\addOneBxIIIxA}{
\begin{tikzpicture}[scale=\scaleOfPictures]
    \begin{pgfonlayer}{foreground}
        \node[ Node] at (0, 0) (n0) {};
        \node[ANode] at (1, 0) (n1) {};
        \node[ Node] at (1, 1) (n2) {};
        \node[ Node] at (0, 1) (n3) {};
        \node[ANode] at (2, 0) (n4) {};

        \node[Basic] at (2, 1) {+ $\kb$};
    \end{pgfonlayer}
    \begin{pgfonlayer}{background}
        \lEdge[cA,swap]{n3}{\a}{n0}
        \lEdge[cB]{n1}{\b}{n2}
        \lEdge[cY,swap]{n0}{\y}{n1}
        \lEdge[cA]{n1}{\a}{n4}
    \end{pgfonlayer}
    \begin{pgfonlayer}{prebackground}
        \draw[Selection] (n2.center) -- (n1.center) -- (n4.center) (n3.center) -- (n3.center);
    \end{pgfonlayer}
\end{tikzpicture}
}
\newcommand{\addOneBxIIIxB}{
\begin{tikzpicture}[scale=\scaleOfPictures]
    \begin{pgfonlayer}{foreground}
        \node[ Node] at (0, 0) (n0) {};
        \node[ANode] at (1, 0) (n1) {};
        \node[ Node] at (1, 1) (n2) {};
        \node[ Node] at (0, 1) (n3) {};
        \node[ANode] at (2, 0) (n4) {};

        \node[Basic] at (2, 1) {+ $\kb$};
    \end{pgfonlayer}
    \begin{pgfonlayer}{background}
        \lEdge[cA,swap]{n3}{\a}{n0}
        \lEdge[cB]{n1}{\b}{n2}
        \lEdge[cY,swap]{n0}{\y}{n1}
        \lEdge[cB, swap]{n1}{\b}{n4}
    \end{pgfonlayer}
\end{tikzpicture}
}

\begin{center}
    \begin{tikzpicture}[scale=\scaleOfPictures]
        \node (Q) [AInstance] at (0, 0) {\addOneBxIIIxQ};
        \node (B) [AResult, left=1cm of Q] {\addOneBxIIIxA};
        \node (D) [AResult, right=1cm of Q] {\addOneBxIIIxNO};
        \node (E) [AResult, below=1cm of Q] {\addOneBxIIIxB};
        \node[anchor=north east, inner xsep=0] at (B.south east) (ok1) {\CaseOk};
        \node[anchor=north east, inner xsep=0] at (D.south east) (ok2) {\CaseOk};
        \node[anchor=east, inner xsep=0] at (E.west) (name) {$G_3 = $};
        \draw[CaseArrow]
        (Q)
        -- node[midway, auto,     , inner sep=0cm] {\CaseAEdge}
        (B);
        \draw[CaseArrow]
        (Q)
        -- node[midway, auto, swap, inner sep=0cm] {\CaseEdgeDoesNotExist}
        (D);
        \draw[CaseArrow]
        (Q)
        -- node[midway, auto, swap, inner sep=0cm] {\CaseBEdge}
        (E);
    \end{tikzpicture}
\end{center}

\noindent
Let us now consider two possible values of edge color $\y$ in the resulting graph.

\paragraph{Case $1^\circ$} ($\y = \b$). Here, to get the graph we look for,
it suffices to build one additional amalgamation instance.
As the ingredients we take two copies of graph $G_3$, having in mind the assumed color
substitution $\y = \b$:

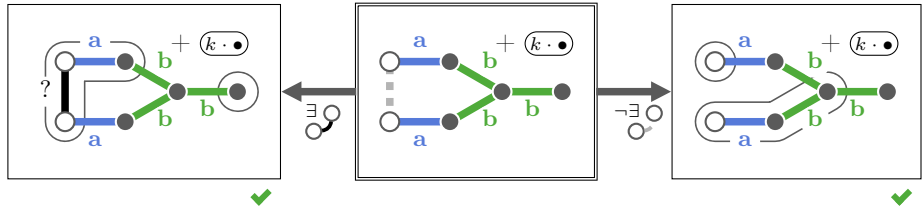
\Instance{\ref{lemma:addOneB}.4}

\newcommand{\addOneBxIVxQ}{
\begin{tikzpicture}[scale=\scaleOfPictures]
    \begin{pgfonlayer}{foreground}
        \node[ANode] at (0, 0) (n0) {};
        \node[ Node] at (1, 0) (n1) {};
        \node[ Node] at (1, 1) (n2) {};
        \node[ANode] at (0, 1) (n3) {};

        \node[ Node] at (1.866, 0.5) (m1) {};
        \node[ Node] at (2.866, 0.5) (m2) {};

        \node[Basic] at (2.4, 1.3) {+ $\kb$};
    \end{pgfonlayer}
    \begin{pgfonlayer}{background}
        \lEdge[cA]{n3}{\a}{n2}
        \lEdge[cB]{n2}{\b}{m1}
        \lEdge[cA,swap]{n0}{\a}{n1}
        \lEdge[cB,swap]{n1}{\b}{m1}
        \lEdge[cB,swap]{m1}{\b}{m2}
        \lEdge[AEdge]{n0}{}{n3}
    \end{pgfonlayer}
\end{tikzpicture}
}
\newcommand{\addOneBxIVxX}{
\begin{tikzpicture}[scale=\scaleOfPictures]
    \begin{pgfonlayer}{foreground}
        \node[ANode] at (0, 0) (n0) {};
        \node[ Node] at (1, 0) (n1) {};
        \node[ Node] at (1, 1) (n2) {};
        \node[ANode] at (0, 1) (n3) {};

        \node[ Node] at (1.866, 0.5) (m1) {};
        \node[ Node] at (2.866, 0.5) (m2) {};

        \node[Basic] at (2.4, 1.3) {+ $\kb$};
    \end{pgfonlayer}
    \begin{pgfonlayer}{background}
        \lEdge[cA]{n3}{\a}{n2}
        \lEdge[cB]{n2}{\b}{m1}
        \lEdge[cA,swap]{n0}{\a}{n1}
        \lEdge[cB,swap]{n1}{\b}{m1}
        \lEdge[cB,swap]{m1}{\b}{m2}
        \lEdge[cBlack]{n0}{\n}{n3}
    \end{pgfonlayer}
    \begin{pgfonlayer}{prebackground}
        \draw[Selection] (n0.center) -- (n3.center) -- (n2.center) (m2.center) -- (m2.center);
    \end{pgfonlayer}
\end{tikzpicture}
}
\newcommand{\addOneBxIVxNO}{
\begin{tikzpicture}[scale=\scaleOfPictures]
    \begin{pgfonlayer}{foreground}
        \node[ANode] at (0, 0) (n0) {};
        \node[ Node] at (1, 0) (n1) {};
        \node[ Node] at (1, 1) (n2) {};
        \node[ANode] at (0, 1) (n3) {};

        \node[ Node] at (1.866, 0.5) (m1) {};
        \node[ Node] at (2.866, 0.5) (m2) {};

        \node[Basic] at (2.4, 1.3) {+ $\kb$};
    \end{pgfonlayer}
    \begin{pgfonlayer}{background}
        \lEdge[cA]{n3}{\a}{n2}
        \lEdge[cB]{n2}{\b}{m1}
        \lEdge[cA,swap]{n0}{\a}{n1}
        \lEdge[cB,swap]{n1}{\b}{m1}
        \lEdge[cB,swap]{m1}{\b}{m2}
    \end{pgfonlayer}
    \begin{pgfonlayer}{prebackground}
        \draw[Selection] (n0.center) -- (n1.center) -- (m1.center) (n3.center) -- (n3.center);
    \end{pgfonlayer}
\end{tikzpicture}
}

\begin{center}
    \begin{tikzpicture}[scale=\scaleOfPictures]
        \node (Q) [AInstance] at (0, 0) {\addOneBxIVxQ};
        \node (B) [AResult, left=1cm of Q] {\addOneBxIVxX};
        \node (D) [AResult, right=1cm of Q] {\addOneBxIVxNO};
        \node[anchor=north east, inner xsep=0] at (B.south east) (ok1) {\CaseOk};
        \node[anchor=north east, inner xsep=0] at (D.south east) (ok2) {\CaseOk};
        \draw[CaseArrow]
        (Q)
        -- node[midway, auto,     , inner sep=0cm] {\CaseEdgeExists}
        (B);
        \draw[CaseArrow]
        (Q)
        -- node[midway, auto, swap, inner sep=0cm] {\CaseEdgeDoesNotExist}
        (D);
    \end{tikzpicture}
\end{center}
It is easy to see that in each case we get an appropriate subgraph required by the lemma.
We may thus move on to the second case.

\paragraph{Case $2^\circ$} ($\y = \a$). Color $\y$ has originally appeared in our considerations,
because we started with the assumed graph $G_1 = \underline{\a\y} + \kb$.
If $\y$ is equal to $\a$, we cannot directly use the technique from the case $1^\circ$., however -- happily -- not everything is lost.
After the instance \ref{lemma:addOneB}.3. we obtained (as a subgraph of $G_3$) the following graph: $G_1' = \underline{\a\b} + \kb$.
It enables us to repeat the whole proof of the lemma \ref{lemma:addOneB}.
with a new value of variable $\y$, now being certain, that we will succeed:
even if none of the previous instances yields the graph we want, we will necessarily fall to the
case $1^\circ$.

Above observation completes the proof of Lemma~\ref{lemma:addOneB}.
\end{proof}

Lemmas~\ref{lemma:addOneA} and~\ref{lemma:addOneB} --- in accordance to the previous remarks --- form
an inductive step that allows to easily prove the following corollary:

\begin{corollary}\label{corollary:addingVertices}
If $\G$ satisfies the condition A) of Lemma~\ref{lemma:withBranches} then for every
$k \in \mathbb{N}$ there exist colors $\a, \x \in \ColorsTwo$ such that $\G$ embeds the graph:
$$
\underline{\a\x} + \kb
$$
\end{corollary}

\noindent
We omit the simple proof.


\subsection{Adding isolated edges}\label{subsec:addingEdges}

In this part of the proof we will be showing a fact similar to the one stated in
Corollary~\ref{corollary:addingVertices}, but respecting the existence of graphs
$\underline{\a\x} + k \cdot \underline{\a} \trianglelefteq \mathbb{G}$ for some $\a, \x \in \ColorsTwo$:
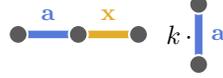
\begin{center}
    \begin{tikzpicture}[scale=\scaleOfPictures]
        \node (n1) [Node] at (0, 0) {};
        \node (n2) [Node] at (0, 1) {};

        \node (m1) [Node] at (-3, 0.5) {};
        \node (m2) [Node] at (-2, 0.5) {};
        \node (m3) [Node] at (-1, 0.5) {};

        \node (label) [Basic] at (-0.32, 0.5) {$k\,\cdot$};

        \begin{pgfonlayer}{background}
            \lEdge[cA, swap]{n1}{\a}{n2}
            \lEdge[cA]{m1}{\a}{m2}
            \lEdge[cX]{m2}{\x}{m3}
        \end{pgfonlayer}
    \end{tikzpicture}
\end{center}
This time the whole reasoning is divided into three lemmas.
Their proofs will be a bit simpler, but the way we should connect them to form a valid inductive step
will be less obvious.

\medskip

\paragraph{Notational remark}
Some parts of the statements of the three lemmas were \circledb{circled}\!\!. Those expressions are
required from the formal point of view, but in fact they make the idea behind the lemmas harder to
grasp. It should be noted that the graph $S$ present in those fragments never changes
--- the lemma 'gets' it from the assumptions and yields it in its thesis in an unchanged form.
Similarly, the discrete graphs $n \cdot \bullet$ contribute to the proof in a very simple way: each
lemma 'uses' a few their isolated vertices (constants $M_\bullet$) and returns the remaining
$(n - M_\bullet)$ vertices.
Due to that fact, when reading the lemmas, one should not pay a great attention to the circled
fragments. All we have to know is that they exist, then we may safely ignore them.

\begin{lemma}\label{lemma:addingP2xAA}
Let $\G$ be a homogeneous, 2-edge-colored graph which embeds $\underline{\a\a} + \circled{n \cdot
\bullet + S}$ for some given $n \in \mathbb{N}$ and colors $\a, \b \in \ColorsTwo$ ($\a \neq \b$).
Then, if $n \geq M_{\ref{lemma:addingP2xAA}}$, $\G$ embeds also one of the following graphs:
\begin{enumerate}
    \item $\underline{\a} + \underline{\a} + \circled{(n - M_{\ref{lemma:addingP2xAA}})\cdot\bullet + S}$,
    \item $\underline{\a\b}\hspace{0.44cm} + \circled{(n - M_{\ref{lemma:addingP2xAA}})\cdot\bullet + S}$
\end{enumerate}
for some constant $ M_{\ref{lemma:addingP2xAA}} \in \mathbb{N}$ (its precise value is not important).
\end{lemma}

\begin{lemma}\label{lemma:addingP2xAB}
Let $\G$ be a homogeneous, 2-edge-colored graph that embeds a graph $\underline{\a\b} + n \cdot
\bullet + S$ for some given $n \geq M_{\ref{lemma:addingP2xAB}}$ and colors $\a, \b \in \ColorsTwo$ ($\a \neq \b$).
Then $\G$ also embeds the graph:
$$
\underline{\x} + \underline{\y} + \circled{(n - M_{\ref{lemma:addingP2xAB}})\cdot\bullet + S}
$$
for some constant $ M_{\ref{lemma:addingP2xAB}} \in \mathbb{N}$ and colors $\x, \y \in \ColorsTwo$.
\end{lemma}

\begin{lemma}\label{lemma:P2plusP3}
Let $\G$ be a \emph{strongly} homogeneous, 2-edge-colored graph that embeds the following graphs:
\begin{enumerate}
    \item $\underline{\a\x} \hspace{0.48cm} + \circled{n \cdot \bullet + S}$,
    \item $\underline{\a} + \underline{\y}  + \circled{n \cdot \bullet + S}$,
\end{enumerate}
for some $n \geq M_{\ref{lemma:P2plusP3}}$ and colors $\a, \b, \x \in \ColorsTwo$ ($\a \neq \b$) and $\y \in \{\a, \x\}$.
Then one of the following cases holds:
\begin{enumerate}
    \item $\G$ embeds a graph $\underline{\a_1\x_1} + \underline{\a_1} + \circled{R}$ for some $\a_1, \x_1 \in \ColorsTwo$,
    \item $\G$ embeds a graph $\underline{\a\b} + \circled{R}$ and also
    embeds either $\underline{\a\a} + \underline{\b} + \circled{R}$ or
    $\underline{\b\b} + \underline{\a} + \circled{R}$.
\end{enumerate}
Above $\circled{R} = \circled{(n - M_{\ref{lemma:P2plusP3}}) \cdot \bullet + S}$, and
$M_{\ref{lemma:P2plusP3}} \in \mathbb{N}$ is --- as in previous lemmas --- some constant
resulting from the structure of the proof.
\end{lemma}

\medskip

\noindent
Proofs of the above lemmas will help us to show the following corollary:

\begin{corollary}\label{corollary:canAddP2xA}
If a strongly homogeneous, 2-edge-colored graph $\G$ satisfies Corollary~\ref{corollary:addingVertices}, i.e.,
for every  $k \in \mathbb{N}$ there exist colors $\a_0, \x_0 \in \ColorsTwo$ such that $\G$ embeds the graph $\underline{\a_0\x_0} + k \cdot
\bullet$, 
then also for every $k \in \mathbb{N}$ there exist
(potentially new) colors $\a', \x' \in \ColorsTwo$ such that $\G$ embeds the following graph:
$$
\underline{\a'\x'} + k \cdot \underline{\a'}
$$
\end{corollary}
\noindent
(We first show how we derive the above corollary from the lemmas,
and only later will we focus on proving the three lemmas.)
\newcommand{\kbb}{\circledb{\scriptsize\ensuremath{\clubsuit}}}

\begin{proof}  

The procedure of 'producing' the desired graph $\underline{\a'\x'} + k \cdot \underline{\a'}$ will be inductive.
Using it, we will be successively getting the following graphs:
\begin{align*}
    &\underline{\a_0\x_0}\hspace{0.5cm}+\ \ (3k\hspace{0.6cm}) \cdot M \cdot \bullet\\
    &\underline{\a_1\x_1}\hspace{0.5cm}+\ \ (3k-1) \cdot M \cdot \bullet\ \ +\ \ \underline{\a_1}\\
    &\underline{\a_2\x_2}\hspace{0.5cm}+\ \ (3k-2) \cdot M \cdot \bullet\ \ +\ \ \underline{\a_1}+\underline{\a_2}\\
    &\underline{\a_3\x_3}\hspace{0.5cm}+\ \ (3k-3) \cdot M \cdot \bullet\ \ +\ \ \underline{\a_1}+\underline{\a_2}+\underline{\a_3}\\
    &\dots\\
    &\underline{\a_{2k}\x_{2k}}\hspace{0.175cm}+\ \ (\rlap{k}\hspace{1.06cm}) \cdot M \cdot \bullet\ \ +\ \ \underline{\a_1}+\underline{\a_2}+\underline{\a_3}+\dots+\underline{\a_{2k}}\\
    &\dots\\
    &\underline{\a_{3k}\x_{3k}}\hspace{0.175cm}+\ \ (\rlap{0}\hspace{1.06cm}) \cdot M \cdot \bullet\ \ +\ \ \underline{\a_1}+\underline{\a_2}+\underline{\a_3}+\dots+\underline{\a_{2k}}+\dots+\underline{\a_{3k}}
\end{align*}
After repeating the inductive step $2k$ times, we will get the graph that -- apart from
the path $\underline{\a_{2k}\x_{2k}}$ -- will contain $2k$ isolated edges colored by $\a_1, \a_2, \dots, \a_{2k} \in \ColorsTwo$ respectively.
It is clear there exists a group of at least $k$ edges painted with a common color $\w$.
If $\a_{2k} = \w$ we get the thesis of the corollary --- we just found a graph:
$$
\underline{\a_{2k}\x_{2k}} + k \cdot \underline{\a_{2k}}
$$
Similarly, if in one of the next  $k$ steps we will get $\a_{2k+i} = \w$ ($i \in \{1,2,\dots, k\}$),
the requirements of the corollary are met. Otherwise, we have $\a_{2k} = \a_{2k+1} = \a_{2k+2} = \dots = \a_{3k} = \overline{\w}$ (where $\overline{\w} \in \ColorsTwo$, $\overline{\w} \neq \w$), so
we have just obtained $k$ isolated edges in a color $\overline{\w}$ together with a path $\overline{\w}\x_{3k}$.
It would complete the proof of the corollary.

It remains to show how to use the three lemmas to build the inductive step.

\paragraph{Inductive step}

At this point it is easy to guess, what was the purpose of the circled fragments
of the form $\circled{n\cdot\bullet + S}$ appearing in the lemmas:

$$
\underline{\a_{2k}\x_{2k}}\hspace{0.175cm}+\ \ \underbrace{(\rlap{k}\hspace{1.06cm}) \cdot M \cdot \bullet}_{\circled{k\cdot M \cdot \bullet}}\ \ +\ \ \underbrace{\underline{\a_1}+\underline{\a_2}+\underline{\a_3}+\dots+\underline{\a_{2k}}}_{\circled{S_{2k}}}\\
$$
The first part $n \cdot \bullet$ corresponds to a 'resource' of vertices that is used by the
lemmas to 'produce' the new edges $\underline{\a_i}$ that appear in the induction scheme we
presented earlier. In turn $S$ is a common notation for the edges that are already produced:
we begin with empty $S_0$ and after each inductive step we add one edge to it. After $i$ steps we get
$S_i = \underline{\a_1}+\underline{\a_2}+\dots+\underline{\a_i}$.
For the sake of simplicity, we \textbf{will omit} both kinds of graphs in the further considerations,
only indicating their presence with symbol $\kbb$.

\medskip

\noindent
Let us assume we have already shown that $\G$ embeds:
$$
\underline{\a\x} + \kbb
$$
Our current goal is to show, that $\G$ also embeds:
$$
\underline{\a'\x'} + \circled{\underline{\a'} + \kbb}
$$
If $\x = \a$ (so we have $\underline{\a\a} + \kbb$), we use lemma~\ref{lemma:addingP2xAA}.,
trying to show that $\G$ embeds $\underline{\a} + \underline{\a} + \kbb$. If we fail because the
second option from lemma takes place, we get the graph $\underline{\a\b} + \kbb$
(for $\b \neq \a$). It allows us to move on to the case $\x = \b$.
If $\x = \b$ (and then we have $\underline{\a\b} + \kbb$), we can now use
lemma~\ref{lemma:addingP2xAB}. In this case we will certainly get the graph
$\underline{\v} + \underline{\w} + \kbb$ (where $\v, \w \in \ColorsTwo$).

Summing the two above cases up, we may end getting one of the three graphs:
\begin{align*}
    &\underline{\a} + \underline{\a} + \kbb, &
    &\underline{\a} + \underline{\b} + \kbb, &
    &\underline{\b} + \underline{\b} + \kbb,
\end{align*}
wherein the latter two are obtained only when $\x = \b$. In other words, we now have:
\begin{enumerate}
    \item $\underline{\a\x} + \kbb$, \hfill (from assumptions)
    \item $\underline{\a} + \underline{\y} + \kbb$ for $\y \in \{\a, \x\}$. \hfill (just obtained)
\end{enumerate}
It turns out that those are exactly the assumptions of Lemma~\ref{lemma:P2plusP3}. Let us use it then.

The lemma lists two possible cases. When the first one holds, we directly get what we wanted --- the graph:
$$
\underline{\a'\x'} + \circled{\underline{\a'} + \kbb}
$$
for some $\a', \x' \in \ColorsTwo$.

The second case makes the situation a bit more complicated. Although just as we wanted, we get
two separate paths --- w.l.o.g.
$\underline{\a\a} + \circled{\underline{\b} + \kbb}$ --- but they do not share a color of some edge,
this being needed to complete the proof. We have to repeat all the steps we made so far,
adding the obtained edge $\underline{\b}$ to $\kbb$:
$$
\kbb' = \circled{\underline{\b} + \kbb}
$$
If we again end up in this 'unfortunate' second case of lemma~\ref{lemma:P2plusP3},
this time we will finally succeed closing the proof. Indeed, in that situation we will get the graph:
$$
\underline{\a\b} + \kbb' = \underline{\a\b} + \circled{\underline{\b} + \kbb}
$$
which corresponds to the graph $\underline{\a'\x'} + \circled{\underline{\b'} + \kbb}$
appearing in the induction scheme (for $\a' = \b$ and $\x' = \a$).
\end{proof}


\begin{proof}[Proof of Lemma~\ref{lemma:addingP2xAA}] 
A simple proof of this lemma consists of
two amalgamations only. To construct the first one, we use two subgraphs of graph $\underline{\a\a} + \kbb$
that is present in the assumptions:

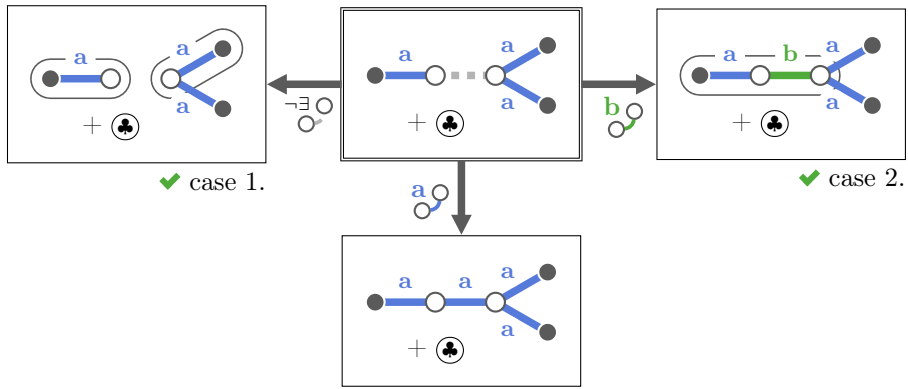
\Instance{\ref{lemma:addingP2xAA}.1} 

\newcommand{\addingPPxAAxQUESTION}{
\begin{tikzpicture}[scale=\scaleOfPictures]
    \begin{pgfonlayer}{foreground}
        \node[ Node] at (0, 0) (n1) {};
        \node[ANode] at (1, 0) (n2) {};
        \node[ANode] at (2, 0) (n3) {};

        \node[ Node] at (2.866, 0.5) (up) {};
        \node[ Node] at (2.866, -0.5) (down) {};
        \node[Basic] at (1, -0.8) {+ $\kbb$};
    \end{pgfonlayer}
    \begin{pgfonlayer}{background}
        \lEdge[cA]{n1}{\a}{n2}
        \lEdge[cA]{n3}{\a}{up}
        \lEdge[cA,swap]{n3}{\a}{down}
        \lEdge[AEdge]{n2}{}{n3}
    \end{pgfonlayer}
\end{tikzpicture}
}

\newcommand{\addingPPxAAxNO}{
\begin{tikzpicture}[scale=\scaleOfPictures]
    \begin{pgfonlayer}{foreground}
        \node[ Node] at (0, 0) (n1) {};
        \node[ANode] at (1, 0) (n2) {};
        \node[ANode] at (2, 0) (n3) {};

        \node[ Node] at (2.866, 0.5) (up) {};
        \node[ Node] at (2.866, -0.5) (down) {};
        \node[Basic] at (1, -0.8) {+ $\kbb$};
    \end{pgfonlayer}
    \begin{pgfonlayer}{background}
        \lEdge[cA]{n1}{\a}{n2}
        \lEdge[cA]{n3}{\a}{up}
        \lEdge[cA,swap]{n3}{\a}{down}
    \end{pgfonlayer}
    \begin{pgfonlayer}{prebackground}
        \draw[Selection] (n1.center) -- (n2.center) (n3.center) -- (up.center);
    \end{pgfonlayer}
\end{tikzpicture}
}

\newcommand{\addingPPxAAxA}{
\begin{tikzpicture}[scale=\scaleOfPictures]
    \begin{pgfonlayer}{foreground}
        \node[ Node] at (0, 0) (n1) {};
        \node[ANode] at (1, 0) (n2) {};
        \node[ANode] at (2, 0) (n3) {};

        \node[ Node] at (2.866, 0.5) (up) {};
        \node[ Node] at (2.866, -0.5) (down) {};
        \node[Basic] at (1, -0.8) {+ $\kbb$};
    \end{pgfonlayer}
    \begin{pgfonlayer}{background}
        \lEdge[cA]{n1}{\a}{n2}
        \lEdge[cA]{n3}{\a}{up}
        \lEdge[cA,swap]{n3}{\a}{down}
        \lEdge[cA]{n2}{\a}{n3}
    \end{pgfonlayer}
\end{tikzpicture}
}

\newcommand{\addingPPxAAxB}{
\begin{tikzpicture}[scale=\scaleOfPictures]
    \begin{pgfonlayer}{foreground}
        \node[ Node] at (0, 0) (n1) {};
        \node[ANode] at (1, 0) (n2) {};
        \node[ANode] at (2, 0) (n3) {};

        \node[ Node] at (2.866, 0.5) (up) {};
        \node[ Node] at (2.866, -0.5) (down) {};
        \node[Basic] at (1, -0.8) {+ $\kbb$};
    \end{pgfonlayer}
    \begin{pgfonlayer}{background}
        \lEdge[cA]{n1}{\a}{n2}
        \lEdge[cA]{n3}{\a}{up}
        \lEdge[cA,swap]{n3}{\a}{down}
        \lEdge[cB]{n2}{\b}{n3}
    \end{pgfonlayer}
    \begin{pgfonlayer}{prebackground}
        \draw[Selection] (n1.center) -- (n2.center) -- (n3.center);
    \end{pgfonlayer}
\end{tikzpicture}
}

\begin{center}
    \begin{tikzpicture}[scale=\scaleOfPictures]
        \node (Q) [AInstance] at (0, 0) {\addingPPxAAxQUESTION};
        \node (B) [AResult, left=1cm of Q] {\addingPPxAAxNO};
        \node (D) [AResult, right=1cm of Q] {\addingPPxAAxB};
        \node (E) [AResult, below=1cm of Q] {\addingPPxAAxA};
        \node[anchor=north east, inner xsep=0] at (B.south east) (ok1) {\CaseOk case 1.};
        \node[anchor=north east, inner xsep=0] at (D.south east) (ok2) {\CaseOk case 2.};
        \draw[CaseArrow]
        (Q)
        -- node[midway, auto,     , inner sep=0cm] {\CaseEdgeDoesNotExist}
        (B);
        \draw[CaseArrow]
        (Q)
        -- node[midway, auto, swap, inner sep=0cm] {\CaseBEdge}
        (D);
        \draw[CaseArrow]
        (Q)
        -- node[midway, auto, swap, inner sep=0cm] {\CaseAEdge}
        (E);
    \end{tikzpicture}
\end{center}
If an $\a$-edge does not appear, either case 1. or 2. of the lemma holds,
so we are done. If it does, we use two copies of the resulting graph to form the next amalgamation:

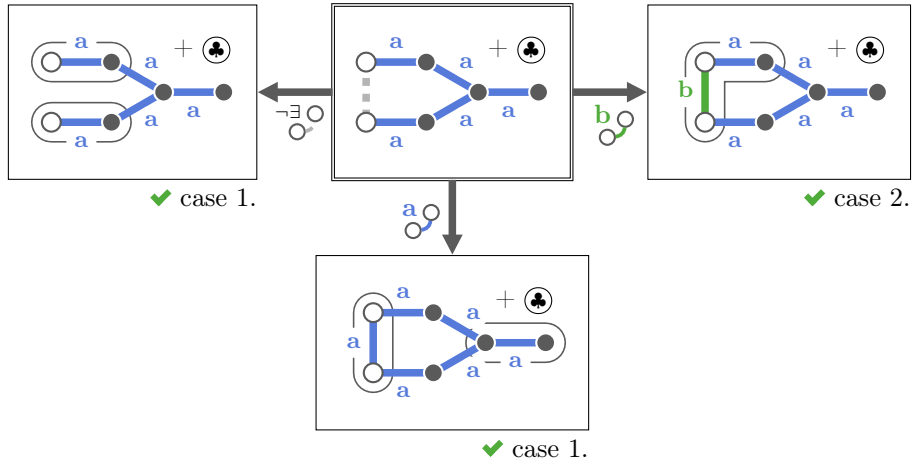
\Instance{\ref{lemma:addingP2xAA}.2}

\newcommand{\addingPPxAAxIIxQUESTION}{
\begin{tikzpicture}[scale=\scaleOfPictures]
    \begin{pgfonlayer}{foreground}
        \node[ Node] at (2, 0.5) (n1) {};
        \node[ANode] at (1, 0.5) (n2) {};
        \node[ANode] at (1, -0.5) (n3) {};
        \node[ Node] at (2, -0.5) (n4) {};

        \node[ Node] at (2.866, 0) (r1) {};
        \node[ Node] at (3.866, 0) (r2) {};
        \node[Basic] at (3.5, 0.7) {+ $\kbb$};
    \end{pgfonlayer}
    \begin{pgfonlayer}{background}
        \lEdge[cA,swap]{r1}{\a}{n1}
        \lEdge[cA,swap]{n1}{\a}{n2}
        \lEdge[cA,swap]{n3}{\a}{n4}
        \lEdge[cA,swap]{n4}{\a}{r1}
        \lEdge[cA,swap]{r1}{\a}{r2}
        \lEdge[AEdge]{n2}{}{n3}
    \end{pgfonlayer}
\end{tikzpicture}
}
\newcommand{\addingPPxAAxIIxA}{
\begin{tikzpicture}[scale=\scaleOfPictures]
    \begin{pgfonlayer}{foreground}
        \node[ Node] at (2, 0.5) (n1) {};
        \node[ANode] at (1, 0.5) (n2) {};
        \node[ANode] at (1, -0.5) (n3) {};
        \node[ Node] at (2, -0.5) (n4) {};

        \node[ Node] at (2.866, 0) (r1) {};
        \node[ Node] at (3.866, 0) (r2) {};
        \node[Basic] at (3.5, 0.7) {+ $\kbb$};
    \end{pgfonlayer}
    \begin{pgfonlayer}{background}
        \lEdge[cA,swap]{r1}{\a}{n1}
        \lEdge[cA,swap]{n1}{\a}{n2}
        \lEdge[cA,swap]{n3}{\a}{n4}
        \lEdge[cA,swap]{n4}{\a}{r1}
        \lEdge[cA,swap]{r1}{\a}{r2}
        \lEdge[cA,swap]{n2}{\a}{n3}
    \end{pgfonlayer}
    \begin{pgfonlayer}{prebackground}
        \draw[Selection] (n2.center) -- (n3.center) (r1.center) -- (r2.center);
    \end{pgfonlayer}
\end{tikzpicture}
}
\newcommand{\addingPPxAAxIIxB}{
\begin{tikzpicture}[scale=\scaleOfPictures]
    \begin{pgfonlayer}{foreground}
        \node[ Node] at (2, 0.5) (n1) {};
        \node[ANode] at (1, 0.5) (n2) {};
        \node[ANode] at (1, -0.5) (n3) {};
        \node[ Node] at (2, -0.5) (n4) {};

        \node[ Node] at (2.866, 0) (r1) {};
        \node[ Node] at (3.866, 0) (r2) {};
        \node[Basic] at (3.5, 0.7) {+ $\kbb$};
    \end{pgfonlayer}
    \begin{pgfonlayer}{background}
        \lEdge[cA,swap]{r1}{\a}{n1}
        \lEdge[cA,swap]{n1}{\a}{n2}
        \lEdge[cA,swap]{n3}{\a}{n4}
        \lEdge[cA,swap]{n4}{\a}{r1}
        \lEdge[cA,swap]{r1}{\a}{r2}
        \lEdge[cB,swap]{n2}{\b}{n3}
    \end{pgfonlayer}
    \begin{pgfonlayer}{prebackground}
        \draw[Selection] (n1.center) -- (n2.center) -- (n3.center);
    \end{pgfonlayer}
\end{tikzpicture}
}
\newcommand{\addingPPxAAxIIxNO}{
\begin{tikzpicture}[scale=\scaleOfPictures]
    \begin{pgfonlayer}{foreground}
        \node[ Node] at (2, 0.5) (n1) {};
        \node[ANode] at (1, 0.5) (n2) {};
        \node[ANode] at (1, -0.5) (n3) {};
        \node[ Node] at (2, -0.5) (n4) {};

        \node[ Node] at (2.866, 0) (r1) {};
        \node[ Node] at (3.866, 0) (r2) {};
        \node[Basic] at (3.5, 0.7) {+ $\kbb$};
    \end{pgfonlayer}
    \begin{pgfonlayer}{background}
        \lEdge[cA,swap]{r1}{\a}{n1}
        \lEdge[cA,swap]{n1}{\a}{n2}
        \lEdge[cA,swap]{n3}{\a}{n4}
        \lEdge[cA,swap]{n4}{\a}{r1}
        \lEdge[cA,swap]{r1}{\a}{r2}
    \end{pgfonlayer}
    \begin{pgfonlayer}{prebackground}
        \draw[Selection] (n1.center) -- (n2.center) (n3.center) -- (n4.center);
    \end{pgfonlayer}
\end{tikzpicture}
}

\begin{center}
    \begin{tikzpicture}[scale=\scaleOfPictures]
        \node (Q) [AInstance] at (0, 0) {\addingPPxAAxIIxQUESTION};
        \node (B) [AResult, left=1cm of Q] {\addingPPxAAxIIxNO};
        \node (D) [AResult, right=1cm of Q] {\addingPPxAAxIIxB};
        \node (E) [AResult, below=1cm of Q] {\addingPPxAAxIIxA};
        \node[anchor=north east, inner xsep=0] at (B.south east) (ok1) {\CaseOk case 1.};
        \node[anchor=north east, inner xsep=0] at (E.south east) (ok3) {\CaseOk case 1.};
        \node[anchor=north east, inner xsep=0] at (D.south east) (ok2) {\CaseOk case 2.};
        \draw[CaseArrow]
        (Q)
        -- node[midway, auto,     , inner sep=0cm] {\CaseEdgeDoesNotExist}
        (B);
        \draw[CaseArrow]
        (Q)
        -- node[midway, auto, swap, inner sep=0cm] {\CaseBEdge}
        (D);
        \draw[CaseArrow]
        (Q)
        -- node[midway, auto, swap, inner sep=0cm] {\CaseAEdge}
        (E);
    \end{tikzpicture}
\end{center}
Here, no matter what the result is, we get one of the cases stated in the lemma, what ends the proof.
(We may notice here, that for this lemma the constant $M_{\ref{lemma:addingP2xAA}}$ is equal to 2,
both isolated vertices were consumed in the first instance of amalgamation.)
\end{proof}

Proof of the next lemma is equally simple --- it is built from three amalgamations,
wherein two of them are very similar, so we omit one of them.


\begin{proof}[Proof of Lemma~\ref{lemma:addingP2xAB}] 

Now in the assumptions we have the graph
$\underline{\a\b} + \kbb \trianglelefteq \G$. We build the first instance as in the previous proof:

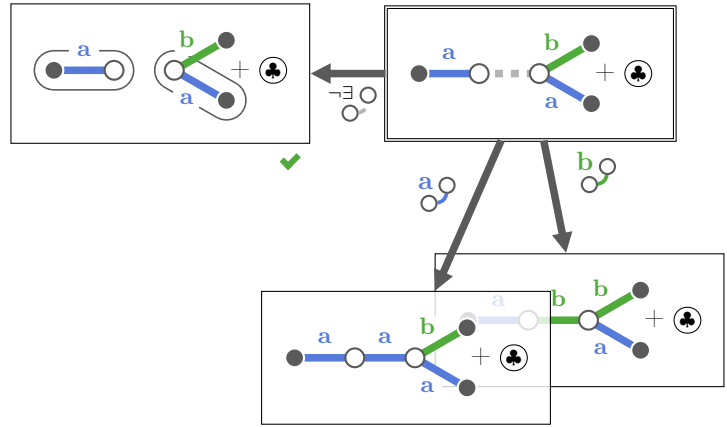
\Instance{\ref{lemma:addingP2xAB}.1} 
\newcommand{\addingPPxABxIxQUESTION}{
\begin{tikzpicture}[scale=\scaleOfPictures]
    \begin{pgfonlayer}{foreground}
        \node[ Node] at (0, 0) (n1) {};
        \node[ANode] at (1, 0) (n2) {};
        \node[ANode] at (2, 0) (n3) {};

        \node[ Node] at (2.866, 0.5) (up) {};
        \node[ Node] at (2.866, -0.5) (down) {};
        \node[Basic] at (3.4, 0) {+ $\kbb$};
    \end{pgfonlayer}
    \begin{pgfonlayer}{background}
        \lEdge[cA]{n1}{\a}{n2}
        \lEdge[cB]{n3}{\b}{up}
        \lEdge[cA,swap]{n3}{\a}{down}
        \lEdge[AEdge]{n2}{}{n3}
    \end{pgfonlayer}
\end{tikzpicture}
}
\newcommand{\addingPPxABxIxNO}{
\begin{tikzpicture}[scale=\scaleOfPictures]
    \begin{pgfonlayer}{foreground}
        \node[ Node] at (0, 0) (n1) {};
        \node[ANode] at (1, 0) (n2) {};
        \node[ANode] at (2, 0) (n3) {};

        \node[ Node] at (2.866, 0.5) (up) {};
        \node[ Node] at (2.866, -0.5) (down) {};
        \node[Basic] at (3.4, 0) {+ $\kbb$};
    \end{pgfonlayer}
    \begin{pgfonlayer}{background}
        \lEdge[cA]{n1}{\a}{n2}
        \lEdge[cB]{n3}{\b}{up}
        \lEdge[cA,swap]{n3}{\a}{down}
    \end{pgfonlayer}
    \begin{pgfonlayer}{prebackground}
        \draw[Selection] (n1.center) -- (n2.center) (n3.center) -- (down.center);
    \end{pgfonlayer}
\end{tikzpicture}
}
\newcommand{\addingPPxABxIxA}{
\begin{tikzpicture}[scale=\scaleOfPictures]
    \begin{pgfonlayer}{foreground}
        \node[ Node] at (0, 0) (n1) {};
        \node[ANode] at (1, 0) (n2) {};
        \node[ANode] at (2, 0) (n3) {};

        \node[ Node] at (2.866, 0.5) (up) {};
        \node[ Node] at (2.866, -0.5) (down) {};
        \node[Basic] at (3.4, 0) {+ $\kbb$};
    \end{pgfonlayer}
    \begin{pgfonlayer}{background}
        \lEdge[cA]{n1}{\a}{n2}
        \lEdge[cB]{n3}{\b}{up}
        \lEdge[cA,swap]{n3}{\a}{down}
        \lEdge[cA]{n2}{\a}{n3}
    \end{pgfonlayer}
\end{tikzpicture}
}
\newcommand{\addingPPxABxIxB}{
\begin{tikzpicture}[scale=\scaleOfPictures]
    \begin{pgfonlayer}{foreground}
        \node[ Node] at (0, 0) (n1) {};
        \node[ANode] at (1, 0) (n2) {};
        \node[ANode] at (2, 0) (n3) {};

        \node[ Node] at (2.866, 0.5) (up) {};
        \node[ Node] at (2.866, -0.5) (down) {};
        \node[Basic] at (3.4, 0) {+ $\kbb$};
    \end{pgfonlayer}
    \begin{pgfonlayer}{background}
        \lEdge[cA]{n1}{\a}{n2}
        \lEdge[cB]{n3}{\b}{up}
        \lEdge[cA,swap]{n3}{\a}{down}
        \lEdge[cB]{n2}{\b}{n3}
    \end{pgfonlayer}
\end{tikzpicture}
}

\begin{center}
    \begin{tikzpicture}[scale=\scaleOfPictures]
        \node (Q) [AInstance] at (0, 0) {\addingPPxABxIxQUESTION};
        \node (B) [AResult, left=1cm of Q] {\addingPPxABxIxNO};
        \node (D) [AResult, below right=1.5cm and -3.2cm of Q] {\addingPPxABxIxB};
        \node (E) [AResult, below left=2cm and -2.2cm of Q, fill=white, fill opacity=.8, text opacity=1] {\addingPPxABxIxA};
        \node[anchor=north east, inner xsep=0] at (B.south east) (ok1) {\CaseOk};
        \draw[CaseArrow]
        (Q)
        -- node[midway, auto,     , inner sep=0cm] {\CaseEdgeDoesNotExist}
        (B);
        \draw[CaseArrow]
        (Q)
        -- node[midway, auto, inner sep=0cm] {\CaseBEdge}
        (D);
        \draw[CaseArrow]
        (Q)
        -- node[midway, auto, swap, inner sep=0cm] {\CaseAEdge}
        (E);
    \end{tikzpicture}
\end{center}
In the case of nonexistent edge we get what we were looking for ---
two disjoint edges ($+\,\kbb$). However, if the edge exists, we have to use
two further instances --- for $\a$ and for $\b$. Again, they are similar, so we omit the second one.

\Instance{\ref{lemma:addingP2xAB}.2} (The third one is analogous.)

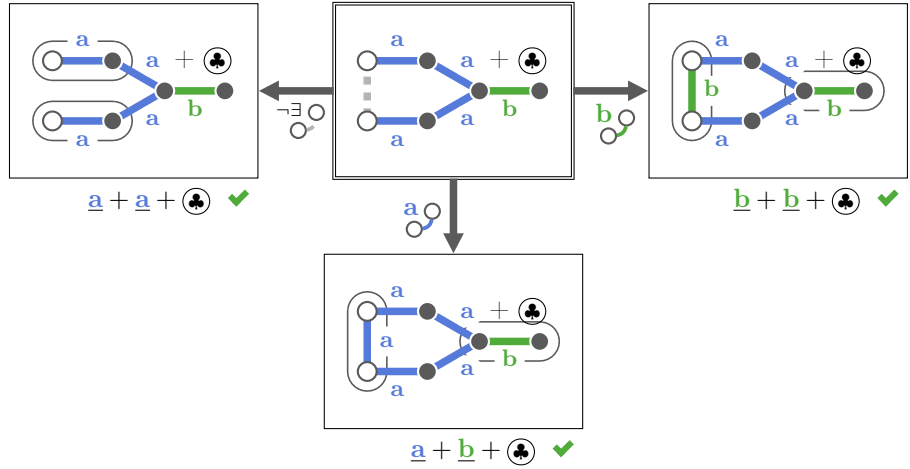
\newcommand{\addingPPxABxIIxQUESTION}{
\begin{tikzpicture}[scale=\scaleOfPictures]
    \begin{pgfonlayer}{foreground}
        \node[ Node] at (2, 0.5) (n1) {};
        \node[ANode] at (1, 0.5) (n2) {};
        \node[ANode] at (1, -0.5) (n3) {};
        \node[ Node] at (2, -0.5) (n4) {};

        \node[ Node] at (2.866, 0) (r1) {};
        \node[ Node] at (3.866, 0) (r2) {};
        \node[Basic] at (3.5, 0.5) {+ $\kbb$};
    \end{pgfonlayer}
    \begin{pgfonlayer}{background}
        \lEdge[cA,swap]{r1}{\a}{n1}
        \lEdge[cA,swap]{n1}{\a}{n2}
        \lEdge[cA,swap]{n3}{\a}{n4}
        \lEdge[cA,swap]{n4}{\a}{r1}
        \lEdge[cB,swap]{r1}{\b}{r2}
        \lEdge[AEdge]{n2}{}{n3}
    \end{pgfonlayer}
\end{tikzpicture}
}
\newcommand{\addingPPxABxIIxA}{
\begin{tikzpicture}[scale=\scaleOfPictures]
    \begin{pgfonlayer}{foreground}
        \node[ Node] at (2, 0.5) (n1) {};
        \node[ANode] at (1, 0.5) (n2) {};
        \node[ANode] at (1, -0.5) (n3) {};
        \node[ Node] at (2, -0.5) (n4) {};

        \node[ Node] at (2.866, 0) (r1) {};
        \node[ Node] at (3.866, 0) (r2) {};
        \node[Basic] at (3.5, 0.5) {+ $\kbb$};
    \end{pgfonlayer}
    \begin{pgfonlayer}{background}
        \lEdge[cA,swap]{r1}{\a}{n1}
        \lEdge[cA,swap]{n1}{\a}{n2}
        \lEdge[cA,swap]{n3}{\a}{n4}
        \lEdge[cA,swap]{n4}{\a}{r1}
        \lEdge[cB,swap]{r1}{\b}{r2}
        \lEdge[cA]{n2}{\a}{n3}
    \end{pgfonlayer}
    \begin{pgfonlayer}{prebackground}
        \draw[Selection] (n2.center) -- (n3.center) (r1.center) -- (r2.center);
    \end{pgfonlayer}
\end{tikzpicture}
}
\newcommand{\addingPPxABxIIxB}{
\begin{tikzpicture}[scale=\scaleOfPictures]
    \begin{pgfonlayer}{foreground}
        \node[ Node] at (2, 0.5) (n1) {};
        \node[ANode] at (1, 0.5) (n2) {};
        \node[ANode] at (1, -0.5) (n3) {};
        \node[ Node] at (2, -0.5) (n4) {};

        \node[ Node] at (2.866, 0) (r1) {};
        \node[ Node] at (3.866, 0) (r2) {};
        \node[Basic] at (3.5, 0.5) {+ $\kbb$};
    \end{pgfonlayer}
    \begin{pgfonlayer}{background}
        \lEdge[cA,swap]{r1}{\a}{n1}
        \lEdge[cA,swap]{n1}{\a}{n2}
        \lEdge[cA,swap]{n3}{\a}{n4}
        \lEdge[cA,swap]{n4}{\a}{r1}
        \lEdge[cB,swap]{r1}{\b}{r2}
        \lEdge[cB]{n2}{\b}{n3}
    \end{pgfonlayer}
    \begin{pgfonlayer}{prebackground}
        \draw[Selection] (n2.center) -- (n3.center) (r1.center) -- (r2.center);
    \end{pgfonlayer}
\end{tikzpicture}
}
\newcommand{\addingPPxABxIIxNO}{
\begin{tikzpicture}[scale=\scaleOfPictures]
    \begin{pgfonlayer}{foreground}
        \node[ Node] at (2, 0.5) (n1) {};
        \node[ANode] at (1, 0.5) (n2) {};
        \node[ANode] at (1, -0.5) (n3) {};
        \node[ Node] at (2, -0.5) (n4) {};

        \node[ Node] at (2.866, 0) (r1) {};
        \node[ Node] at (3.866, 0) (r2) {};
        \node[Basic] at (3.5, 0.5) {+ $\kbb$};
    \end{pgfonlayer}
    \begin{pgfonlayer}{background}
        \lEdge[cA,swap]{r1}{\a}{n1}
        \lEdge[cA,swap]{n1}{\a}{n2}
        \lEdge[cA,swap]{n3}{\a}{n4}
        \lEdge[cA,swap]{n4}{\a}{r1}
        \lEdge[cB,swap]{r1}{\b}{r2}
    \end{pgfonlayer}
    \begin{pgfonlayer}{prebackground}
        \draw[Selection] (n1.center) -- (n2.center) (n3.center) -- (n4.center);
    \end{pgfonlayer}
\end{tikzpicture}
}

\begin{center}
    \begin{tikzpicture}[scale=\scaleOfPictures]
        \node (Q) [AInstance] at (0, 0) {\addingPPxABxIIxQUESTION};
        \node (B) [AResult, left=1cm of Q] {\addingPPxABxIIxNO};
        \node (D) [AResult, right=1cm of Q] {\addingPPxABxIIxB};
        \node (E) [AResult, below=1cm of Q] {\addingPPxABxIIxA};
        \node[anchor=north east, inner xsep=0] at (B.south east) (ok1) {$\underline{\a} + \underline{\a} + \kbb$ \CaseOk};
        \node[anchor=north east, inner xsep=0] at (E.south east) (ok3) {$\underline{\a} + \underline{\b} + \kbb$ \CaseOk};
        \node[anchor=north east, inner xsep=0] at (D.south east) (ok2) {$\underline{\b} + \underline{\b} + \kbb$ \CaseOk};
        \draw[CaseArrow]
        (Q)
        -- node[midway, auto,     , inner sep=0cm] {\CaseEdgeDoesNotExist}
        (B);
        \draw[CaseArrow]
        (Q)
        -- node[midway, auto, swap, inner sep=0cm] {\CaseBEdge}
        (D);
        \draw[CaseArrow]
        (Q)
        -- node[midway, auto, swap, inner sep=0cm] {\CaseAEdge}
        (E);
    \end{tikzpicture}
\end{center}
This ends the proof of the lemma, since in all above cases we get the subgraphs we need.
\end{proof}

\noindent
Now the only remaining part is the proof of Lemma~\ref{lemma:P2plusP3}.

\begin{proof}[Proof of Lemma~\ref{lemma:P2plusP3}] 
From the assumptions we get the subgraphs $\underline{\a\x} + \kbb \trianglelefteq \G$ and
$\underline{\a} + \underline{\y}  + \kbb \trianglelefteq \G$ for some colors $\a, \b, \x \in \ColorsTwo$ ($\a \neq
\b$) and $\y \in \{\a, \x\}$.


We will start by considering the following instance:


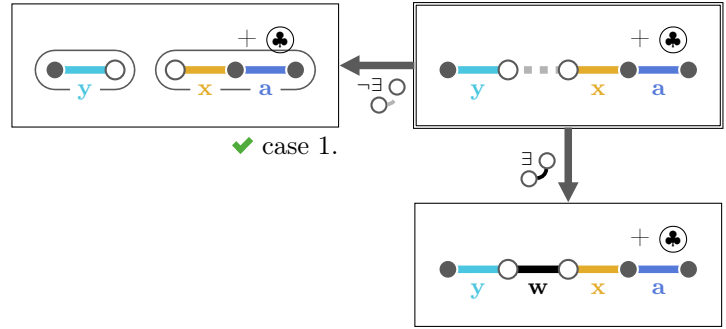
\Instance{\ref{lemma:P2plusP3}.1}

\newcommand{\addingPPPxIxQUESTION}{
\begin{tikzpicture}[scale=1*\scaleOfPictures]
    \begin{pgfonlayer}{foreground}
        \node[ Node] at (0, 0) (n1) {};
        \node[ANode] at (1, 0) (n2) {};
        \node[ANode] at (2, 0) (n3) {};
        \node[ Node] at (3, 0) (n4) {};
        \node[ Node] at (4, 0) (n5) {};
        \node[Basic] at (3.5, 0.5) {+ $\kbb$};
    \end{pgfonlayer}
    \begin{pgfonlayer}{background}
        \lEdge[cY,swap]{n1}{\y}{n2}
        \lEdge[AEdge]{n2}{}{n3}
        \lEdge[cX,swap]{n3}{\x}{n4}
        \lEdge[cA,swap]{n4}{\a}{n5}
    \end{pgfonlayer}
\end{tikzpicture}
}
\newcommand{\addingPPPxIxX}{
\begin{tikzpicture}[scale=1*\scaleOfPictures]
    \begin{pgfonlayer}{foreground}
        \node[ Node] at (0, 0) (n1) {};
        \node[ANode] at (1, 0) (n2) {};
        \node[ANode] at (2, 0) (n3) {};
        \node[ Node] at (3, 0) (n4) {};
        \node[ Node] at (4, 0) (n5) {};
        \node[Basic] at (3.5, 0.5) {+ $\kbb$};
    \end{pgfonlayer}
    \begin{pgfonlayer}{background}
        \lEdge[cY,swap]{n1}{\y}{n2}
        \lEdge[cBlack, swap]{n2}{\w}{n3}
        \lEdge[cX,swap]{n3}{\x}{n4}
        \lEdge[cA,swap]{n4}{\a}{n5}
    \end{pgfonlayer}
\end{tikzpicture}
}
\newcommand{\addingPPPxIxNO}{
\begin{tikzpicture}[scale=1*\scaleOfPictures]
    \begin{pgfonlayer}{foreground}
        \node[ Node] at (0, 0) (n1) {};
        \node[ANode] at (1, 0) (n2) {};
        \node[ANode] at (2, 0) (n3) {};
        \node[ Node] at (3, 0) (n4) {};
        \node[ Node] at (4, 0) (n5) {};
        \node[Basic] at (3.5, 0.5) {+ $\kbb$};
    \end{pgfonlayer}
    \begin{pgfonlayer}{background}
        \lEdge[cY,swap]{n1}{\y}{n2}
        \lEdge[cX,swap]{n3}{\x}{n4}
        \lEdge[cA,swap]{n4}{\a}{n5}
    \end{pgfonlayer}
    \begin{pgfonlayer}{prebackground}
        \draw[Selection] (n1.center) -- (n2.center) (n3.center) -- (n4.center) -- (n5.center);
    \end{pgfonlayer}
\end{tikzpicture}
}

\begin{center}
    \begin{tikzpicture}[scale=\scaleOfPictures]
        \node (Q) [AInstance] at (0, 0) {\addingPPPxIxQUESTION};
        \node (B) [AResult, left=1cm of Q] {\addingPPPxIxNO};
        \node (E) [AResult, below=1cm of Q] {\addingPPPxIxX};
        \node[anchor=north east, inner xsep=0] at (B.south east) (ok1) {\CaseOk case 1.};
        \draw[CaseArrow]
        (Q)
        -- node[midway, auto,     , inner sep=0cm] {\CaseEdgeDoesNotExist}
        (B);
        \draw[CaseArrow]
        (Q)
        -- node[midway, auto, swap, inner sep=0cm] {\CaseEdgeExists}
        (E);
    \end{tikzpicture}
\end{center}
If in its solution the edge will not emerge, we get appropriate graph: since
$\y \in \{\a, \x\}$, we know that $\y$ will appear somewhere on the path $\underline{\a\x}$,
and this suffices to fulfill the case 1. of the lemma we are proving.

Otherwise we get a path $P = \underline{\y\w\x\a} + \kbb$. We will consider two cases,
depending on whether it has the form $W = \underline{\a\b\n\n} + \kbb$ or not.

We should first observe, that there is only one case when $P$ does not match $W$. Indeed: When $\x =
\b$, $P$ is bound to have the form $W$. In the other case $P$ takes the shape $\underline{\a\w\a\a}
+ \kbb$, since $\y \in \{ \a, \x \}$, and yet now $\x = \a$. It follow immediately that the only
case when $P$ is not of the form $W$ is $P = \underline{\a\a\a\a} + \kbb$.

\paragraph{Case $1^\circ$} ($P = \underline{\a\b\mathbf{\alpha\beta}} + \kbb$, where $\mathbf{\alpha}, \mathbf{\beta} \in \ColorsTwo$)
Here, the only amalgamation instance is built as follows:

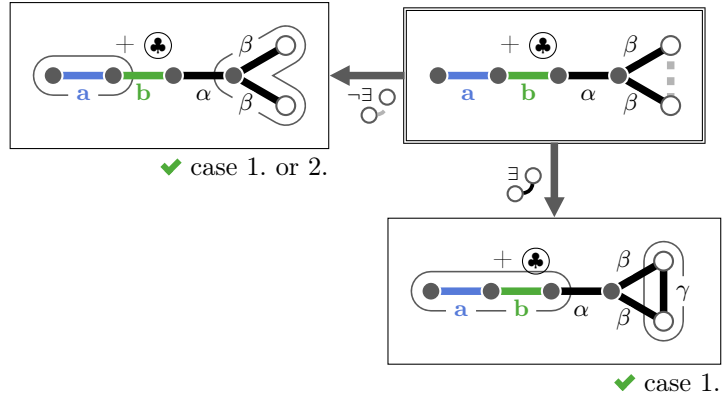
\Instance{\ref{lemma:P2plusP3}.2}

\newcommand{\addingPPPxIIxQUESTION}{
\begin{tikzpicture}[scale=1*\scaleOfPictures]
    \begin{pgfonlayer}{foreground}
        \node[ Node] at (0, 0) (n1) {};
        \node[ Node] at (1, 0) (n2) {};
        \node[ Node] at (2, 0) (n3) {};
        \node[ Node] at (3, 0) (n4) {};
        \node[ANode] at (3.866, 0.5) (n5) {};
        \node[ANode] at (3.866, -0.5) (down) {};
        \node[Basic] at (1.5, 0.5) {+ $\kbb$};
    \end{pgfonlayer}
    \begin{pgfonlayer}{background}
        \lEdge[cA,swap]{n1}{\a}{n2}
        \lEdge[cB,swap]{n2}{\b}{n3}
        \lEdge[cBlack,swap]{n3}{$\mathbf{\alpha}$}{n4}
        \lEdge[cBlack]{n4}{$\mathbf{\beta}$}{n5}
        \lEdge[cBlack,swap]{n4}{$\mathbf{\beta}$}{down}
        \lEdge[AEdge]{n5}{}{down}
    \end{pgfonlayer}
\end{tikzpicture}
}
\newcommand{\addingPPPxIIxX}{
\begin{tikzpicture}[scale=1*\scaleOfPictures]
    \begin{pgfonlayer}{foreground}
        \node[ Node] at (0, 0) (n1) {};
        \node[ Node] at (1, 0) (n2) {};
        \node[ Node] at (2, 0) (n3) {};
        \node[ Node] at (3, 0) (n4) {};
        \node[ANode] at (3.866, 0.5) (n5) {};
        \node[ANode] at (3.866, -0.5) (down) {};
        \node[Basic] at (1.5, 0.5) {+ $\kbb$};
    \end{pgfonlayer}
    \begin{pgfonlayer}{background}
        \lEdge[cA,swap]{n1}{\a}{n2}
        \lEdge[cB,swap]{n2}{\b}{n3}
        \lEdge[cBlack,swap]{n3}{$\mathbf{\alpha}$}{n4}
        \lEdge[cBlack]{n4}{$\mathbf{\beta}$}{n5}
        \lEdge[cBlack,swap]{n4}{$\mathbf{\beta}$}{down}
        \lEdge[cBlack]{n5}{$\gamma$}{down}
    \end{pgfonlayer}
    \begin{pgfonlayer}{prebackground}
        \draw[Selection] (n1.center) -- (n2.center) -- (n3.center) (n5.center) -- (down.center);
    \end{pgfonlayer}
\end{tikzpicture}
}
\newcommand{\addingPPPxIIxNO}{
\begin{tikzpicture}[scale=1*\scaleOfPictures]
    \begin{pgfonlayer}{foreground}
        \node[ Node] at (0, 0) (n1) {};
        \node[ Node] at (1, 0) (n2) {};
        \node[ Node] at (2, 0) (n3) {};
        \node[ Node] at (3, 0) (n4) {};
        \node[ANode] at (3.866, 0.5) (n5) {};
        \node[ANode] at (3.866, -0.5) (down) {};
        \node[Basic] at (1.5, 0.5) {+ $\kbb$};
    \end{pgfonlayer}
    \begin{pgfonlayer}{background}
        \lEdge[cA,swap]{n1}{\a}{n2}
        \lEdge[cB,swap]{n2}{\b}{n3}
        \lEdge[cBlack,swap]{n3}{$\mathbf{\alpha}$}{n4}
        \lEdge[cBlack]{n4}{$\mathbf{\beta}$}{n5}
        \lEdge[cBlack,swap]{n4}{$\mathbf{\beta}$}{down}
    \end{pgfonlayer}
    \begin{pgfonlayer}{prebackground}
        \draw[Selection] (down.center) -- (n4.center) -- (n5.center) (n1.center) -- (n2.center);
    \end{pgfonlayer}
\end{tikzpicture}
}

\begin{center}
    \begin{tikzpicture}[scale=\scaleOfPictures]
        \node (Q) [AInstance] at (0, 0) {\addingPPPxIIxQUESTION};
        \node (B) [AResult, left=1cm of Q] {\addingPPPxIIxNO};
        \node (E) [AResult, below=1cm of Q] {\addingPPPxIIxX};
        \node[anchor=north east, inner xsep=0] at (B.south east) (ok1) {\CaseOk case 1. or 2.};
        \node[anchor=north east, inner xsep=0] at (E.south east) (ok1) {\CaseOk case 1.};
        \draw[CaseArrow]
        (Q)
        -- node[midway, auto,     , inner sep=0cm] {\CaseEdgeDoesNotExist}
        (B);
        \draw[CaseArrow]
        (Q)
        -- node[midway, auto, swap, inner sep=0cm] {\CaseEdgeExists}
        (E);
    \end{tikzpicture}
\end{center}
If as a result of amalgamation we get an edge, we may easily fulfill case 1.
of our lemma --- the only thing we need is that color $\underline{\gamma}$ appears on the path $\underline{\a\b}$,
and this of course is happening, since $\gamma \in \{\a, \b\} = \ColorsTwo$.

If in turn the edge was not produced, we get (as a subgraph): $\underline{\beta\beta} + \underline{\a} + \kbb$.
Now, depending on the value of $\beta$, either case 1. or 2. is fulfilled. Indeed, when
$\beta = \a$ we obtain the subgraph $\underline{\a\a} + \underline{\a} + \kbb$ and case 1. of the lemma holds.
When we get $\beta = \b$, then (together with graph $\underline{\a\b} + \kbb \trianglelefteq P$) we have all
what is needed for case 2. of the lemma.

\paragraph{Case $2^\circ$} ($P = \underline{\a\a\a\a} + \kbb$)
Here, the simple amalgamation instance similar to the one from case $1^\circ$ (picture omitted) completes the
proof only in cases \raisebox{-0.2cm}{\CaseEdgeDoesNotExist} and \raisebox{-0.2cm}{\CaseAEdge}. If instead
we got the following result
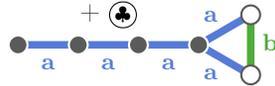
\begin{center}
    \begin{tikzpicture}[scale=\scaleOfPictures]
        \begin{pgfonlayer}{foreground}
            \node[ Node] at (0, 0) (n1) {};
            \node[ Node] at (1, 0) (n2) {};
            \node[ Node] at (2, 0) (n3) {};
            \node[ Node] at (3, 0) (n4) {};
            \node[ANode] at (3.866, 0.5) (n5) {};
            \node[ANode] at (3.866, -0.5) (down) {};
            \node[Basic] at (1.5, 0.5) {+ $\kbb$};
        \end{pgfonlayer}
        \begin{pgfonlayer}{background}
            \lEdge[cA,swap]{n1}{\a}{n2}
            \lEdge[cA,swap]{n2}{\a}{n3}
            \lEdge[cA,swap]{n3}{\a}{n4}
            \lEdge[cA]{n4}{\a}{n5}
            \lEdge[cA,swap]{n4}{\a}{down}
            \lEdge[cB]{n5}{\b}{down}
        \end{pgfonlayer}
    \end{tikzpicture}
\end{center}
we cannot use it for case 1. of the lemma, and to satisfy the case 2. an additional graph
$\heartsuit = \underline{\a\b} + \kbb$ is required.
Another sequence of amalgamations awaits --- four extra instances will be needed.

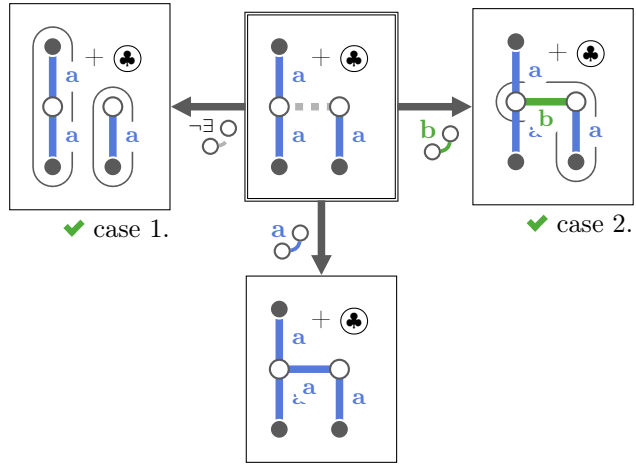
\Instance{\ref{lemma:P2plusP3}.3}

\newcommand{\addingPPPxMISSINGxQUESTION}{
\begin{tikzpicture}[scale=\scaleOfPictures]
    \begin{pgfonlayer}{foreground}
        \node[ Node] at (0, 0) (n1) {};
        \node[ANode] at (0, 1) (n2) {};
        \node[ Node] at (0, 2) (n3) {};

        \node[ Node] at (1, 0) (n4) {};
        \node[ANode] at (1, 1) (n5) {};
        \node[Basic] at (1, 1.8) {+ $\kbb$};
    \end{pgfonlayer}
    \begin{pgfonlayer}{background}
        \lEdge[cA,swap]{n1}{\a}{n2}
        \lEdge[AEdge,swap]{n2}{}{n5}
        \lEdge[cA,swap]{n2}{\a}{n3}
        \lEdge[cA,swap]{n4}{\a}{n5}
    \end{pgfonlayer}
\end{tikzpicture}
}
\newcommand{\addingPPPxMISSINGxB}{
\begin{tikzpicture}[scale=\scaleOfPictures]
    \begin{pgfonlayer}{foreground}
        \node[ Node] at (0, 0) (n1) {};
        \node[ANode] at (0, 1) (n2) {};
        \node[ Node] at (0, 2) (n3) {};

        \node[ Node] at (1, 0) (n4) {};
        \node[ANode] at (1, 1) (n5) {};
        \node[Basic] at (1, 1.8) {+ $\kbb$};
    \end{pgfonlayer}
    \begin{pgfonlayer}{background}
        \lEdge[cA,swap]{n1}{\a}{n2}
        \lEdge[cB,swap]{n2}{\b}{n5}
        \lEdge[cA,swap]{n2}{\a}{n3}
        \lEdge[cA,swap]{n4}{\a}{n5}
    \end{pgfonlayer}
    \begin{pgfonlayer}{prebackground}
        \draw[Selection] (n2.center) -- (n5.center) -- (n4.center);
    \end{pgfonlayer}
\end{tikzpicture}
}
\newcommand{\addingPPPxMISSINGxNO}{
\begin{tikzpicture}[scale=\scaleOfPictures]
    \begin{pgfonlayer}{foreground}
        \node[ Node] at (0, 0) (n1) {};
        \node[ANode] at (0, 1) (n2) {};
        \node[ Node] at (0, 2) (n3) {};

        \node[ Node] at (1, 0) (n4) {};
        \node[ANode] at (1, 1) (n5) {};
        \node[Basic] at (1, 1.8) {+ $\kbb$};
    \end{pgfonlayer}
    \begin{pgfonlayer}{background}
        \lEdge[cA,swap]{n1}{\a}{n2}
        \lEdge[cA,swap]{n2}{\a}{n3}
        \lEdge[cA,swap]{n4}{\a}{n5}
    \end{pgfonlayer}
    \begin{pgfonlayer}{prebackground}
        \draw[Selection] (n1.center) -- (n2.center) -- (n3.center) (n4.center) -- (n5.center);
    \end{pgfonlayer}
\end{tikzpicture}
}
\newcommand{\addingPPPxMISSINGxA}{
\begin{tikzpicture}[scale=\scaleOfPictures]
    \begin{pgfonlayer}{foreground}
        \node[ Node] at (0, 0) (n1) {};
        \node[ANode] at (0, 1) (n2) {};
        \node[ Node] at (0, 2) (n3) {};

        \node[ Node] at (1, 0) (n4) {};
        \node[ANode] at (1, 1) (n5) {};
        \node[Basic] at (1, 1.8) {+ $\kbb$};
    \end{pgfonlayer}
    \begin{pgfonlayer}{background}
        \lEdge[cA,swap]{n1}{\a}{n2}
        \lEdge[cA,swap]{n2}{\a}{n5}
        \lEdge[cA,swap]{n2}{\a}{n3}
        \lEdge[cA,swap]{n4}{\a}{n5}
    \end{pgfonlayer}
\end{tikzpicture}
}

\begin{center}
    \begin{tikzpicture}[scale=\scaleOfPictures]
        \node (Q) [AInstance] at (0, 0) {\addingPPPxMISSINGxQUESTION};
        \node (B) [AResult, left=1cm of Q] {\addingPPPxMISSINGxNO};
        \node (D) [AResult, right=1cm of Q] {\addingPPPxMISSINGxB};
        \node (E) [AResult, below=1cm of Q] {\addingPPPxMISSINGxA};
        \node[anchor=north east, inner xsep=0] at (B.south east) (ok1) {\CaseOk case 1.};
        \node[anchor=north east, inner xsep=0] at (D.south east) (ok1) {\CaseOk case 2.};
        \draw[CaseArrow]
        (Q)
        -- node[midway, auto,     , inner sep=0cm] {\CaseEdgeDoesNotExist}
        (B);
        \draw[CaseArrow]
        (Q)
        -- node[midway, auto, swap, inner sep=0cm] {\CaseBEdge}
        (D);
        \draw[CaseArrow]
        (Q)
        -- node[midway, auto, swap, inner sep=0cm] {\CaseAEdge}
        (E);
    \end{tikzpicture}
\end{center}
Above, when the edge does not exist case~1. of lemma easily follows.
If in turn we get a $\b$-edge, there appears the graph~$\heartsuit$ we are searching for,
allowing to meet the requirements of case~2. Let us assume then, that we got an $\a$-edge.

\Instance{\ref{lemma:P2plusP3}.4} The graph we just obtained allows to build the following instance:

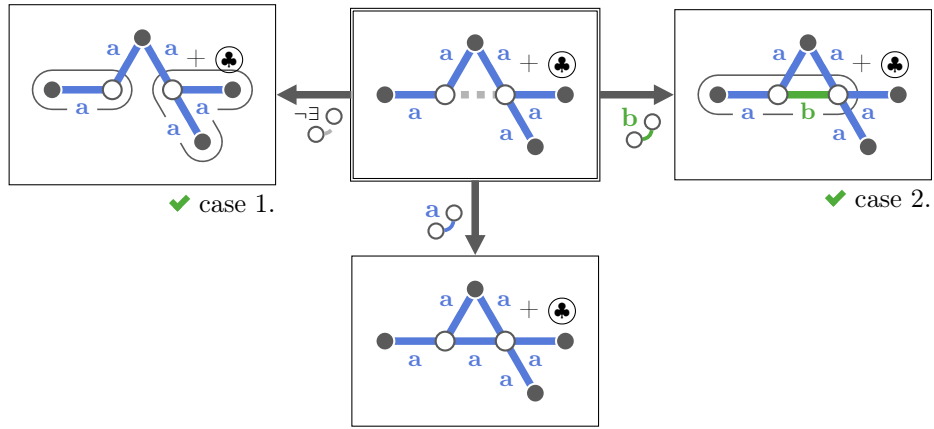
\newcommand{\addingPPPxIIIxQUESTION}{
\begin{tikzpicture}[scale=\scaleOfPictures]
    \begin{pgfonlayer}{foreground}
        \node[ Node] at (0, 0) (n1) {};
        \node[ANode] at (1, 0) (n2) {};
        \node[ANode] at (2, 0) (n3) {};
        \node[ Node] at (3, 0) (n4) {};
        \node[ Node] at (2.5, -0.866) (n5) {};
        \node[ Node] at (1.5, 0.866) (up) {};
        \node[Basic] at (2.7, 0.5) {+ $\kbb$};
    \end{pgfonlayer}
    \begin{pgfonlayer}{background}
        \lEdge[cA,swap]{n1}{\a}{n2}
        \lEdge[AEdge,swap]{n2}{}{n3}
        \lEdge[cA,swap]{n3}{\a}{n4}
        \lEdge[cA,swap]{n3}{\a}{n5}

        \lEdge[cA]{n2}{\a}{up}
        \lEdge[cA]{up}{\a}{n3}
    \end{pgfonlayer}
\end{tikzpicture}
}
\newcommand{\addingPPPxIIIxB}{
\begin{tikzpicture}[scale=\scaleOfPictures]
    \begin{pgfonlayer}{foreground}
        \node[ Node] at (0, 0) (n1) {};
        \node[ANode] at (1, 0) (n2) {};
        \node[ANode] at (2, 0) (n3) {};
        \node[ Node] at (3, 0) (n4) {};
        \node[ Node] at (2.5, -0.866) (n5) {};
        \node[ Node] at (1.5, 0.866) (up) {};
        \node[Basic] at (2.7, 0.5) {+ $\kbb$};
    \end{pgfonlayer}
    \begin{pgfonlayer}{background}
        \lEdge[cA,swap]{n1}{\a}{n2}
        \lEdge[cB,swap]{n2}{\b}{n3}
        \lEdge[cA,swap]{n3}{\a}{n4}
        \lEdge[cA,swap]{n3}{\a}{n5}

        \lEdge[cA]{n2}{\a}{up}
        \lEdge[cA]{up}{\a}{n3}
    \end{pgfonlayer}
    \begin{pgfonlayer}{prebackground}
        \draw[Selection] (n1.center) -- (n2.center) -- (n3.center);
    \end{pgfonlayer}
\end{tikzpicture}
}
\newcommand{\addingPPPxIIIxA}{
\begin{tikzpicture}[scale=\scaleOfPictures]
    \begin{pgfonlayer}{foreground}
        \node[ Node] at (0, 0) (n1) {};
        \node[ANode] at (1, 0) (n2) {};
        \node[ANode] at (2, 0) (n3) {};
        \node[ Node] at (3, 0) (n4) {};
        \node[ Node] at (2.5, -0.866) (n5) {};
        \node[ Node] at (1.5, 0.866) (up) {};
        \node[Basic] at (2.7, 0.5) {+ $\kbb$};
    \end{pgfonlayer}
    \begin{pgfonlayer}{background}
        \lEdge[cA,swap]{n1}{\a}{n2}
        \lEdge[cA,swap]{n2}{\a}{n3}
        \lEdge[cA,swap]{n3}{\a}{n4}
        \lEdge[cA,swap]{n3}{\a}{n5}

        \lEdge[cA]{n2}{\a}{up}
        \lEdge[cA]{up}{\a}{n3}
    \end{pgfonlayer}
\end{tikzpicture}
}
\newcommand{\addingPPPxIIIxNO}{
\begin{tikzpicture}[scale=\scaleOfPictures]
    \begin{pgfonlayer}{foreground}
        \node[ Node] at (0, 0) (n1) {};
        \node[ANode] at (1, 0) (n2) {};
        \node[ANode] at (2, 0) (n3) {};
        \node[ Node] at (3, 0) (n4) {};
        \node[ Node] at (2.5, -0.866) (n5) {};
        \node[ Node] at (1.5, 0.866) (up) {};
        \node[Basic] at (2.7, 0.5) {+ $\kbb$};
    \end{pgfonlayer}
    \begin{pgfonlayer}{background}
        \lEdge[cA,swap]{n1}{\a}{n2}
        \lEdge[cA,swap]{n3}{\a}{n4}
        \lEdge[cA,swap]{n3}{\a}{n5}

        \lEdge[cA]{n2}{\a}{up}
        \lEdge[cA]{up}{\a}{n3}
    \end{pgfonlayer}
    \begin{pgfonlayer}{prebackground}
        \draw[Selection] (n1.center) -- (n2.center) (n4.center) -- (n3.center) -- (n5.center);
    \end{pgfonlayer}
\end{tikzpicture}
}

\begin{center}
    \begin{tikzpicture}[scale=\scaleOfPictures]
        \node (Q) [AInstance] at (0, 0) {\addingPPPxIIIxQUESTION};
        \node (B) [AResult, left=1cm of Q] {\addingPPPxIIIxNO};
        \node (D) [AResult, right=1cm of Q] {\addingPPPxIIIxB};
        \node (E) [AResult, below=1cm of Q] {\addingPPPxIIIxA};
        \node[anchor=north east, inner xsep=0] at (B.south east) (ok1) {\CaseOk case 1.};
        \node[anchor=north east, inner xsep=0] at (D.south east) (ok1) {\CaseOk case 2.};
        \draw[CaseArrow]
        (Q)
        -- node[midway, auto,     , inner sep=0cm] {\CaseEdgeDoesNotExist}
        (B);
        \draw[CaseArrow]
        (Q)
        -- node[midway, auto, swap, inner sep=0cm] {\CaseBEdge}
        (D);
        \draw[CaseArrow]
        (Q)
        -- node[midway, auto, swap, inner sep=0cm] {\CaseAEdge}
        (E);
    \end{tikzpicture}
\end{center}
As before, the lack of an edge of the appearance of a $\b$-edge lead us straight to the cases 1. or 2.
Again, we assume we unluckily got an $\a$-edge.

\Instance{\ref{lemma:P2plusP3}.5}
From the result of previous instance we take the subgraph $\circ\a\a\a + \kbb$, and, pairing it with the
graph
$\underline{\a} + \underline{\a} + \kbb$, we build an amalgamation as follows:

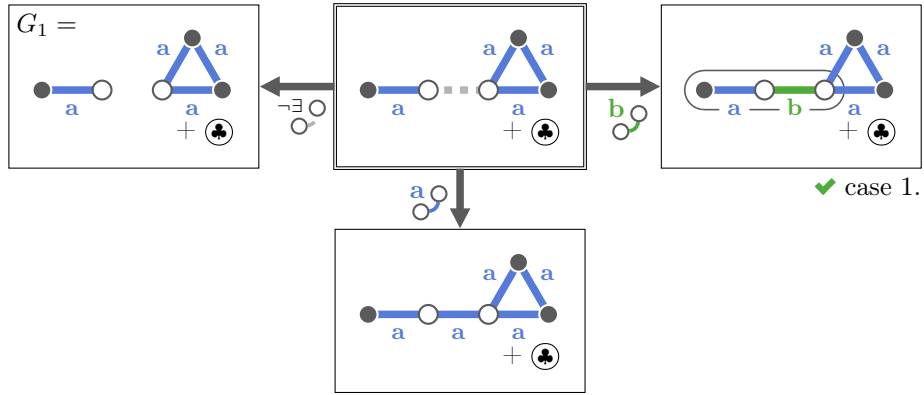
\newcommand{\addingPPPxIVxQUESTION}{
\begin{tikzpicture}[scale=\scaleOfPictures]
    \begin{pgfonlayer}{foreground}
        \node[ Node] at (-1, 0) (n0) {};
        \node[ANode] at (0, 0) (n1) {};
        \node[ANode] at (1, 0) (n2) {};
        \node[ Node] at (2, 0) (n3) {};
        \node[ Node] at (1.5, 0.866) (up) {};
        \node[Basic] at (1.7, -0.7) {+ $\kbb$};
    \end{pgfonlayer}
    \begin{pgfonlayer}{background}
        \lEdge[cA,swap]{n0}{\a}{n1}
        \lEdge[AEdge,swap]{n1}{}{n2}
        \lEdge[cA,swap]{n2}{\a}{n3}
        \lEdge[cA]{n2}{\a}{up}
        \lEdge[cA]{up}{\a}{n3}
    \end{pgfonlayer}
\end{tikzpicture}
}
\newcommand{\addingPPPxIVxB}{
\begin{tikzpicture}[scale=\scaleOfPictures]
    \begin{pgfonlayer}{foreground}
        \node[ Node] at (-1, 0) (n0) {};
        \node[ANode] at (0, 0) (n1) {};
        \node[ANode] at (1, 0) (n2) {};
        \node[ Node] at (2, 0) (n3) {};
        \node[ Node] at (1.5, 0.866) (up) {};
        \node[Basic] at (1.7, -0.7) {+ $\kbb$};
    \end{pgfonlayer}
    \begin{pgfonlayer}{background}
        \lEdge[cA,swap]{n0}{\a}{n1}
        \lEdge[cB,swap]{n1}{\b}{n2}
        \lEdge[cA,swap]{n2}{\a}{n3}
        \lEdge[cA]{n2}{\a}{up}
        \lEdge[cA]{up}{\a}{n3}
    \end{pgfonlayer}
    \begin{pgfonlayer}{prebackground}
        \draw[Selection] (n0.center) -- (n1.center) -- (n2.center);
    \end{pgfonlayer}
\end{tikzpicture}
}
\newcommand{\addingPPPxIVxA}{
\begin{tikzpicture}[scale=\scaleOfPictures]
    \begin{pgfonlayer}{foreground}
        \node[ Node] at (-1, 0) (n0) {};
        \node[ANode] at (0, 0) (n1) {};
        \node[ANode] at (1, 0) (n2) {};
        \node[ Node] at (2, 0) (n3) {};
        \node[ Node] at (1.5, 0.866) (up) {};
        \node[Basic] at (1.7, -0.7) {+ $\kbb$};
    \end{pgfonlayer}
    \begin{pgfonlayer}{background}
        \lEdge[cA,swap]{n0}{\a}{n1}
        \lEdge[cA,swap]{n1}{\a}{n2}
        \lEdge[cA,swap]{n2}{\a}{n3}
        \lEdge[cA]{n2}{\a}{up}
        \lEdge[cA]{up}{\a}{n3}
    \end{pgfonlayer}
\end{tikzpicture}
}
\newcommand{\addingPPPxIVxNO}{
\begin{tikzpicture}[scale=\scaleOfPictures]
    \begin{pgfonlayer}{foreground}
        \node[ Node] at (-1, 0) (n0) {};
        \node[ANode] at (0, 0) (n1) {};
        \node[ANode] at (1, 0) (n2) {};
        \node[ Node] at (2, 0) (n3) {};
        \node[ Node] at (1.5, 0.866) (up) {};
        \node[Basic] at (1.7, -0.7) {+ $\kbb$};
    \end{pgfonlayer}
    \begin{pgfonlayer}{background}
        \lEdge[cA,swap]{n0}{\a}{n1}
        \lEdge[cA,swap]{n2}{\a}{n3}
        \lEdge[cA]{n2}{\a}{up}
        \lEdge[cA]{up}{\a}{n3}
    \end{pgfonlayer}
\end{tikzpicture}
}

\begin{center}
\centering
\makebox[\linewidth]{%
    \begin{tikzpicture}[scale=\scaleOfPictures]
        \node (Q) [AInstance] at (0, 0) {\addingPPPxIVxQUESTION};
        \node (B) [AResult, left=1cm of Q] {\addingPPPxIVxNO};
        \node (D) [AResult, right=1cm of Q] {\addingPPPxIVxB};
        \node (E) [AResult, below=0.8cm of Q] {\addingPPPxIVxA};
        \node[anchor=north west, inner xsep=0.1cm] at (B.north west) (name) {$G_1=$};
        \node[anchor=north east, inner xsep=0] at (D.south east) (ok1) {\CaseOk case 1.};
        \draw[CaseArrow]
        (Q)
        -- node[midway, auto,     , inner sep=0cm] {\CaseEdgeDoesNotExist}
        (B);
        \draw[CaseArrow]
        (Q)
        -- node[midway, auto, swap, inner sep=0cm] {\CaseBEdge}
        (D);
        \draw[CaseArrow]
        (Q)
        -- node[midway, auto, swap, inner sep=0cm] {\CaseAEdge}
        (E);
    \end{tikzpicture}
}
\end{center}
Here, if we got a $\b$-edge, we finish with case 2., having found the graph $\heartsuit$.
If the edge was not present, we immediately get a graph $G_1$ that later will help us to finish
the proof.
If in turn an $\a$-edge appeared, we need to perform one additional amalgamation in order
to get the same $G_1$.

\Instance{\ref{lemma:P2plusP3}.6} (building $G_1$)
Now, we pair the previous result with the path $P$:

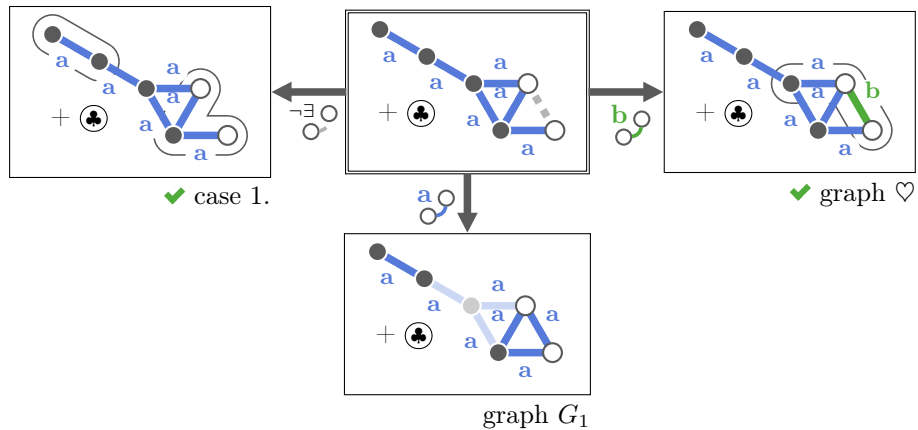
\newcommand{\addingPPPxVxQUESTION}{
\begin{tikzpicture}[scale=0.9*\scaleOfPictures,rotate=-30]
    \begin{pgfonlayer}{foreground}
        \node[ Node] at (-1, 0) (n0) {};
        \node[ Node] at (0, 0) (n1) {};
        \node[ Node] at (1, 0) (n2) {};
        \node[ Node] at (1.866, -0.5) (n3) {};
        \node[ANode] at (2.732, 0) (n4) {};
        \node[ANode] at (1.866, 0.5) (up) {};
        \node[Basic] at (0.2, -1.1) {+ $\kbb$};
    \end{pgfonlayer}
    \begin{pgfonlayer}{background}
        \lEdge[cA,swap]{up}{\a}{n3}
        \lEdge[AEdge,swap]{n4}{}{up}
        \lEdge[cA,swap]{n0}{\a}{n1}
        \lEdge[cA,swap]{n1}{\a}{n2}
        \lEdge[cA,swap]{n2}{\a}{n3}
        \lEdge[cA,swap]{n3}{\a}{n4}
        \lEdge[cA]{n2}{\a}{up}
    \end{pgfonlayer}
\end{tikzpicture}
}
\newcommand{\addingPPPxVxA}{
\begin{tikzpicture}[scale=0.9*\scaleOfPictures,rotate=-30]
    \begin{pgfonlayer}{foreground}
        \node[ Node] at (-1, 0) (n0) {};
        \node[ Node] at (0, 0) (n1) {};
        \node[ Node, fill=cGray!30] at (1, 0) (n2) {};
        \node[ Node] at (1.866, -0.5) (n3) {};
        \node[ANode] at (2.732, 0) (n4) {};
        \node[ANode] at (1.866, 0.5) (up) {};
        \node[Basic] at (0.2, -1.1) {+ $\kbb$};
    \end{pgfonlayer}
    \begin{pgfonlayer}{background}
        \lEdge[cA,swap]{up}{\a}{n3}
        \lEdge[cA,swap]{n4}{\a}{up}
        \lEdge[cA,swap]{n0}{\a}{n1}
        \lEdge[cA,swap, opacity=0.3]{n1}{\a}{n2}
        \lEdge[cA,swap, opacity=0.3]{n2}{\a}{n3}
        \lEdge[cA,swap]{n3}{\a}{n4}
        \lEdge[cA, opacity=0.3]{n2}{\a}{up}
    \end{pgfonlayer}
\end{tikzpicture}
}
\newcommand{\addingPPPxVxB}{
\begin{tikzpicture}[scale=0.9*\scaleOfPictures,rotate=-30]
    \begin{pgfonlayer}{foreground}
        \node[ Node] at (-1, 0) (n0) {};
        \node[ Node] at (0, 0) (n1) {};
        \node[ Node] at (1, 0) (n2) {};
        \node[ Node] at (1.866, -0.5) (n3) {};
        \node[ANode] at (2.732, 0) (n4) {};
        \node[ANode] at (1.866, 0.5) (up) {};
        \node[Basic] at (0.2, -1.1) {+ $\kbb$};
    \end{pgfonlayer}
    \begin{pgfonlayer}{background}
        \lEdge[cA,swap]{up}{\a}{n3}
        \lEdge[cB,swap]{n4}{\b}{up}
        \lEdge[cA,swap]{n0}{\a}{n1}
        \lEdge[cA,swap]{n1}{\a}{n2}
        \lEdge[cA,swap]{n2}{\a}{n3}
        \lEdge[cA,swap]{n3}{\a}{n4}
        \lEdge[cA]{n2}{\a}{up}
    \end{pgfonlayer}
    \begin{pgfonlayer}{prebackground}
        \draw[Selection] (n2.center) -- (up.center) -- (n4.center);
    \end{pgfonlayer}
\end{tikzpicture}
}
\newcommand{\addingPPPxVxNO}{
\begin{tikzpicture}[scale=0.9*\scaleOfPictures,rotate=-30]
    \begin{pgfonlayer}{foreground}
        \node[ Node] at (-1, 0) (n0) {};
        \node[ Node] at (0, 0) (n1) {};
        \node[ Node] at (1, 0) (n2) {};
        \node[ Node] at (1.866, -0.5) (n3) {};
        \node[ANode] at (2.732, 0) (n4) {};
        \node[ANode] at (1.866, 0.5) (up) {};
        \node[Basic] at (0.2, -1.1) {+ $\kbb$};
    \end{pgfonlayer}
    \begin{pgfonlayer}{background}
        \lEdge[cA,swap]{up}{\a}{n3}
        \lEdge[cA,swap]{n0}{\a}{n1}
        \lEdge[cA,swap]{n1}{\a}{n2}
        \lEdge[cA,swap]{n2}{\a}{n3}
        \lEdge[cA,swap]{n3}{\a}{n4}
        \lEdge[cA]{n2}{\a}{up}
    \end{pgfonlayer}
    \begin{pgfonlayer}{prebackground}
        \draw[Selection] (up.center) -- (n3.center) -- (n4.center) (n0.center) -- (n1.center);
    \end{pgfonlayer}
\end{tikzpicture}
}

\begin{center}
    \begin{tikzpicture}[scale=\scaleOfPictures]
        \node (Q) [AInstance, inner ysep=0.1cm] at (0, 0) {\addingPPPxVxQUESTION};
        \node (B) [AResult, inner ysep=0.1cm, left=1cm of Q] {\addingPPPxVxNO};
        \node (D) [AResult, inner ysep=0.1cm, right=1cm of Q] {\addingPPPxVxB};
        \node (E) [AResult, inner ysep=0.1cm, below=0.8cm of Q] {\addingPPPxVxA};
        \node[anchor=north east, inner xsep=0] at (B.south east) (ok1) {\CaseOk case 1.};
        \node[anchor=north east, inner xsep=0] at (D.south east) (ok2) {\CaseOk graph $\heartsuit$};
        \node[anchor=north east, inner xsep=0] at (E.south east) (ok3) {graph $G_1$};
        \draw[CaseArrow]
        (Q)
        -- node[midway, auto,     , inner sep=0cm] {\CaseEdgeDoesNotExist}
        (B);
        \draw[CaseArrow]
        (Q)
        -- node[midway, auto, swap, inner sep=0cm] {\CaseBEdge}
        (D);
        \draw[CaseArrow]
        (Q)
        -- node[midway, auto, swap, inner sep=0cm] {\CaseAEdge}
        (E);
    \end{tikzpicture}
\end{center}
In two out of three possible cases we finish immediately, while in the third one the expected graph
$G_1$ appears as a subgraph.

\Instance{\ref{lemma:P2plusP3}.7}
Using the graph $G_1$ and the result of instance~\ref{lemma:P2plusP3}.4., we perform the last amalgamation in the proof of this lemma,
thus providing the final missing link needed to finalize the proof of Corollary~\ref{corollary:canAddP2xA}.

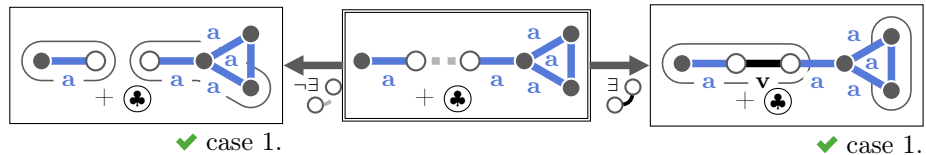
\newcommand{\addingPPPxVIxQUESTION}{
\begin{tikzpicture}[scale=0.9*\scaleOfPictures]
    \begin{pgfonlayer}{foreground}
        \node[ Node] at (-2, 0) (nn) {};
        \node[ANode] at (-1, 0) (n0) {};
        \node[ANode] at (0, 0) (n1) {};
        \node[ Node] at (1, 0) (n2) {};
        \node[ Node] at (1.866, -0.5) (n3) {};
        \node[ Node] at (1.866, 0.5) (up) {};
        \node[Basic] at (-0.5, -0.7) {+ $\kbb$};
    \end{pgfonlayer}
    \begin{pgfonlayer}{background}
        \lEdge[cA,swap]{up}{\a}{n3}
        \lEdge[cA,swap]{nn}{\a}{n0}
        \lEdge[AEdge,swap]{n0}{}{n1}
        \lEdge[cA,swap]{n1}{\a}{n2}
        \lEdge[cA,swap]{n2}{\a}{n3}
        \lEdge[cA]{n2}{\a}{up}
    \end{pgfonlayer}
\end{tikzpicture}
}
\newcommand{\addingPPPxVIxX}{
\begin{tikzpicture}[scale=0.9*\scaleOfPictures]
    \begin{pgfonlayer}{foreground}
        \node[ Node] at (-2, 0) (nn) {};
        \node[ANode] at (-1, 0) (n0) {};
        \node[ANode] at (0, 0) (n1) {};
        \node[ Node] at (1, 0) (n2) {};
        \node[ Node] at (1.866, -0.5) (n3) {};
        \node[ Node] at (1.866, 0.5) (up) {};
        \node[Basic] at (-0.5, -0.7) {+ $\kbb$};
    \end{pgfonlayer}
    \begin{pgfonlayer}{background}
        \lEdge[cA,swap]{up}{\a}{n3}
        \lEdge[cA,swap]{nn}{\a}{n0}
        \lEdge[cBlack,swap]{n0}{\v}{n1}
        \lEdge[cA,swap]{n1}{\a}{n2}
        \lEdge[cA,swap]{n2}{\a}{n3}
        \lEdge[cA]{n2}{\a}{up}
    \end{pgfonlayer}
    \begin{pgfonlayer}{prebackground}
        \draw[Selection] (nn.center) -- (n0.center) -- (n1.center) (n3.center) -- (up.center);
    \end{pgfonlayer}
\end{tikzpicture}
}
\newcommand{\addingPPPxVIxNO}{
\begin{tikzpicture}[scale=0.9*\scaleOfPictures]
    \begin{pgfonlayer}{foreground}
        \node[ Node] at (-2, 0) (nn) {};
        \node[ANode] at (-1, 0) (n0) {};
        \node[ANode] at (0, 0) (n1) {};
        \node[ Node] at (1, 0) (n2) {};
        \node[ Node] at (1.866, -0.5) (n3) {};
        \node[ Node] at (1.866, 0.5) (up) {};
        \node[Basic] at (-0.5, -0.7) {+ $\kbb$};
    \end{pgfonlayer}
    \begin{pgfonlayer}{background}
        \lEdge[cA,swap]{up}{\a}{n3}
        \lEdge[cA,swap]{nn}{\a}{n0}
        \lEdge[cA,swap]{n1}{\a}{n2}
        \lEdge[cA,swap]{n2}{\a}{n3}
        \lEdge[cA]{n2}{\a}{up}
    \end{pgfonlayer}
    \begin{pgfonlayer}{prebackground}
        \draw[Selection] (nn.center) -- (n0.center) (n1.center) -- (n2.center) -- (n3.center);
    \end{pgfonlayer}
\end{tikzpicture}
}

\begin{center}
    \begin{tikzpicture}[scale=\scaleOfPictures]
        \node (Q) [AInstance, inner sep=0.1cm] at (0, 0) {\addingPPPxVIxQUESTION};
        \node (B) [AResult, inner sep=0.15cm, left=0.8cm of Q] {\addingPPPxVIxNO};
        \node (D) [AResult, inner sep=0.15cm, right=0.8cm of Q] {\addingPPPxVIxX};
        \node[anchor=north east, inner xsep=0] at (B.south east) (ok1) {\CaseOk case 1.};
        \node[anchor=north east, inner xsep=0] at (D.south east) (ok2) {\CaseOk case 1.};
        \draw[CaseArrow]
        (Q)
        -- node[midway, auto,     , inner sep=0cm] {\CaseEdgeDoesNotExist}
        (B);
        \draw[CaseArrow]
        (Q)
        -- node[midway, auto, swap, inner sep=0cm] {\CaseEdgeExists}
        (D);
    \end{tikzpicture}
\end{center}
No matter if the edge appeared or not, case 1. of the lemma gets fulfilled, what finishes the proof.
\end{proof}


\subsection{Adding paths of length 2}\label{subsec:addingPaths}

In the previous part of the proof we had a 'resource' $\kbb$ of isolated vertices
and we could use them as needed to construct successive instances of amalgamation.
From now on --- thanks to Corollary~\ref{corollary:canAddP2xA}. --- we may afford to
maintain an arbitrarily large collection of edges $\underline{\a}$ ($\a \in \ColorsTwo$).

The aim of the next four amalgamations will be to show, that
we actually can afford even more --- a collection of 2-edge paths of the form $\underline{\a\x}$ (for some $\x \in
\ColorsTwo$). It is the last step we need to make before showing the ultimate goal of this branch of the proof
 --- deriving the existence of arbitrarily long $\a\b$-paths in $\G$.

Let us formalize the lemma we intend to prove:

\begin{lemma}\label{lemma:addingPPP}
If a strongly homogeneous, 2-edge-colored graph $\G$ satisfies Corollary~\ref{corollary:canAddP2xA}, i.e.,
for every $k \in
\mathbb{N}$ there exist colors $\a, \x \in \ColorsTwo$ such that $\G$ embeds the graph $\underline{\a\x}
+ k\cdot\underline{\a}$, then 
for every $n\in\mathbb{N}$ there exist colors $\a, \y\in\ColorsTwo$ such that
$\G$ embeds the graph 
\[
n\cdot\underline{\a\y}\  \trianglelefteq\  \G
\]
\end{lemma}

\begin{proof}  
As in the previous part, the proof will be inductive. This time,
aiming to find $n\cdot\underline{\a\y} \trianglelefteq \G$ (for some $\y \in \ColorsTwo$), we will
produce successively all the graphs bellow:
\begin{align*}
    &\underline{\a\x_0}\ \ +\ \ (2n\hspace{0.6cm}) \cdot M \cdot \underline{\a}\\
    &\underline{\a\x_1}\ \ +\ \ (2n-1) \cdot M \cdot \underline{\a}\ \ +\ \ \underline{\a\y_1}\\
    &\underline{\a\x_2}\ \ +\ \ (2n-2) \cdot M \cdot \underline{\a}\ \ +\ \ \underline{\a\y_1}+\underline{\a\y_2}\\
    &\underline{\a\x_3}\ \ +\ \ (2n-3) \cdot M \cdot \underline{\a}\ \ +\ \ \underline{\a\y_1}+\underline{\a\y_2}+\underline{\a\y_3}\\
    &\dots\\
    &\underline{\a\x_{2n}}\hspace{0.09cm}+\ \ (\rlap{$0$}\phantom{2}\hspace{0.896cm}) \cdot M \cdot \underline{\a}\ \ +\ \ \underline{\a\y_1}+\underline{\a\y_2}+\underline{\a\y_3}+\dots+\underline{\a\y_{2n}}
\end{align*}
At each point, to produce one isolated path $\underline{\a\y_i}$ we will have to get some constant
number $M \in \mathbb{N}$ of isolated edges $\underline{\a}$ from our 'resource'.
After completing $2n$ steps, among the resulting paths $\underline{\a\y_\bullet}$, by pigeonhole principle,
there exists a subset of $n$ paths all colored the same way. This will finish the proof.

\newcommand{\kbbb}{\circledb{\scriptsize\ensuremath{\blacklozenge}}}

Similarly as before, to hide the unnecessary details, we will use the symbol $\kbbb$ for
the frequently appearing graphs of the form  $\alpha \cdot \underline{\a}\ \ +\ \ \underline{\a\y_1}+\dots+\underline{\a\y_i}$
 --- they are almost passive in the steps of the coming proof.
It is enough to remember, that each time we need a new isolated edge $\underline{\a}$,
we take it from $\kbbb$. Moreover, after each inductive step we
add to $\kbbb$ a new isolated path $\underline{\a\y_i}$.

\paragraph{Inductive step}
From the assumptions we have the graph $\underline{\a\u} + \kbbb$ (for some $\a, \u \in \ColorsTwo$),
and this time our goal is to prove that $\G$ embeds a graph $\underline{\a\v} + \underline{\a\w} +
\kbbb$. As we have already mentioned, we only have to consider four instances of amalgamation.

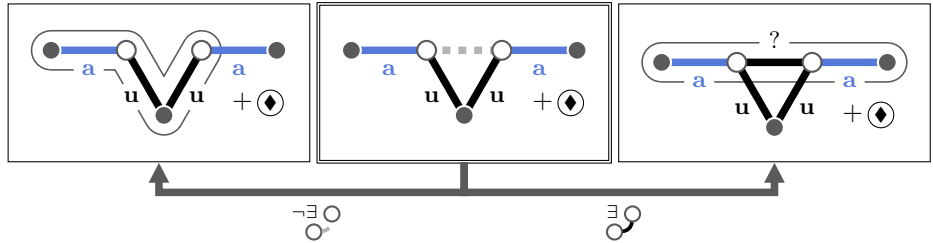
\Instance{\ref{lemma:addingPPP}.1}
\newcommand{\IIIxQ}[0]{
\begin{tikzpicture}[yscale=-1, xscale=1]
    \node[Basic] at (2.75, 0.7) (kb) {$+\,\kbbb$};
    \node[ Node] at (0, 0) (n1) {};
    \node[ANode] at (1, 0) (n2) {};
    \node[ Node] at (1.5, 0.866) (n3) {};
    \node[ANode] at (2, 0) (n4) {};
    \node[ Node] at (3, 0) (n5) {};
    \begin{pgfonlayer}{background}
        \lEdge[cA, swap]{n1}{\a}{n2}
        \lEdge[cBlack, swap]{n2}{\u}{n3}
        \lEdge[cBlack, swap]{n3}{\u}{n4}
        \lEdge[cA, swap]{n4}{\a}{n5}
        \lEdge[AEdge]{n2}{}{n4}
    \end{pgfonlayer}
\end{tikzpicture}
}

\newcommand{\IIIxX}{
\begin{tikzpicture}[yscale=-1, xscale=1]
    \node[Basic] at (2.75, 0.7) (kb) {$+\,\kbbb$};
    \node[ Node] at (0, 0) (n1) {};
    \node[ANode] at (1, 0) (n2) {};
    \node[ Node] at (1.5, 0.866) (n3) {};
    \node[ANode] at (2, 0) (n4) {};
    \node[ Node] at (3, 0) (n5) {};
    \begin{pgfonlayer}{background}
        \lEdge[cA, swap]{n1}{\a}{n2}
        \lEdge[cBlack, swap]{n2}{\u}{n3}
        \lEdge[cBlack, swap]{n3}{\u}{n4}
        \lEdge[cA, swap]{n4}{\a}{n5}
        \lEdge[cBlack]{n2}{\n}{n4}
    \end{pgfonlayer}
    \begin{pgfonlayer}{prebackground}
        \draw[Selection]
        (n1.center) to (n2.center) to (n4.center) to (n5.center);
    \end{pgfonlayer}
\end{tikzpicture}
}

\newcommand{\IIIxNO}{
\begin{tikzpicture}[yscale=-1, xscale=1]
    \node[Basic] at (2.75, 0.7) (kb) {$+\,\kbbb$};
    \node[ Node] at (0, 0) (n1) {};
    \node[ANode] at (1, 0) (n2) {};
    \node[ Node] at (1.5, 0.866) (n3) {};
    \node[ANode] at (2, 0) (n4) {};
    \node[ Node] at (3, 0) (n5) {};
    \begin{pgfonlayer}{background}
        \lEdge[cA, swap]{n1}{\a}{n2}
        \lEdge[cBlack, swap]{n2}{\u}{n3}
        \lEdge[cBlack, swap]{n3}{\u}{n4}
        \lEdge[cA, swap]{n4}{\a}{n5}
    \end{pgfonlayer}

    \begin{pgfonlayer}{prebackground}
        \draw[Selection]
        (n1.center) to (n2.center) to (n3.center) to (n4.center);
    \end{pgfonlayer}
\end{tikzpicture}
}

\begin{center}
    \begin{tikzpicture}[scale=\scaleOfPictures]
        \node (Q) [AInstance, minimum height=2.1cm] at (0, 0) {\IIIxQ};
        \node (B) [AResult, minimum height=2.1cm, left=0.1cm of Q] {\IIIxNO};
        \node (D) [AResult, minimum height=2.1cm, right=0.1cm of Q] {\IIIxX};
        \coordinate[below=0.4cm of Q] (QQ);
        \draw[CaseArrow]
        (Q) -- (QQ)
        -- node[midway, auto,     , inner sep=0cm] {\CaseEdgeDoesNotExist}
        (QQ -| B) -- (B);
        \draw[CaseArrow]
        (Q) -- (QQ)
        -- node[midway, auto, swap, inner sep=0cm] {\CaseEdgeExists}
        (QQ -| D) -- (D);
    \end{tikzpicture}
\end{center}

\noindent
No matter what the result will be, we will get the following path:
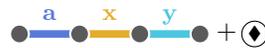
\begin{center}
    \begin{tikzpicture}[scale=\scaleOfPictures]
        \node[ Node] (n1) at (0, 0) {};
        \node[ Node] (n2) at (1, 0) {};
        \node[ Node] (n3) at (2, 0) {};
        \node[ Node] (n4) at (3, 0) {};
        \node[Basic] (kb) at (3.7, 0) {$+\,\kbbb$};
        \begin{pgfonlayer}{background}
            \lEdge[cA]{n1}{\a}{n2}
            \lEdge[cX]{n2}{\x}{n3}
            \lEdge[cY]{n3}{\y}{n4}
        \end{pgfonlayer}
    \end{tikzpicture}
\end{center}
for some $\x, \y \in \ColorsTwo$.

\Instance{\ref{lemma:addingPPP}.2}
Using it (with an additional edge $\underline{\a}$ taken from $\kbbb$), we build the
following instance:

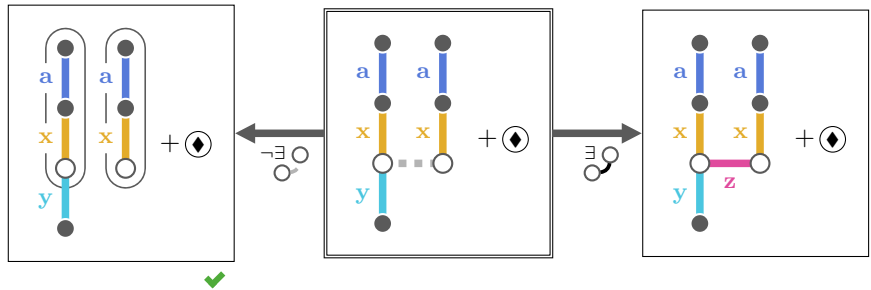
\newcommand{\addingPPPPxIIxQ}{
\begin{tikzpicture}[yscale=-1*\scaleOfPictures, xscale=\scaleOfPictures]
    \node[ Node] at (-0.5, -1) (n1) {};
    \node[ Node] at (-0.5,  0) (n2) {};
    \node[ANode] at (-0.5,  1) (n3) {};
    \node[ANode] at ( 0.5,  1) (n4) {};
    \node[ Node] at ( 0.5,  0) (n5) {};
    \node[ Node] at ( 0.5, -1) (n6) {};
    \node[ Node] at (-0.5,  2) (d1) {};
    \node[Basic] (kb) at (1.5, 0.6) {$+\,\kbbb$};
    \begin{pgfonlayer}{background}
        \lEdge[cA, swap]{n1}{\a}{n2}
        \lEdge[cX, swap]{n2}{\x}{n3}
        \lEdge[cY, swap]{n3}{\y}{d1}
        \lEdge[cX]{n4}{\x}{n5}
        \lEdge[cA]{n5}{\a}{n6}
        \lEdge[AEdge]{n3}{}{n4}
    \end{pgfonlayer}
\end{tikzpicture}
}
\newcommand{\addingPPPPxIIxX}{
\begin{tikzpicture}[yscale=-1*\scaleOfPictures, xscale=\scaleOfPictures]
    \node[ Node] at (-0.5, -1) (n1) {};
    \node[ Node] at (-0.5,  0) (n2) {};
    \node[ANode] at (-0.5,  1) (n3) {};
    \node[ANode] at ( 0.5,  1) (n4) {};
    \node[ Node] at ( 0.5,  0) (n5) {};
    \node[ Node] at ( 0.5, -1) (n6) {};
    \node[ Node] at (-0.5,  2) (d1) {};
    \node[Basic] (kb) at (1.5, 0.6) {$+\,\kbbb$};
    \begin{pgfonlayer}{background}
        \lEdge[cA,swap]{n1}{\a}{n2}
        \lEdge[cX,swap]{n2}{\x}{n3}
        \lEdge[cY,swap]{n3}{\y}{d1}
        \lEdge[cX]{n4}{\x}{n5}
        \lEdge[cA]{n5}{\a}{n6}
        \lEdge[cZ,swap]{n3}{\z}{n4}
    \end{pgfonlayer}
\end{tikzpicture}
}
\newcommand{\addingPPPPxIIxNO}{
\begin{tikzpicture}[yscale=-1*\scaleOfPictures, xscale=\scaleOfPictures]
    \node[ Node] at (-0.5, -1) (n1) {};
    \node[ Node] at (-0.5,  0) (n2) {};
    \node[ANode] at (-0.5,  1) (n3) {};
    \node[ANode] at ( 0.5,  1) (n4) {};
    \node[ Node] at ( 0.5,  0) (n5) {};
    \node[ Node] at ( 0.5, -1) (n6) {};
    \node[ Node] at (-0.5,  2) (d1) {};
    \node[Basic] (kb) at (1.5, 0.6) {$+\,\kbbb$};
    \begin{pgfonlayer}{background}
        \lEdge[cA,swap]{n1}{\a}{n2}
        \lEdge[cX,swap]{n2}{\x}{n3}
        \lEdge[cY,swap]{n3}{\y}{d1}
        \lEdge[cX]{n4}{\x}{n5}
        \lEdge[cA]{n5}{\a}{n6}
    \end{pgfonlayer}
    \begin{pgfonlayer}{prebackground}
        \draw[Selection] (n1.center) -- (n2.center) -- (n3.center) (n4.center) -- (n5.center) -- (n6.center);
    \end{pgfonlayer}
\end{tikzpicture}
}

\begin{center}
    \begin{tikzpicture}[scale=\scaleOfPictures]
        \node (Q) [AInstance] at (0, 0) {\addingPPPPxIIxQ};
        \node (B) [AResult, left=1.2cm of Q] {\addingPPPPxIIxNO};
        \node (D) [AResult, right=1.2cm of Q] {\addingPPPPxIIxX};
        \node[anchor=north east, inner xsep=0] at (B.south east) (ok1) {\CaseOk};
        \draw[CaseArrow]
        (Q)
        -- node[midway, auto,     , inner sep=0cm] {\CaseEdgeDoesNotExist}
        (B);
        \draw[CaseArrow]
        (Q)
        -- node[midway, auto, swap, inner sep=0cm] {\CaseEdgeExists}
        (D);
    \end{tikzpicture}
\end{center}
If the edge is not present in the solution, we readily get two disjoint paths of length 2.
Suppose then that some $\z$-edge appeared ($\z \in \ColorsTwo$).
If $\z = \y$, we move on straight to the instance 4. If in turn $\z \neq \y$, an additional step in
necessary:

\Instance{\ref{lemma:addingPPP}.3}
Once more we get one edge from $\kbbb$ and build an instance similar to the previous one ---
the only difference is the new edge colored with $\z \neq \y$.

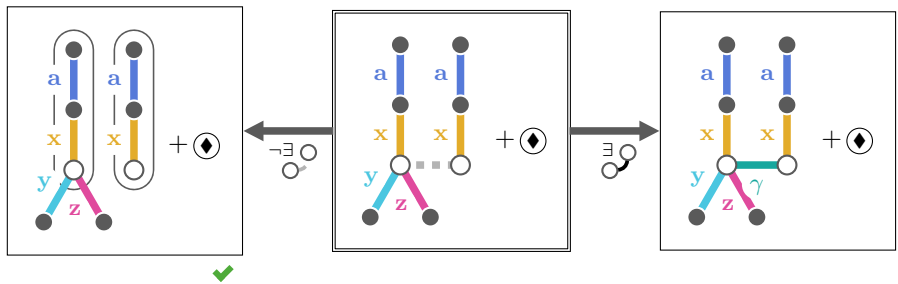
\newcommand{\addingPPPPxIIIxQ}{
\begin{tikzpicture}[yscale=-1*\scaleOfPictures, xscale=\scaleOfPictures]
    \node[ Node] at (-0.5, -1) (n1) {};
    \node[ Node] at (-0.5,  0) (n2) {};
    \node[ANode] at (-0.5,  1) (n3) {};
    \node[ANode] at ( 0.5,  1) (n4) {};
    \node[ Node] at ( 0.5,  0) (n5) {};
    \node[ Node] at ( 0.5, -1) (n6) {};
    \node[ Node] at ( 0  ,  1.866) (d1) {};
    \node[Basic] (kb) at (1.5, 0.6) {$+\,\kbbb$};
    \node[ Node] at (-1,  1.866) (left) {};
    \begin{pgfonlayer}{background}
        \lEdge[cA, swap]{n1}{\a}{n2}
        \lEdge[cX, swap]{n2}{\x}{n3}
        \lEdge[cY, swap]{n3}{\y}{left}
        \lEdge[cZ, swap]{n3}{\z}{d1}
        \lEdge[cX]{n4}{\x}{n5}
        \lEdge[cA]{n5}{\a}{n6}
        \lEdge[AEdge]{n3}{}{n4}
    \end{pgfonlayer}
\end{tikzpicture}
}
\newcommand{\addingPPPPxIIIxX}{
\begin{tikzpicture}[yscale=-1*\scaleOfPictures, xscale=\scaleOfPictures]
    \node[ Node] at (-0.5, -1) (n1) {};
    \node[ Node] at (-0.5,  0) (n2) {};
    \node[ANode] at (-0.5,  1) (n3) {};
    \node[ANode] at ( 0.5,  1) (n4) {};
    \node[ Node] at ( 0.5,  0) (n5) {};
    \node[ Node] at ( 0.5, -1) (n6) {};
    \node[ Node] at ( 0  ,  1.866) (d1) {};
    \node[Basic] (kb) at (1.5, 0.6) {$+\,\kbbb$};
    \node[ Node] at (-1,  1.866) (left) {};
    \begin{pgfonlayer}{background}
        \lEdge[cA, swap]{n1}{\a}{n2}
        \lEdge[cX, swap]{n2}{\x}{n3}
        \lEdge[cY, swap]{n3}{\y}{left}
        \lEdge[cZ, swap]{n3}{\z}{d1}
        \lEdge[cX]{n4}{\x}{n5}
        \lEdge[cA]{n5}{\a}{n6}
        \lEdge[cZZ, swap]{n3}{$\gamma$}{n4}
    \end{pgfonlayer}
\end{tikzpicture}
}
\newcommand{\addingPPPPxIIIxNO}{
\begin{tikzpicture}[yscale=-1*\scaleOfPictures, xscale=\scaleOfPictures]
    \node[ Node] at (-0.5, -1) (n1) {};
    \node[ Node] at (-0.5,  0) (n2) {};
    \node[ANode] at (-0.5,  1) (n3) {};
    \node[ANode] at ( 0.5,  1) (n4) {};
    \node[ Node] at ( 0.5,  0) (n5) {};
    \node[ Node] at ( 0.5, -1) (n6) {};
    \node[ Node] at ( 0  ,  1.866) (d1) {};
    \node[Basic] (kb) at (1.5, 0.6) {$+\,\kbbb$};
    \node[ Node] at (-1,  1.866) (left) {};
    \begin{pgfonlayer}{background}
        \lEdge[cA, swap]{n1}{\a}{n2}
        \lEdge[cX, swap]{n2}{\x}{n3}
        \lEdge[cY, swap]{n3}{\y}{left}
        \lEdge[cZ, swap]{n3}{\z}{d1}
        \lEdge[cX]{n4}{\x}{n5}
        \lEdge[cA]{n5}{\a}{n6}
    \end{pgfonlayer}
    \begin{pgfonlayer}{prebackground}
        \draw[Selection] (n1.center) -- (n2.center) -- (n3.center) (n4.center) -- (n5.center) -- (n6.center);
    \end{pgfonlayer}
\end{tikzpicture}
}

\begin{center}
    \begin{tikzpicture}[scale=\scaleOfPictures]
        \node (Q) [AInstance] at (0, 0) {\addingPPPPxIIIxQ};
        \node (B) [AResult, left=1.2cm of Q] {\addingPPPPxIIIxNO};
        \node (D) [AResult, right=1.2cm of Q] {\addingPPPPxIIIxX};
        \node[anchor=north east, inner xsep=0] at (B.south east) (ok1) {\CaseOk};
        \draw[CaseArrow]
        (Q)
        -- node[midway, auto,     , inner sep=0cm] {\CaseEdgeDoesNotExist}
        (B);
        \draw[CaseArrow]
        (Q)
        -- node[midway, auto, swap, inner sep=0cm] {\CaseEdgeExists}
        (D);
    \end{tikzpicture}
\end{center}
If we do not obtain an edge, we end having -- as before -- two disjoint paths.
When some edge exists, we are sure that its color $\Color{cZZ}{\gamma} \in \ColorsTwo$
is either equal $\y$ or $\z$, since $\y \neq \z$ and $\ColorsTwo$ has only two elements.
W.l.o.g. let us assume, that $\Color{cZZ}{\gamma} = \y$.

Then we have, as a result of Instance~\ref{lemma:addingPPP}.2 or~\ref{lemma:addingPPP}.3,
a graph of the form:

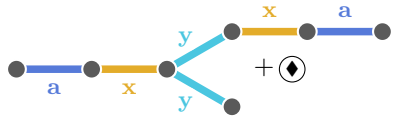
\begin{center}
    \begin{tikzpicture}[rotate=-90]
        \node[ Node] at (-0.5, -1) (n1) {};
        \node[ Node] at (-0.5,  0) (n2) {};
        \node[ Node] at (-0.5,  1) (n3) {};
        \node[ Node] at ( 0  ,  1.866) (d1) {};
        \node[ Node] at ( -1  ,  2.866) (d2) {};
        \node[ Node] at ( -1  ,  3.866) (d3) {};
        \node[Basic] (kb) at (-0.5, 2.5) {$+\,\kbbb$};
        \node[ Node] at (-1,  1.866) (left) {};
        \begin{pgfonlayer}{background}
            \lEdge[cA, swap]{n1}{\a}{n2}
            \lEdge[cX, swap]{n2}{\x}{n3}
            \lEdge[cY]{n3}{\y}{left}
            \lEdge[cY, swap]{n3}{\y}{d1}
            \lEdge[cX]{left}{\x}{d2}
            \lEdge[cA]{d2}{\a}{d3}
        \end{pgfonlayer}
    \end{tikzpicture}
\end{center}
Using it we may create the last amalgamation instance and finalize the proof.

\Instance{\ref{lemma:addingPPP}.4}

\newcommand{\addingPPPPxIVxQ}{
\begin{tikzpicture}[yscale=-1*\scaleOfPictures, xscale=\scaleOfPictures]
    \node[Basic] (kb) at (1, -0.7) {$+\,\kbbb$};
    \draw (0,0) node[Node] (n1) {}
    ++(1,0) node[Node] (n2) {}
    ++(1,0) node[Node] (n3) {}
    ++(0.866,-0.5) node[Node] (u1) {}
    ++(1,0) node[ANode] (u2) {}
    ++(1,0) node[Node] (u3) {}
    ++(0,1) node[Node] (d3) {}
    ++(-1,0) node[ANode] (d2) {}
    ++(-1,0) node[Node] (d1) {};
    \begin{pgfonlayer}{background}
        \lEdge[cA]{n1}{\a}{n2}
        \lEdge[cX]{n2}{$x$}{n3}
        \lEdge[cY]{n3}{\y}{u1}
        \lEdge[cY,swap]{n3}{\y}{d1}
        \lEdge[cX]{u1}{\x}{u2}
        \lEdge[cX,swap]{d1}{\x}{d2}
        \lEdge[cA]{u2}{\a}{u3}
        \lEdge[cA,swap]{d2}{\a}{d3}
        \lEdge[AEdge]{u2}{}{d2}
    \end{pgfonlayer}
\end{tikzpicture}
}

\newcommand{\addingPPPPxIVxX}{
\begin{tikzpicture}[yscale=-1*\scaleOfPictures, xscale=\scaleOfPictures]
    \node[Basic] (kb) at (1, -0.7) {$+\,\kbbb$};
    \draw (0,0) node[Node] (n1) {}
    ++(1,0) node[Node] (n2) {}
    ++(1,0) node[Node] (n3) {}
    ++(0.866,-0.5) node[Node] (u1) {}
    ++(1,0) node[ANode] (u2) {}
    ++(1,0) node[Node] (u3) {}
    ++(0,1) node[Node] (d3) {}
    ++(-1,0) node[ANode] (d2) {}
    ++(-1,0) node[Node] (d1) {};
    \begin{pgfonlayer}{background}
        \lEdge[cA]{n1}{\a}{n2}
        \lEdge[cX]{n2}{$x$}{n3}
        \lEdge[cY]{n3}{\y}{u1}
        \lEdge[cY,swap]{n3}{\y}{d1}
        \lEdge[cX]{u1}{\x}{u2}
        \lEdge[cX,swap]{d1}{\x}{d2}
        \lEdge[cA]{u2}{\a}{u3}
        \lEdge[cA,swap]{d2}{\a}{d3}
        \lEdge[cBlack, swap]{u2}{$\alpha$}{d2}
    \end{pgfonlayer}
    \begin{pgfonlayer}{prebackground}
        \draw[Selection]
        (u3.center) -- (u2.center) -- (d2.center)
        (n1.center) -- (n2.center) -- (n3.center);
    \end{pgfonlayer}
\end{tikzpicture}
}

\newcommand{\addingPPPPxIVxNO}{
\begin{tikzpicture}[yscale=-1*\scaleOfPictures, xscale=\scaleOfPictures]
    \node[Basic] (kb) at (1, -0.7) {$+\,\kbbb$};
    \draw (0,0) node[Node] (n1) {}
    ++(1,0) node[Node] (n2) {}
    ++(1,0) node[Node] (n3) {}
    ++(0.866,-0.5) node[Node] (u1) {}
    ++(1,0) node[ANode] (u2) {}
    ++(1,0) node[Node] (u3) {}
    ++(0,1) node[Node] (d3) {}
    ++(-1,0) node[ANode] (d2) {}
    ++(-1,0) node[Node] (d1) {};
    \begin{pgfonlayer}{background}
        \lEdge[cA]{n1}{\a}{n2}
        \lEdge[cX]{n2}{$x$}{n3}
        \lEdge[cY]{n3}{\y}{u1}
        \lEdge[cY,swap]{n3}{\y}{d1}
        \lEdge[cX]{u1}{\x}{u2}
        \lEdge[cX,swap]{d1}{\x}{d2}
        \lEdge[cA]{u2}{\a}{u3}
        \lEdge[cA,swap]{d2}{\a}{d3}
    \end{pgfonlayer}
    \begin{pgfonlayer}{prebackground}
        \draw[Selection]
        (u3.center) -- (u2.center) -- (u1.center)
        (d3.center) -- (d2.center) -- (d1.center);
    \end{pgfonlayer}
\end{tikzpicture}
}

\begin{center}
    \begin{tikzpicture}[scale=\scaleOfPictures]
        \node (Q) [AInstance] at (0, 0) {\addingPPPPxIVxQ};
        \node (B) [AResult, below left=1.5cm and -2cm of Q] {\addingPPPPxIVxNO};
        \node (D) [AResult, below right=1.5cm and -2cm of Q] {\addingPPPPxIVxX};
        \node[anchor=north east, inner xsep=0] at (B.south east) (ok1) {\CaseOk};
        \node[anchor=north east, inner xsep=0] at (D.south east) (ok2) {\CaseOk};
        \draw[CaseArrow]
        (Q)
        -- node[midway, auto,     , inner sep=-0.1cm, pos=0.4] {\CaseEdgeDoesNotExist}
        (B);
        \draw[CaseArrow]
        (Q)
        -- node[midway, auto, swap, inner sep=-0.1cm, pos=0.4] {\CaseEdgeExists}
        (D);
    \end{tikzpicture}
\end{center}
Independently form the existence of an edge, in the result we may fine a subgraph of
the following form:
$$
\underline{\a\v} + \underline{\a\w} + \kbbb
$$
Its presence ends the proof of the lemma.
\end{proof}


\subsection{Producing arbitrarily long paths}\label{subsec:longPaths}

There is the last thing to do in case A) --- showing that in $\G$ arbitrarily long paths exist. It is formalized by the following lemma:

\begin{lemma}\label{lemma:longPaths}
    If a homogeneous, 2-edge-colored graph $\G$ satisfies Lemma~\ref{lemma:addingPPP}, i.e.,
    for every $n\in\mathbb{N}$ there exist colors $\a,\x\in\ColorsTwo$ such that 
    $\G$ embeds the graph $n\cdot\underline{\a\x}\  \trianglelefteq\  \G$,
    then $\G$ also embeds an arbitrarily long $\a\b$-path ($\b$ being the second element of $\ColorsTwo$).
\end{lemma}

\begin{proof} 
Once more we conduct an induction.
    \paragraph{Inductive step}
    Here we will be showing how from shorter paths we may produce longer ones:
    Assuming that we have a graph that is a sum of paths of length
    $d \in \mathbb{N}, d \geq 2$., we will build (using amalgamation) a graph
    that is a sum of (fewer) paths of length $2d-1$.

    More precisely, we will show how from the graph $2k \cdot P$ (where $k\in\mathbb{N}$ and $P$ is some $\a\b$-path of
    length $d$) we can derive in $k$ steps a graph $Q_1 + Q_2 + \dots + Q_k$, where $Q_\bullet$ are
    $\a\b$-paths of length $2d-1$. Taking sufficiently large $k$ we will ensure, that among those
    $k$ paths (by pigeonhole principle) there will be a group of size $n$ of equally colored ones.

    The outline of the procedure is as follows:
    \begin{align*}
        &(k\phantom{-1}) \cdot P\\
        &(k-1) \cdot P\ \ +\ \ Q_1\\
        &(k-2) \cdot P\ \ +\ \ Q_1+Q_2\\
        &(k-3) \cdot P\ \ +\ \ Q_1+Q_2+Q_3\\
        &\dots\\
        &(\rlap{$0$}\hspace{0.896cm}) \cdot P\ \ +\ \ Q_1+Q_2+Q_3+\dots+Q_k
    \end{align*}

    \noindent
    To show a single step, one amalgamation will be enough:

    \Instance{\ref{lemma:longPaths}.1}
    \newcommand{\longPathsxIxQ}
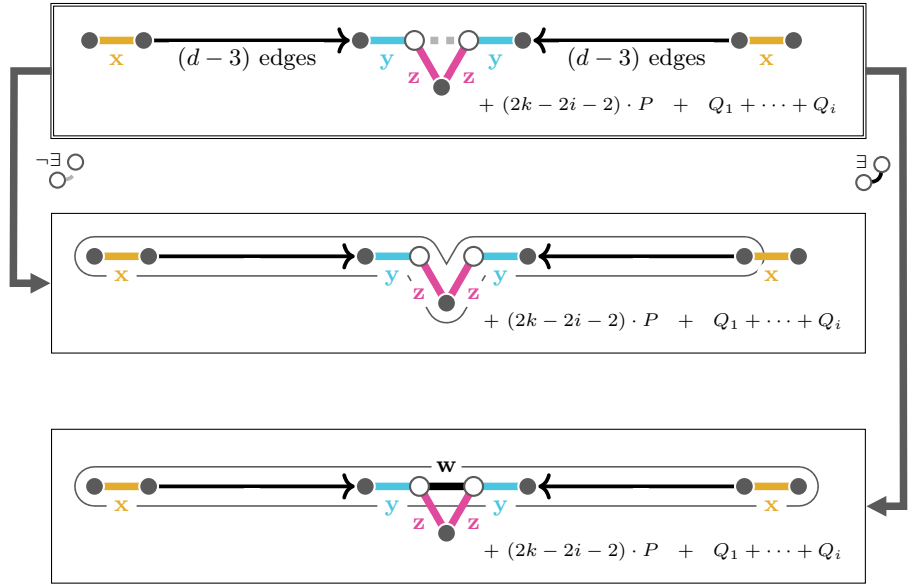
{
    \begin{tikzpicture}[scale=0.9*\scaleOfPictures]
        \begin{pgfonlayer}{foreground}
            \node (n0) [ Node] at (-2, 0) {};
            \node (n1) [ Node] at (-1, 0) {};
            \node (n2) [ Node] at (3, 0) {};
            \node (n3) [ANode] at (4, 0) {};
            \node (nc) [ Node] at (4.5, -0.866) {};
            \node (n4) [ANode] at (5, 0) {};
            \node (n5) [ Node] at (6, 0) {};
            \node (n6) [ Node] at (10, 0) {};
            \node (n7) [ Node] at (11, 0) {};
            \node (kb) [Basic] at (8.5, -1.2) {\scriptsize$+\ (2k-2i-2)\cdot P \ \ +\ \  Q_1 + \dots + Q_i$};
        \end{pgfonlayer}
        \begin{pgfonlayer}{background}
            \draw[ThinEdge, ->]
            (n1.center) to
                node[auto, midway, fill=white, minimum width=0, inner xsep=0.1cm, inner ysep=0, swap] {{\strut\small$(d-3)$ edges}}
            (n2);
            \draw[ThinEdge, ->]
            (n6.center) to
                node[auto, midway, fill=white, minimum width=0, inner xsep=0.1cm, inner ysep=0, ] {{\strut\small$(d-3)$ edges}}
            (n5);
            \lEdge[AEdge]{n3}{}{n4}
            \lEdge[cX,swap]{n0}{\x}{n1}
            \lEdge[cY,swap]{n2}{\y}{n3}
            \lEdge[cZ,swap]{n3}{\z}{nc}
            \lEdge[cZ,swap]{nc}{\z}{n4}
            \lEdge[cY,swap]{n4}{\y}{n5}
            \lEdge[cX,swap]{n6}{\x}{n7}
        \end{pgfonlayer}
    \end{tikzpicture}
    }
    \newcommand{\longPathsxIxQX}{
    \begin{tikzpicture}[scale=0.9*\scaleOfPictures]
        \begin{pgfonlayer}{foreground}
            \node (n0) [ Node] at (-2, 0) {};
            \node (n1) [ Node] at (-1, 0) {};
            \node (n2) [ Node] at (3, 0) {};
            \node (n3) [ANode] at (4, 0) {};
            \node (nc) [ Node] at (4.5, -0.866) {};
            \node (n4) [ANode] at (5, 0) {};
            \node (n5) [ Node] at (6, 0) {};
            \node (n6) [ Node] at (10, 0) {};
            \node (n7) [ Node] at (11, 0) {};
            \node (kb) [Basic] at (8.5, -1.2) {\scriptsize$+\ (2k-2i-2)\cdot P \ \ +\ \  Q_1 + \dots + Q_i$};
        \end{pgfonlayer}
        \begin{pgfonlayer}{background}
            \draw[ThinEdge, ->]
            (n1.center) to
            node[auto, midway, fill=white, minimum width=0, inner xsep=0.1cm, inner ysep=0, swap] {}
            (n2);
            \draw[ThinEdge, ->]
            (n6.center) to
            node[auto, midway, fill=white, minimum width=0, inner xsep=0.1cm, inner ysep=0, ] {}
            (n5);
            \lEdge[cBlack]{n3}{\w}{n4}
            \lEdge[cX,swap]{n0}{\x}{n1}
            \lEdge[cY,swap]{n2}{\y}{n3}
            \lEdge[cZ,swap]{n3}{\z}{nc}
            \lEdge[cZ,swap]{nc}{\z}{n4}
            \lEdge[cY,swap]{n4}{\y}{n5}
            \lEdge[cX,swap]{n6}{\x}{n7}
        \end{pgfonlayer}
        \begin{pgfonlayer}{prebackground}
            \draw[Selection] (n0.center) -- (n1.center) -- (n2.center) -- (n3.center) -- (n4.center) -- (n5.center) -- (n6.center) -- (n7.center);
        \end{pgfonlayer}
    \end{tikzpicture}
    }
    \newcommand{\longPathsxIxQNO}{
    \begin{tikzpicture}[scale=0.9*\scaleOfPictures]
        \begin{pgfonlayer}{foreground}
            \node (n0) [ Node] at (-2, 0) {};
            \node (n1) [ Node] at (-1, 0) {};
            \node (n2) [ Node] at (3, 0) {};
            \node (n3) [ANode] at (4, 0) {};
            \node (nc) [ Node] at (4.5, -0.866) {};
            \node (n4) [ANode] at (5, 0) {};
            \node (n5) [ Node] at (6, 0) {};
            \node (n6) [ Node] at (10, 0) {};
            \node (n7) [ Node] at (11, 0) {};
            \node (kb) [Basic] at (8.5, -1.2) {\scriptsize$+\ (2k-2i-2)\cdot P \ \ +\ \  Q_1 + \dots + Q_i$};
        \end{pgfonlayer}
        \begin{pgfonlayer}{background}
            \draw[ThinEdge, ->]
            (n1.center) to
            node[auto, midway, fill=white, minimum width=0, inner xsep=0.1cm, inner ysep=0, swap] {}
            (n2);
            \draw[ThinEdge, ->]
            (n6.center) to
            node[auto, midway, fill=white, minimum width=0, inner xsep=0.1cm, inner ysep=0, ] {}
            (n5);
            \lEdge[cX,swap]{n0}{\x}{n1}
            \lEdge[cY,swap]{n2}{\y}{n3}
            \lEdge[cZ,swap]{n3}{\z}{nc}
            \lEdge[cZ,swap]{nc}{\z}{n4}
            \lEdge[cY,swap]{n4}{\y}{n5}
            \lEdge[cX,swap]{n6}{\x}{n7}
        \end{pgfonlayer}
        \begin{pgfonlayer}{prebackground}
            \draw[Selection] (n0.center) -- (n1.center) -- (n2.center) -- (n3.center) -- (nc.center) -- (n4.center) -- (n5.center) -- (n6.center);
        \end{pgfonlayer}
    \end{tikzpicture}
    }

    \begin{center}
        \begin{tikzpicture}[scale=\scaleOfPictures]
            \node (Q) [AInstance, minimum width=13.5cm*\scaleOfPictures] at (0, 0) {\longPathsxIxQ};
            \node (B) [AResult, minimum width=13.5cm*\scaleOfPictures, below=1cm of Q] {\longPathsxIxQNO};
            \node (D) [AResult, minimum width=13.5cm*\scaleOfPictures, below=1cm of B] {\longPathsxIxQX};
            \coordinate[left=0.5cm of Q.west] (pB);
            \coordinate[left=0.5cm of B.west] (ppB);
            \coordinate[right=0.5cm of Q.east] (pD);
            \coordinate[right=0.5cm of D.east] (ppD);
            \draw[CaseArrow]
                (Q) to
                (pB) to node[pos=0.45, auto, inner sep=0cm] {\CaseEdgeDoesNotExist}
                (ppB) to
                (B);
            \draw[CaseArrow]
                (Q) to
                (pD) to node[pos=0.25*0.9, auto, swap, inner sep=0cm] {\CaseEdgeExists}
                (ppD) to
                (D);
        \end{tikzpicture}
    \end{center}
    In both cases we get the path of length $2d - 1$ we wanted. It should be noted here, that
    the construction required a pair of equally colored paths.

    Repeating the above amalgamation $k$ times (according to the previously mentioned outline)
    we get a collection of $k$ disjoint paths, each of length $2d - 1$.
    They are not necessarily painted the same way, but fixing some sufficiently large $k$ (e.g. $k = 2^{d-1} \cdot n$ surely would do),
    we may choose a subset of $n$ same-looking paths. We are allowed to do that, since by
    assumption for each $k\in \mathbb{N}$ we can produce a graph $k \cdot P$
    (for some $\a\b$-path $P$ of length $d$). This ends the proof of the inductive step.

    Because at the very beginning we can choose an arbitrarily large sum of equal paths $\a\x$, then
    using the inductive step repeatedly we will be proving the possibility of
    producing collections of paths of increasing lengths:
    $$
        2 \xrightarrow{\hspace{0.8cm}} (2\cdot2) - 1 = 3 \xrightarrow{\hspace{0.8cm}} 5 \xrightarrow{\hspace{0.8cm}} 9 \xrightarrow{\hspace{0.8cm}} 17 \xrightarrow{\hspace{0.8cm}} \dots
    $$
    This observation completes our proof.
\end{proof}

\subsubsection{Summary}\label{subsec:branchAsummary}
    We made our way to the end of Section~\ref{sec:coreA}. The chain of lemmas that were
    stated has its beginning at the case A) of Lemma~\ref{lemma:withBranches}.
    As we move along this chain, we show the possibility of adding to the initial graphs respectively:
    \begin{itemize}
        \item first, an arbitrary number of isolated vertices,
        \item then, isolated edges,
        \item next, 2-edge paths,
        \item and finally, arbitrarily long paths.
    \end{itemize}
    At the end of the chain, we have obtained the second case of Theorem~\ref{thm:core},
    so we may at last consider the case A) as resolved.

\vspace{3mm}

\para{Acknowledgements}
We are grateful to the anonymous referees for valuable comments.

\bibliographystyle{abbrv}
\bibliography{bib}

\end{document}